\documentclass[11pt,letterpaper]{article}
\usepackage{unicode-math}
\usepackage{amsmath,amsthm,nicefrac}
\usepackage{amstext}
\usepackage{libertine}

\usepackage{subfigure}
\usepackage{hhline}
\usepackage[table,dvipsnames]{xcolor}

\usepackage{typearea}
\typearea{14}
\usepackage[font={small,it}]{caption}

\usepackage{setspace}
\usepackage[compact]{titlesec}

\usepackage[shortlabels]{enumitem}
\usepackage[breaklinks]{hyperref}
\hypersetup{colorlinks=true,%
            citebordercolor={.6 .6 .6},linkbordercolor={.6 .6 .6},%
citecolor=blue,urlcolor=black,linkcolor=blue}

\usepackage[nameinlink]{cleveref}
\Crefname{algocf}{Algorithm}{Algorithms}
\crefname{algocfline}{line}{lines}
\Crefname{invariant}{Invariant}{Invariants}
\Crefname{claim}{Claim}{Claims}
\Crefname{subclaim}{Subclaim}{Subclaims}

\makeatletter
\AddToHook{cmd/appendix/before}{\def\cref@section@alias{appendix}\def\cref@subsection@alias{appendix}}
\makeatother

\usepackage{epsfig}
\usepackage{amsthm}

\usepackage{mathrsfs}
\usepackage{xspace}
\usepackage{soul}
\usepackage{latexsym}
\usepackage{bbm}

\usepackage{mathtools}
\usepackage{microtype}

\usepackage{accents}

\usepackage{framed}

\definecolor{DarkGray}{rgb}{0.66, 0.66, 0.66}
\definecolor{DarkPowderBlue}{rgb}{0.0, 0.2, 0.6}
\definecolor{fluorescentyellow}{rgb}{0.8, 1.0, 0.0}
\definecolor{cerulean}{rgb}{0.0, 0.48, 0.65}
\definecolor{bleudefrance}{rgb}{0.19, 0.55, 0.91}

\usepackage[ruled,vlined,linesnumbered,algonl]{algorithm2e}
\SetEndCharOfAlgoLine{}
\SetKwComment{Comment}{\footnotesize$\triangleright$\ }{}

\SetCommentSty{mycommfont}

\usepackage{thmtools,thm-restate}

\usepackage{mathtools}
\usepackage{bm}

\usepackage{xspace}
\usepackage{textgreek}

\usepackage{nicefrac}
\usepackage{typearea}
\typearea{14}
\usepackage[font={small,it}]{caption}

\allowdisplaybreaks

\newtheorem{theorem}{Theorem}[section]
\newtheorem{lemma}[theorem]{Lemma}
\newtheorem{observation}[theorem]{Observation}
\newtheorem{claim}[theorem]{Claim}

\newtheorem{corollary}[theorem]{Corollary}

\theoremstyle{definition}

\theoremstyle{remark}

\newcommand{\ronmfl}{\textsf{RO}\xspace}
\newcommand{\fnmfl}{\textsf{FR}\xspace}
\newcommand{\rfnmfl}{\overline{\fnmfl}}

\newcommand{\rocip}{\textsf{RO}\xspace}
\newcommand{\fcip}{\textsf{FR}\xspace}
\newcommand{\rfcip}{\overline{\fcip}}

\DeclareMathOperator{\BACKUP}{\textsc{Backup}}

\newcommand{\rmfl}{\mathsf{OFF}\xspace}

\newcommand\unc{U\xspace}
\newcommand\backup{\mathrm{backup}\xspace}
\newcommand\crpt{\varsigma\xspace}

\newcommand\csamples{\samples'}

\newcommand{\family}{\mathcal{F}}
\newcommand{\univ}{X}
\newcommand{\cover}{\mathrm{cov}}
\newcommand{\cost}{\mathrm{cost}}
\newcommand{\samples}{S}
\newcommand{\fosc}{\textsf{FR}\xspace}
\newcommand{\rfosc}{\overline{\fosc}}
\newcommand{\omfl}{\textsf{OL}\xspace}
\newcommand{\onst}{\textsf{OL}\xspace}
\newcommand{\offst}{\textsf{OFF}\xspace}
\newcommand{\rosc}{\textsf{RO}\xspace}
\newcommand{\bk}{\textsf{bk}\xspace}
\newcommand\skewed{\mathrm{rare}}
\newcommand\clean{\mathrm{clean}\xspace}
\newcommand\Clean{\samples_\mathrm{clean}\xspace}
\newcommand{\event}{\mathcal{E}}

\newcommand\backargs[2]{\backup\left(#1 \ \middle| \ #2\right)\xspace}

\DeclareMathOperator{\OPTLP}{\mathrm{OPT}_{LP}}
\DeclareMathOperator{\OPT}{\mathrm{OPT}}
\DeclareMathOperator{\SOL}{\mathrm{SOL}}

\DeclarePairedDelimiter\ceil{\lceil}{\rceil}

\DeclareMathOperator{\ONLINEMFL}{\mathsf{OL}}

\DeclareMathOperator{\SOLST}{\mathrm{SOL}}

\newcommand\objhat{\widehat{\mathsf{Obj}}\xspace}
\newcommand\obj{\mathsf{Obj}\xspace}

\DeclarePairedDelimiter\floor{\lfloor}{\rfloor}
\DeclarePairedDelimiter\abs{\lvert}{\rvert}%
\DeclarePairedDelimiter\inp{\langle}{\rangle}
\newcommand{\E}{\ensuremath{\mathbb{E}}}

\newcommand{\pr}{\ensuremath{\mathbb{P}}}
\newcommand{\Var}{\mathrm{Var}}

\newcommand{\ind}{\mathbbm{1}}

\usepackage{graphicx}

\DeclareMathOperator*{\expectation}{\mathbb{E}}

\newcommand{\prob}{\Pr\probarg}
\DeclarePairedDelimiterX{\probarg}[1]{[}{]}{%
	\ifnum\currentgrouptype=16 \else\begingroup\fi
	\activatebar#1
	\ifnum\currentgrouptype=16 \else\endgroup\fi
}

\newcommand{\expect}{\expectation\expectarg}
\DeclarePairedDelimiterX{\expectarg}[1]{[}{]}{%
	\ifnum\currentgrouptype=16 \else\begingroup\fi
	\activatebar#1
	\ifnum\currentgrouptype=16 \else\endgroup\fi
}

\newcommand{\expectover}[1]{\expectation_{#1}\expectarg}

\DeclarePairedDelimiterX{\wklarg}[1]{(}{)}{%
	\ifnum\currentgrouptype=16 \else\begingroup\fi
	\activatebars#1
	\ifnum\currentgrouptype=16 \else\endgroup\fi
}

\newcommand{\innermid}{\nonscript\;\delimsize\vert\nonscript\;}
\newcommand{\activatebar}{%
	\begingroup\lccode`\~=`\|
	\lowercase{\endgroup\let~}\innermid 
	\mathcode`|=\string"8000
}

\newcommand{\innermids}{\nonscript\;\delimsize\vert\delimsize\vert\nonscript\;}
\newcommand{\activatebars}{%
	\begingroup\lccode`\~=`\|
	\lowercase{\endgroup\let~}\innermids 
	\mathcode`|=\string"8000
}

\newcommand{\eat}[1]{}

\newcommand{\cC}{\mathcal{C}}

\newcommand{\cA}{\mathcal{A}}
\newcommand{\cB}{\mathcal{B}}

\newcommand{\argmin}{{\rm argmin}}

\usepackage{graphicx}

\newcommand{\oas}{{\sc oas}\xspace}

\title{Online Algorithms with a Sample:\\ Tight Bounds and Adversarial Robustness}

\author{
Anish Hebbar\thanks{Duke University, Durham  NC, USA emails: \texttt{anishshripad.hebbar@duke.edu}, \, \texttt{debmalya@cs.duke.edu}.}
\and
Ravi Kumar\thanks{Google Research, Mountain View CA, USA,   email: \texttt{ravi.k53@gmail.com}.}
\and
Roie Levin\thanks{Rutgers University, Piscataway NJ, USA, email: \texttt{roie.levin@rutgers.edu}.}
\and
Joseph (Seffi) Naor\thanks{Technion -- Israel Institute of Technology, Haifa, Israel, email: \texttt{naor@cs.technion.ac.il}.}
\and
Debmalya Panigrahi$^*$
}

\date{}

\begin{document}

\maketitle

\begin{abstract}
Suppose an online algorithm is given an unbiased $p$-sample of its input as offline advice; can the algorithm exploit the sample to achieve beyond-worst-case performance? We study this \emph{online algorithms with a sample} (\oas) model. We show a tight $O\left(\log \nicefrac{1}{p} \cdot \log m  + \log n\right)$-competitive algorithm for set cover, exponentially improving upon the $O\left(\nicefrac{1}{p} \cdot \log (mn)\right)$ guarantee of Gupta et al. (SODA'24) and answering an open question therein. Our techniques extend to covering integer programs and non-metric facility location, also yielding tight bounds for these problems. Further, we give an $O(\log \nicefrac{1}{p}/ \log \log \nicefrac{1}{p})$-competitive algorithm for metric facility location, answering an open question of Argue et al. (NeurIPS'22).

We then introduce and study the robust variant of the \oas model, in which an adversary is allowed to arbitrarily modify $k$ elements of the $p$-sample. For set cover, covering integer programs, and non-metric facility location, we obtain a tight competitive ratio of $O\left(\log \nicefrac{k}{p} \cdot \log m + \log n\right)$. For metric facility location and Steiner tree, we obtain tight competitive ratios of $O\left(\log \nicefrac{k}{p} / \log \log \nicefrac{k}{p} \right)$ and $O\left(\log \nicefrac{k}{p}\right)$ respectively. To the best of our knowledge, these are the first results for robust algorithms in the {\oas} setting.
\end{abstract}

\pagenumbering{gobble}

\clearpage

\pagenumbering{arabic}

\section{Introduction}\label{sec:introduction}Recent years have witnessed growing interest in designing algorithms that remain reliable in the presence of imperfect or corrupted data. This question lies at the heart of robust statistics, where one seeks estimators that tolerate adversarial contamination, and has become increasingly important in modern  systems, where data may be noisy, manipulated, or otherwise unreliable. Similar robustness questions have also been studied in combinatorial optimization, including semi-random models, stochastic optimization, sublinear algorithms, streaming and sketching, and property testing, where algorithms must extract useful information from incomplete, noisy, or adversarially perturbed observations. Together, these developments have established robustness to imperfect data as a central algorithmic paradigm across statistics and combinatorial optimization.

In this paper, we introduce and study robustness in online algorithms that leverage sampled offline information. In the \emph{online algorithms with a sample} (\oas) model, the algorithm is provided with an unbiased $p$-sample of the input, for some parameter $p\in[0,1]$, before the online phase begins. For example, for set cover, this corresponds to the algorithm being provided $pn$ elements sampled uniformly at random before the remaining elements arrive in an arbitrary (possibly adversarial) order.
The goal is to exploit this sample to circumvent lower bounds in the classical worst-case online setting. The \oas model was introduced by Kaplan~et al.~\cite{KaplanNR20,KaplanNR22} for the secretary problem and online weighted matching. (Earlier, Kumar et al.~\cite{KumarPSSV19} studied a related semi-online model for bipartite matching, which was subsequently adopted by Lattanzi et al.~\cite{LattanziMVWZ21} for correlation clustering.) Argue et al.~\cite{ArgueFGS22} later extended the \oas model to a broad range of minimization problems, including clustering, network design, and scheduling, obtaining improved competitive ratios for problems such as Steiner tree, metric facility location, and load balancing.  Recently, Gupta et al.~\cite{GuptaKL24} studied the \oas model for online covering problems, including set cover and covering integer programs.

Despite this rapid progress, nearly all existing work assumes that the sample is generated perfectly from the input. This assumption stands in sharp contrast to the broader literature on robust algorithm design, where adversarial contamination is treated as a first-order concern. Indeed, in practice, samples may be corrupted by faulty data collection, adversarial manipulation, or mis-specified sampling procedures.  Motivated by these issues, we ask: \textsl{can online algorithms be made provably robust to adversarial sample corruption}?

Orthogonal to robustness, our understanding of the algorithmic power of an offline sample remains incomplete even for the non-robust setting. Existing upper and lower bounds exhibit substantial gaps for several central problems, leaving open the precise dependence of the competitive ratio on the sampling probability $p$. Can we close these gaps while simultaneously obtaining tight guarantees and robustness?

\subsection{Our Results}

\paragraph{Tight bounds.}
Our first set of results establishes tight competitive ratios for several fundamental problems in the \oas model.

For set cover, Gupta et al.~\cite{GuptaKL24} gave an $O\left(\nicefrac{1}{p} \cdot \log (mn)\right)$-competitive algorithm and posed the following open question: 
\textsl{``We submit as an interesting open problem the task of determining the tight dependence on $p$ \ldots. We conjecture that it should be $O(\log(mn) \log(1/p))$.''}

We resolve this question in the affirmative and give an $O\left(\log \nicefrac{1}{p} \cdot \log m + \log n\right)$-competitive algorithm for set cover, and this bound is tight for polynomial-time algorithms. We also extend this result to more general covering IPs and non-metric facility location, yielding tight bounds for these problems as well. Here $m$ is the number of sets/columns/facilities, and $n$ is the number of elements/rows/clients.

\begin{restatable}[Set Cover, CIP, and Non-Metric Facility Location]{theorem}{CoveringMain}
\label{thm:cov-main}
There is a randomized $O\left(\log \nicefrac{1}{p} \cdot \log m + \log n\right)$-competitive algorithm
for set cover in the \oas model with a $p$-sample. Furthermore, the same competitive ratio can be achieved
for online covering integer programs (CIPs) and online non-metric facility location.
\end{restatable}

Next, we consider metric facility location. Argue et al.~\cite{ArgueFGS22} claimed a nearly tight $O\left(\log \nicefrac{1}{p}\right)$-competitive algorithm for the uniform cost case. We identify a gap in the original analysis (see \Cref{sec:mfl-error}) and provide a new proof establishing a slightly improved (and tight) guarantee for the same algorithm, which also generalizes to the nonuniform cost case.  En route, we also answer an open question in their work: \textsl{``We leave it as an open problem to extend our result to give an $O\left(\frac{\log \nicefrac{1}{p}}{\log \log \nicefrac{1}{p}}\right)$ algorithm.''}
\begin{restatable}[Metric Facility Location]{theorem}{MetricFLMain}
\label{thm:mfl-main}
There is an $O\!\left(\frac{\log \nicefrac{1}{p}}{\log \log \nicefrac{1}{p}}\right)$-competitive randomized algorithm for metric facility location in the \oas model with a $p$-sample.
\end{restatable}
Leveraging a recent reduction of \cite{gao26}, as an immediate corollary we obtain competitive algorithms for the stochastic versions of the problems above in which the input is sampled according to a Markov Random Field (MRF) as opposed to a product distribution.  Let $\Delta$
be the weighted degree of the MRF, which is a natural parameter measuring the maximum correlation strength of the distribution it
represents.%
\footnote{When $\Delta = 0$, the MRF represents a product distribution, and as $\Delta \rightarrow \infty$, the MRF can represent joint distributions with arbitrary correlations.  We refer the reader to \cite{gao26} for a more complete treatment of this emergent branch of online algorithms.}
\begin{corollary}[Applying {\cite[Theorem 3.5]{gao26}} to \Cref{thm:cov-main,thm:mfl-main}]
    There exist $O(\Delta \cdot \log m + \log n)$-competitive algorithms for online set cover, CIP, and non-metric facility location, and an $O(\Delta)$-competitive algorithm for online metric facility location in the MRF arrival model.  Furthermore, these bounds are tight.
\end{corollary}

\paragraph{Adversarial Robustness.}
Our second contribution is to introduce and study adversarial robustness in the \oas\ framework. We consider a natural contamination model, inspired by robust statistics, in which the observed sample may differ from a true unbiased $p$-sample in up to $k$ positions; we call this
the \emph{$(p, k)$-robust \oas} model. 
We seek online algorithms whose competitive ratios degrade gracefully with the level of contamination.

One consequence of sample contamination is that an adversary can effectively ``hide'' $O(k/p)$ elements by replacing sampled elements with dummy elements. (Strictly speaking, the adversary can replace only $k$ sampled elements, but there is constant probability that at most $k$ elements from a designated set of size $O(k/p)$ are sampled, allowing the adversary to hide the entire sampled portion of that set.) This observation effectively imports worst-case online lower bounds with $n=O(k/p)$ into the $(p,k)$-robust \oas\ setting.

Our results show that these lower bounds are essentially tight for all the problems considered above, as well as for Steiner tree (for which a tight bound of $O\left(\log \nicefrac{1}{p}\right)$ was already known without sample contamination~\cite{ArgueFGS22}).

\begin{restatable}[Robust Covering Problems]{theorem}{RobustCovering}
\label{thm:robust-cov}
There is a randomized $O\left(\log \nicefrac{k}{p} \cdot \log m + \log n\right)$-competitive algorithm for set cover in the $(p,k)$-robust \oas\ model. Furthermore, the same competitive ratio can be achieved for covering integer programs and non-metric facility location.
\end{restatable}
\begin{restatable}[Robust Metric Facility Location]{theorem}{RobustMetricFacilityLocation}
There is a randomized $O\left(\frac{\log \nicefrac{k}{p}}{\log\log \nicefrac{k}{p}}\right)$-competitive algorithm for metric facility location in the $(p,k)$-robust \oas\ model.
\end{restatable}
\begin{restatable}[Robust Steiner Tree]{theorem}{RobustSteinerTree}
There is a deterministic $O\left(\log \nicefrac{k}{p}\right)$-competitive algorithm for Steiner tree in the $(p,k)$-robust \oas\ model.
\end{restatable}
\paragraph{Notation.}
For any ordered sequence $z_1, z_2, \ldots$ and a valid index $i$, let $z_{< i}$ denote the prefix $z_1, \ldots, z_{i-1}$.  Throughout, we assume (w.l.o.g.) $p \leq 1/2$ and $k \geq 1$. 

\subsection{Techniques and Overview}

\paragraph{Set Cover in the \oas Model.}
\Cref{sec:sc} contains our results for set cover, and we start with the no-corruptions case in \Cref{sec:psamplesc} to illustrate the core idea. Our jumping-off point is the previous state-of-the-art strategy of \cite{GuptaKL24}, which was to run an approximation algorithm on the sample (specifically, a random order online algorithm\footnote{It is not clear if a random order algorithm is strictly necessary or if a generic approximation algorithm suffices, but random order is convenient for analysis and will be especially so when we move to the with-corruptions case.}) and a na\"ive algorithm (buy the cheapest set covering any still uncovered element, a.k.a.\ the ``backup'' set) on the remainder. They obtain a competitive ratio linear in $\nicefrac{1}{p}$ by charging these backup costs to the cost of the solution for the sample, and this analysis is tight.\footnote{Consider an instance where one set covers the universe, but every element also has many singleton sets containing it. If $p=\nicefrac{1}{n}$ and the sample consists of a single element, then any algorithm will pick a singleton with high probability because it cannot distinguish the sets containing this element. The sum of backup costs of the remaining instance is thus $n$ times the cost to cover the sample.}
To improve this to, e.g., $\log \nicefrac{1}{p}$, one must fall back on a worst-case online algorithm such as \cite{buchbinder2009online} when $p$ is small, but this strategy should somehow interpolate smoothly between \cite{GuptaKL24} and \cite{buchbinder2009online} with respect to $p$.

A natural idea is to warm-start \cite{buchbinder2009online} with information from the sample, because \cite{buchbinder2009online} admits both potential and primal-dual\footnote{Dinitz et al.\ \cite{dinitz21} precisely run a primal-dual algorithm with a dual estimated from samples.} based analyses that seem amenable to a head start argument. Sadly, they are not. Consider an instance consisting only of singleton sets. For any $p \leq \nicefrac{1}{2}$, the only solution for the sample is to buy the relevant singletons, and this does not inform how to solve the remainder of the instance. 
Here the saving grace is that this all-singleton instance is trivial to solve optimally, but more generally, why should it be the case that instances for which the sample is unlikely to be informative must structurally be very ``disjoint'' and easy to approximate? We answer this question via a 3-stage algorithm that precisely shows this learnable-or-easy dichotomy: (1) run a (random order) approximation algorithm on the sample, then (2) pass uncovered elements to \cite{buchbinder2009online}, but \emph{with a budget cap of $O(\log \nicefrac{1}{p})$} (as a function of $p$) on the allowed success amplification, and only then (3) use the na\"ive algorithm to cover any still surviving elements. The heart of the argument is to show that the instance remaining after stages (1) and (2) admits a good approximation by the na\"ive buy-backup-set algorithm.

\paragraph{Set Cover in the Robust \oas Model.}
Next, in \Cref{sec:setcover}, we extend the algorithm to the robust \oas  setting. The algorithm itself changes only slightly: during the sample phase, we additionally enforce a budget constraint. The analysis, however, becomes substantially more delicate because even a single corruption can introduce arbitrary correlations among the sampled elements. For example, an adversary may replace a particular sampled element by a dummy unless the entire sample equals some predetermined set $S$. Conditioned on observing that element in the corrupted sample, the remainder of the sample is then completely determined. Our goal is therefore to extract from the corrupted sample a sufficiently unbiased subsample that still provides useful information for the first stage of the algorithm.

The key insight is to imagine terminating the sample phase immediately before the first corrupted sample appears. Intuitively, this leaves an effectively uniform sample of density $p'=\nicefrac{p}{k}$, suggesting a competitive ratio of $O(\log \nicefrac{1}{p'})=O(\log \nicefrac{k}{p})$. The difficulty is that the adversary chooses which samples to corrupt \emph{after} the sample has been realized, so the resulting prefix need not be uniformly distributed. To overcome this obstacle, we prove a technical statement that we call the \emph{near-uniformity lemma}. Consider the process of repeatedly drawing uniformly random samples without replacement from the uncorrupted elements. We show that the first $O(\nicefrac{np}{k})$ samples obtained in this way are nearly uniformly distributed over all but an expected $O(\nicefrac{k}{p})$ elements of the input. Since the uncorrupted samples that appear before the first corrupted sample are generated by precisely this process, the lemma recovers the intuition above and the analysis goes through. 

We leave the extensions to CIPs and non-metric facility location to \Cref{sec:cip} and \Cref{sec:nmfl} respectively.

\paragraph{Metric Facility Location.}
For the uncorrupted case in \Cref{sec:metricpsample}, we use the same algorithm as \cite{ArgueFGS22}, but we patch the analysis (we outline the issue in \cite{ArgueFGS22} in \Cref{sec:mfl-error}). The idea is to run a constant factor approximation on the sample and then to use Meyerson's algorithm for the online portion. The main fix is to switch from an expectation to a high probability argument: for each cluster in $\OPT$, if we condition on seeing at least half of the expected number of samples from this cluster (a high probability event), we can successfully charge to the cost of connecting this cluster online to the portion of it paid during the offline phase. The bad case in which we see too few samples is sufficiently unlikely that we can afford to suffer the worst case approximation ratio of Meyerson's algorithm for that cluster. To extend this analysis to the robust \oas model, we borrow intuition from our set cover analysis and only use the prefix of the samples before the first corrupted element. Though this subset is not truly uniform, this turns out not to affect the analysis materially. We give details of this analysis in \Cref{sec:metricfacility}.

\paragraph{Steiner Tree.}
In \Cref{sec:steinertree}, we handle the Steiner tree problem. The main idea behind the algorithm of \cite{ArgueFGS22} for uncorrupted samples is that the Eulerian tour of the metric-completion of the terminals, which is a good proxy for the optimal Steiner tree, is a cycle that gets cut into arcs of expected length $\nicefrac1p$ by the sampled terminals. The greedy algorithm, which is $O(\log n)$-competitive in the worst case, morally runs on each such arc disjointly, thus yielding a competitive ratio $O(\log \nicefrac1p)$. Once again, we show that the prefix of roughly $\nicefrac{np}{k}$ uncorrupted samples is sufficient to salvage a competitive ratio of $O(\log \nicefrac{k}{p})$ since it is uniform enough to divide the cycle into roughly $\nicefrac{np}{k}$ equal parts containing approximately $\nicefrac{k}{p}$ terminals each.

\subsection{Related Work}

\textbf{Online Algorithms with Advice.}
The works most directly related to ours are those of Argue~et al.~\cite{ArgueFGS22} and Gupta~et al.~\cite{GuptaKL24}: the former gave $p$-sample algorithms for load balancing, Steiner tree, and metric facility location, the latter gave $p$-sample algorithms for covering problems. Kaplan~et al.~\cite{KaplanNR20,KaplanNR22} studied sample-based models for the secretary problem and online weighted matching, and Correa~et al.~\cite{CorreaCFOT21} considered a closely related independent-sampling variant. Kumar~et al.~\cite{KumarPSSV19} considered a similar semi-online version of bipartite matching where part of the input is available offline and the rest arrives online, and Lattanzi~et al.~\cite{LattanziMVWZ21} later applied this model to correlation clustering. Chen~et al.~\cite{ChenIMXZ24} also studied a related sample-based model for online ranking. Our work is related to learning-augmented online algorithms, where predictions are used to improve online performance~\cite{lykouris2021competitive,mitzenmacherpredictions}: in fact one can view $p$-sample algorithms as playing both the role of the machine learning model trained on historical data and the online algorithm making use of the prediction.

\textbf{Beyond Worst-Case Analysis.}
Sample-augmented algorithms fit into the broader literature on beyond worst-case analysis; see \cite{beyondworstcase} for a comprehensive survey. Most relevant to us are random-order online algorithms for set cover and covering integer programs~\cite{KL21}, non-metric facility location~\cite{GuptaKL24}, and metric facility location~\cite{meyerson2001online,kaplan2023almost}. A nearby but distinct way to relax the worst case is to assume the online input is generated by a distribution, and this model has been studied for Steiner forest, facility location, and vertex cover~\cite{garg2008stochastic}, stochastic set cover~\cite{ggl13}, online stochastic covering and packing~\cite{devanurcovering}, and stochastic submodular cover~\cite{AgarwalAK19}, among many others. Gao~et al.~\cite{gao26} recently connected $p$-sample guarantees to guarantees for stochastic arrivals with graphical dependencies.

\textbf{Adversarial Robustness.} There has recently been an explosion of interest in adversarial robustness for statistics and machine learning. The literature is too vast to cite exhaustively here, but we mention some examples in the contexts of Gaussian mixture models \cite{bakshi20,bakshi26}, mean estimation \cite{hopkins19,kane24}, inference in graphical models \cite{hopkins24}. (See also the surveys \cite{diakonikolas2021robustness,diakonikolas2023algorithmic}.) There is also a line of work, straddling theory and practice, on building deep learning models that are robust to adversarially injected training data; see the influential paper \cite{madry18} and the subsequent tutorial \cite{kolter2018adversarial} for further references.

\section{Online Set Cover}
\label{sec:sc}

Let $(\family, \univ)$ be a set system consisting of $|\family| = m$ sets from a universe of $|\univ| = n$ elements. Each set $F \in \family$ has a nonnegative value $\cost(F)$.  For any collection of sets $\cC \subseteq \family$, let $\cover(\cC) = \cup_{F \in \cC} F$ and $\cost(\cC) = \sum_{F \in \cC} \cost(F)$.  The goal in set cover is to find $\cC \subseteq \family$ such that $\cover(\cC) = \univ$ and $\cost(\cC)$ is minimum. In the \emph{online} version, $\univ$ is initially unknown and its elements are revealed one at a time.  Each element $x \in \univ$ upon arrival, reveals the sets in $\family$ to which it belongs, and if $x$ is not already covered, the algorithm must immediately pick one such set to cover $x$. The goal is to minimize the total cost of the sets chosen by the algorithm.  Let $\OPT(\univ')$ denote the optimum set cover cost of $\univ' \subseteq \univ$; let $\OPT = \OPT(\univ)$.

In the \emph{$p$-sample model} for set cover, the algorithm samples a set $\samples$ uniformly at random from all subsets of size $s = pn$, before the online sequence begins. Then the online input consists of the remaining elements 
$\univ \setminus \samples$ in an arbitrary (and possibly adversarial) order. 

In the \emph{$(p,k)$-adversarially robust model}, the adversary is allowed to modify any $k$ elements in $\samples$ after they are sampled, but before they are considered by the algorithm; we use $\csamples$ to denote the corrupted samples. 

We use the following rounding algorithm~\cite[Section 2.2.2]{buchbinder2009design}:
\begin{lemma}[Set cover rounding]
\label{lem:sc-rounding}
    Let $\cA$ be a fractional online set cover algorithm that takes an  element sequence $x_1, x_2, \ldots$ as input and produces coordinate-wise nondecreasing fractional solutions $y_1 \leq y_2 \leq \cdots$, with competitive ratio $\alpha$. There is an online algorithm $\mathsf{Round}$ that takes $y_1 \leq y_2 \leq \cdots$ as input and a parameter $\beta \in (0,1]$, and produces (not necessarily feasible) integral solutions $\cC_1 \subseteq \cC_2 \subseteq \cdots$ with the following properties: (i) for every step $t$, 
    $\expect*{\cost(\cC_t)} \leq \alpha \log \nicefrac1\beta \cdot \OPT(\{x_1, \ldots, x_t\})$ and (ii) for every $t' \leq t$, 
    $\Pr[x_{t'} \notin \cover(\cC_t)] \leq \beta$.
\end{lemma}
We refer to the (not necessarily feasible) integral online set cover algorithm $\overline \cA = \mathsf{Round}(\cA, \beta)$ as \emph{rounding} algorithm $\cA$ with \emph{boosting parameter} $\log \nicefrac1\beta$.

All our algorithms assume black box access to a fractional online set cover algorithm $\fosc$, which can be rounded online using \Cref{lem:sc-rounding} (yielding $\rfosc$), and a random order set cover algorithm $\rosc$.\footnote{While we only need black box access to $\fosc$ and $\rosc$ to describe the algorithm, we achieve tight bounds by choosing these as the seminal algorithm of \cite{AABN09,buchbinder2009online} and the LearnOrCover algorithm of \cite{KL21}.} 
 
\subsection{$p$-Sample Model}
\label{sec:psamplesc}

We begin with our  algorithm for the no-corruptions case.   For simplicity, for a current solution $\cC$ and an element $x \notin \cover(\cC)$, let
$\backup(x) = \arg\min \{ \cost(F) \mid x \in F, F \in \family \}$, i.e., it is the least cost \emph{backup cover} for an element not covered by the current solution.  

The algorithm has the following phases:
\begin{enumerate}[nosep]
\item In the offline phase, run the $\rosc$ algorithm on the (randomly permuted) sample $\samples$ and let
$\cC_\samples^{\rosc}$ be the solution obtained.  Let $\cC^{\rfosc}_{0} = \cC^\bk_0 = \emptyset$.
\item In the online phase, let $x$ be the element arriving at step $t$:
\begin{enumerate}
\item if $x \in \cover(\cC_\samples^{\rosc})$, then 
$\cC^{\rfosc}_{t} = \cC^{\rfosc}_{t-1}, \cC^{\bk}_{t}= \cC^{\bk}_{t-1}$ and skip. 
\item apply the $\fosc$ algorithm and round the fractional solution online with boosting parameter $\log \nicefrac1p$ (\Cref{lem:sc-rounding}) to obtain a solution $\cC^{\rfosc}_t$. 
\item 
if $x \notin \cover(\cC_\samples^{\rosc} \cup \cC^{\rfosc}_t \cup \cC^{\bk}_{t-1})$, then 
$\cC^\bk_t = \cC^\bk_{t-1} \cup \{\backup(x)\}$,
else $\cC^\bk_{t} = \cC^\bk_{t-1}$.
\end{enumerate}
The solution at this step is given by
$\cC^\rosc_\samples \cup \cC^{\rfosc}_t \cup  \cC^\bk_t$. 
\end{enumerate}
Let $\cC^{\rfosc}$ (resp., $\cC^\bk$) denote the final value of $\cC^{\rfosc}_t$ (resp., $\cC^\bk_t$).   The final solution is then $\cC^\rosc_\samples \cup \cC^{\rfosc} \cup \cC^\bk$.

\subsubsection{Analysis}
To analyze, we need an idea from \cite{GuptaKL24}, which we reproduce here for completeness.  Intuitively, the total backup cost of elements that were not sampled cannot be too large and can be charged to the cost of $\rosc$ on the sample.  Let $\unc^{\rosc} = \{ x \in \univ \setminus \samples \mid x \notin \cover(\cC_\samples^{\rosc}) \}$ be the
elements not covered in the offline phase. 
\begin{lemma}
\label{lem:backup_costs}
It holds that $\expect*{\sum_{x \in \unc^{\rosc}} \cost(\backup(x))} \leq \nicefrac{\Delta}{p} \cdot \OPT$, where $\Delta$ is the competitive ratio of $\rosc$.
\end{lemma}
\begin{proof}
Let $\pi = (\pi_1, \ldots, \pi_n)$ be a uniformly chosen random
permutation of $\univ$ and let 
$\samples$ denote the first $pn$ elements according to $\pi$.  Since the 
$s$ elements of $\samples$ are the input to the
$\rosc$ algorithm, let $\cC^{\rosc}_i$ be its output on the prefix $(\pi_1, \ldots, \pi_i)$;   
from our earlier notation, we have $\cC_\samples^\rosc
= \cC^\rosc_s$.

For simplicity, let $b_{\cC}(x) = \cost(\backup(x))$ if $x \notin \cover(\cC)$ and $0$ otherwise.  
We need to upper bound the quantity 
$\expect*{\sum_{x \in \unc^{\rosc}} \cost(\backup(x))} = \expect*{\sum_{i = s+1}^n b_{\cC^{\rosc}_{s}}(\pi_i)}$.
Note that by definition
$b_{\cC'}(x) \leq b_{\cC}(x)$ if
$\cC' \supseteq \cC$; therefore  
$b_{\cC_{s}^{\rosc}}(\pi_i)
\leq b_{\cC_{j-1}^{\rosc}}(\pi_i)$ for $j \leq s$. 
Averaging this inequality over 
$j = 1, \ldots, s$, taking expectations, and summing over $i = s+1, \ldots, n$, we obtain:
\begin{eqnarray}
\label{eq:31}
\expect*{\sum_{i = s+1}^n b_{\cC_{s}^{\rosc}} (\pi_i)}
& \leq & 
\sum_{i = s+1}^n \frac{1}{s} \cdot
\sum_{j = 1}^s \expect*{b_{\cC_{j-1}^{\rosc}}(\pi_i)}.
\end{eqnarray}
We now claim that 
$\expect*{b_{\cC_{j-1}^{\rosc}}(\pi_i)} = \expect*{b_{\cC_{j-1}^{\rosc}}(\pi_j)}$.  Indeed, conditioned on $(\pi_1, \ldots, \pi_{j-1})$ and the randomness of
$\rosc$, the output $\cC_{j-1}^\rosc$ is fixed and since $\pi_i$ and $\pi_j$ are both uniformly distributed over $\univ \setminus \{\pi_1, \ldots, \pi_{j-1}\}$, the claim follows.   Also, 
$\cost(\cC_{\samples}^{\rosc}) \geq \sum_{j = 1}^s b_{\cC_{j-1}^{\rosc}}(\pi_j)$ and $\rosc$ is $\Delta$-competitive.  Applying these observations to \eqref{eq:31}, we obtain 
\begin{align*}
\expect*{\sum_{i = s+1}^n b_{\cC_{s}^{\rosc}} (\pi_i)}
& \leq 
\frac{n-s}{s} \cdot 
\sum_{j = 1}^s \expect*{b_{\cC_{j-1}^{\rosc}}(\pi_j)}
\leq
\frac{1}{p}\cdot
\expect*{\cost(\cC_{\samples}^{\rosc})} 
\leq \frac{\Delta}{p} \cdot \expect*{\OPT(\samples)}
\leq \frac{\Delta}{p} \cdot \OPT.
\qedhere
\end{align*}
\end{proof}

\begin{theorem}
\label{thm:oxtsc_poly}
If $\rosc$ is a random order set cover algorithm with competitive ratio $\Delta$ and $\fosc$ is a fractional online set cover algorithm with competitive ratio $\alpha$, then the algorithm above has expected competitive ratio $2\Delta + \alpha \log \nicefrac{1}{p}$. 
\end{theorem}
\begin{proof}
We bound the expected cost
$\expect*{\cost(\cC^\rosc_\samples) + \cost(\cC^{\rfosc}) + \cost(\cC^\bk)}$ arising from the different phases of the algorithm.  
The expected cost from phase (1) is given by $\expect*{\cost(\cC_\samples^\rosc)} 
\leq \Delta \cdot \expect*{\OPT(\samples)} \leq \Delta \cdot \OPT$.  The expected cost from phase 2(b) is determined by the $\alpha$-competitive guarantee of $\fosc$ and \Cref{lem:sc-rounding}, yielding
a bound of $\expect*{\cost(\cC^{\rfosc})} \leq \alpha \log \nicefrac 1p \cdot \OPT$. 

To understand the cost of phase 2(c),  observe that (i) by \Cref{lem:sc-rounding}, the probability that an element $x \in \unc^{\rosc}$ is handled by phase 2(c) is at most $p$ and (ii) by \Cref{lem:backup_costs},  $\expect*{\sum_{x \in \unc^{\rosc}} \cost(\backup(x))} \leq  \nicefrac{\Delta}{p} \cdot \OPT$.  Combining these, the expected cost of phase 2(c) is $\expect*{\cost(\cC^\bk)} \leq \Delta \cdot \OPT$.  

Summing these three expected costs completes the proof. 
\end{proof}
We obtain an $O(\log \nicefrac{1}{p} \cdot \log m  + \log n)$-competitive algorithm by instantiating $\rosc$ to be the $O(\log (mn))$-competitive LearnOrCover algorithm~\cite{KL21} and $\fosc$ to be the $O(\log m)$-competitive online fractional primal-dual algorithm~\cite{buchbinder2009online}. 
\subsection{$(p, k)$-Adversarially Robust Model}\label{sec:setcover}We turn to the model in which the adversary is allowed to corrupt $k$ of the $pn$ samples; let $\csamples$ be the resulting set.  To avoid being misled by the adversary, the main idea is to discard all samples after and including the first corrupted element (after randomly permuting $\csamples$).  Intuitively, we expect to retain a $1/k$ fraction of the samples, so the expected backup cost (as in  \Cref{lem:backup_costs}) should only lose an additional factor of $k$.  At this point, one might hope to use the approach in \Cref{sec:psamplesc} to complete the proof.  

Unfortunately, there are two issues.  First, the algorithm does not know which samples are corrupted.  Second, even if it is told the first corrupted sample, the conditional distribution of the clean (i.e., uncorrupted) prefix is no longer uniform, so we cannot appeal to the  guarantees of $\rosc$.  We circumvent these issues with a key technical step: we show that a large fraction of the clean prefix is sufficiently close to being uniform. 

For simplicity, we also assume  we have an estimate $\widehat{k} \geq k$ such that $\log{(k/p)} \leq \log(\widehat{k}/p) \leq 2 \log(k/p)$; we show how to
remove this assumption in \Cref{sec:app-guess-double}.

The algorithm has the following phases:
\begin{enumerate}[nosep]
    \item In the offline phase, run the $\rosc$ algorithm on the (randomly permuted) corrupted sample $\csamples$.  Let $\cC^\rosc_i$ be the solution obtained after processing first $i$ elements of $\csamples$. (Unlike the algorithm in \Cref{sec:psamplesc}, the
    output of $\rosc$ will be used in the solution only on a carefully chosen prefix of $\csamples$.)  Let $\ell_0 = 0$ and $\cC^\bk_0 = \emptyset$.

    \item In the online phase, let $x$ be the element arriving at step $t$
    \begin{enumerate}
    \item apply the $\fosc$ algorithm and round the fractional solution online with boosting parameter
    $\log(\widehat{k}/p)$ (\Cref{lem:sc-rounding}) to obtain a solution 
    $\cC^{\rfosc}_t$.  
    \item if $x \notin \cover(\cC^{\rfosc}_t \cup \cC^\bk_{t-1} \cup \cC^{\rosc}_{\ell_{t-1}})$, then 
$\cC^\bk_t = \cC^\bk_{t-1} \cup \{\backup(x)\}$,
else $\cC^\bk_t = \cC^{\bk}_{t-1}$.
    \item let $\ell_t = \max \{ i \mid i \geq \ell_{t-1} \mbox{ and } \cost(\cC^\rosc_i) \leq
    \cost(\cC^{\rfosc}_t) + \cost(\cC^\bk_t) \}$. 
    \end{enumerate}
The solution at this step is given by
$\cC^\rosc_{\ell_t} \cup \cC^{\rfosc}_t \cup  \cC^\bk_t$. 
\end{enumerate}
Let $\cC^{\rosc}, \cC^{\rfosc}, \cC^\bk$ be the final values of $\cC^\rosc_{\ell_t}, \cC^{\rfosc}_t, \cC^\bk_t$, respectively.   The final solution is $\cC^\rosc \cup \cC^{\rfosc} \cup \cC^\bk$.    By invariance (phase 2(c)), the cost of the solution on the prefix of $\csamples$ is at most the sum of the backup cost and the cost of $\rfosc$.

\subsubsection{Analysis}
Let $\crpt$ be the index of the first corrupted sample in $\csamples$, ordered by the
random permutation in phase (1); note that the algorithm does \emph{not} know
$\crpt$.
Let $\cC^{\rosc}_{\clean} = \cC^{\rosc}_{\crpt-1}$ be the solution of $\rosc$ on the clean prefix, and let
$\unc^{\rosc} = \{ x \in \univ \setminus \samples \mid x \notin \cover(\cC^\rosc_\clean) \}$ be the elements not covered by $\rosc$ on the clean prefix.  As in \Cref{sec:psamplesc}, the main difficulty is bounding the cost of the backup (i.e., phase 2(b)).  For
a prefix $P$ of $\csamples$, let
$\event(P)$ denote the event ``$\csamples_{< |P|+1} = P$ and $\crpt > |P| + 1$,'' i.e., the first $|P|$ positions are $P$ and the next
position is uncorrupted.  

We first identify a structural property of $\csamples$.  For a subset $P \subseteq \csamples$, define
\[
\skewed(P) := \left\{x \in \univ \setminus P 
\ \ \middle| \ \ \prob*{\csamples_{\abs{P}+1} = x  | \event(P)} 
< \frac{1}{2(n-\abs{P})} \right\},
\]
to be the set of elements whose chance of appearing next is too small compared with the uniform benchmark, given that the prefix $P$ seen so far is uncorrupted.  (For contrast, note that without corruptions, each of the remaining elements in $\univ \setminus P$\, would be equally likely to appear next in random order, i.e., the probability above would be $\frac{1}{n-\abs{P}}$ for every element in $\univ \setminus P$.)  A key property we show is the following:
\begin{restatable}[Near-Uniformity]{lemma}{uniformity}\label{lem:uniformitybound}
    For any $1 \leq i \leq \frac{np}{4k}$, we have $
    \expectover{\csamples_{< i}}*{\abs{\skewed(\csamples_{< i})} \ | \ \crpt > i } = O\!\left(\nicefrac{k}{p}\right)$.
\end{restatable}
We defer the technical proof to the \Cref{sec:near-unif}.  The following
is a robust analog of \Cref{lem:backup_costs}.
\begin{lemma}\label{lem:scwc_unc1}
It holds that $\expect*{\sum_{x \in \unc^{\rosc}} \cost(\backup(x))} \leq O\!\left(\nicefrac{k \Delta}{p}\right) \cdot \OPT$.
\end{lemma}
\begin{proof}
Let $T = (np)/(4k)$ (assume $T$ is a non-zero integer for simplicity).  As in the proof of
\Cref{lem:backup_costs}, let $b_{\cC}(x) = \cost(\backup(x))$ if $x \notin \cover(\cC)$ and $0$ otherwise.  Let $B_i = \sum_{x \in \univ} b_{\cC^{\rosc}_i}(x)$; $B_i$ is monotone, i.e.,
$B_i \leq B_{i-1}$. Notice that the elements in 
$\unc^{\rosc}$ are not covered by
$\cC^{\rosc}_{\clean} = \cC^{\rosc}_{\crpt-1}$.
Hence, $\sum_{x \in \unc^{\rosc}} \cost(\backup(x)) \leq B_{\crpt-1}$.  
We claim:
\begin{eqnarray}
\label{eqn:bkupcost}
\expect*{B_{\crpt-1}}
= \sum_{i = 1}^{np-k+1} \prob*{\crpt = i} \cdot
\expect*{B_{\crpt-1} | \crpt = i}
\leq O\!\left(\frac{k}{p}\right) \cdot \OPT 
+ \ O\!\left(\frac{k}{np}\right) \sum_{i = 1}^T \prob*{\crpt > i} \cdot \expect*{B_{i-1} | \crpt > i}.
\end{eqnarray}
Indeed, we do a case analysis on where the first corruption happens: 
\begin{itemize}[nosep]
\item For $i = 1$, we have $\prob*{\crpt = 1} \leq k/np$ and $B_0 \leq n \cdot \OPT$, yielding the first term in \eqref{eqn:bkupcost}.  
\item For $2 \leq i \leq T+1$, we have $B_{i-1} \leq B_{i-2}$ by monotonicity and  
$\prob*{\crpt = i | \crpt > i-1} = k/(np-i+1) = O(k/np)$ for $i \leq T$. 
So the contribution from a specific index $i$ is at most $(k/np) \prob*{\crpt > i-1} \expect*{B_{i-2} | \crpt > i-1}$.  Summing over $i = 2, \ldots, T+1$ yields the second term in \eqref{eqn:bkupcost}.  
\item For $i > T+1$, we have $B_{i-1} \leq B_T \leq (1/T) \sum_{i = 1}^T B_{i-1}$ by monotonicity.  So the contribution is at most $    \frac{1}{T} \sum_{i = 1}^T \prob*{\crpt > T+1} \expect*{B_{i-1} | \crpt > T+1}$. Now, for each $i \leq T$, the random variable $B_{i-1}$ depends only on the clean prefix of length $i-1$, which has the same distribution under either condition $\crpt > T+1$ or $\crpt > i$.  Therefore, $  \expect*{B_{i-1} | \crpt > T+1} =     \expect*{B_{i-1} | \crpt > i}$. Since $\prob*{\crpt > T+1} \leq \prob*{\crpt > i}$ for every $i \leq T$, the contribution can be bounded by $  \frac{1}{T} \sum_{i = 1}^T \prob*{\crpt > i} \expect*{B_{i-1} | \crpt > i}$, which also yields the second term in \eqref{eqn:bkupcost}.
\end{itemize}
We now analyze $B_{i-1}$ for $i \leq T$.  Condition on $\event(P)$. If an element is not in the set $\skewed(P)$, by definition, 
the conditional probability of it
appearing at the $i$th position in $\csamples$ is at least $1/2n$.  Furthermore, $b_{\cC^{\rosc}_{i-1}}(x) = 0$ for $x \in P$ since $\rosc$ has already covered the elements of $P$.   Thus,
\[
B_{i-1} = \sum_{x \in \univ} b_{\cC^{\rosc}_{i-1}}(x) 
\leq 2n \expect*{b_{\cC^{\rosc}_{i-1}}(\csamples_i) | \event(P)} + \abs{\skewed(P)} \cdot \OPT. 
\]
Taking expectations, summing over $i = 1, \ldots, T$, and using
\Cref{lem:uniformitybound}, we obtain:
\begin{align*}
& \sum_{i = 1}^T \prob*{\crpt > i} \cdot \expect*{B_{i-1} | \crpt > i} \\
& \leq 2n \expect*{\sum_{i = 1}^T \ind{[\crpt > i]} \cdot b_{\cC^{\rosc}_{i-1}}(\csamples_i)} + \OPT \cdot \sum_{i = 1}^T \prob*{\crpt > i} \expect*{\abs{\skewed(\csamples_{< i})} | \crpt > i} \\
& \leq 2n \expect*{\cost(\cC^{\rosc}_{\clean})} + T \cdot O\!\left(\frac{k}{p}\right) \cdot \OPT \\
& \leq O(n \Delta) \cdot \OPT,
\end{align*}
where in the last inequality the $\rosc$ guarantees apply since
conditioned on the clean prefix, its order is uniform at random. 
Using this bound in \eqref{eqn:bkupcost} completes the proof.
\end{proof}

\begin{restatable}{thm}{setthm}\label{thm:set-corruptions}
If $\rosc$ is a random order set cover algorithm with competitive ratio $\Delta$ and $\fosc$ is a fractional online set cover algorithm with competitive ratio $\alpha$, then the algorithm above has competitive ratio $O(\Delta + \alpha \log \nicefrac{k}{p})$.
\end{restatable}
\begin{proof}
As in the proof of \Cref{thm:oxtsc_poly}, it is
easy to see that the cost arising from phase 2(a) 
of the algorithm can be bounded as
$\expect*{\cost(\cC^{\rfosc})}
\leq \alpha \log \nicefrac{\widehat{k}}{p} \cdot \OPT$.  Also, by the invariant condition in phase 2(c), we know that $\cost(\cC^{\rosc}) \leq
\cost(\cC^{\rfosc}) + \cost(\cC^{\bk})$.  Hence, it remains to 
bound $\expect*{\cost(\cC^{\bk})}$.

To this end, define the \emph{critical} step
\[
t^*:= \min \{ t \mid \cost(\cC^{\rfosc}_t) + \cost(\cC^\bk_t) \geq \cost(\cC^{\rosc}_{\clean}) \}.
\]
(If the set is vacuous, $t^* = \infty$.)

If $t^* = \infty$, then $\cost(\cC^{\bk}) \leq \cost(\cC^{\rosc}_{\clean})$.  Otherwise, we do a case analysis on $t^*$:
\begin{itemize}[nosep]
\item For $t \leq t^*$, by the definition of $t^*$, we have
$\cost(\cC^{\bk}_{t^*-1}) \leq  \cost(\cC^{\rosc}_{\clean})$.  If a backup set was added to the solution at step $t^*$, it contributes at most
$\OPT$ to the backup cost.  Hence, the backup cost up to
step $t^*$ is at most $\cost(\cC^{\rosc}_{\clean}) + \OPT$.  Now, using the guarantees of $\rosc$ on the clean prefix, it follows that
the expected backup cost for $t \leq t^*$ is 
at most $(\Delta + 1) \cdot \OPT$.
\item For $t > t^*$, by definition, $\cost(\cC^{\rosc}_{\clean})$ is less than $\cost(\cC^{\rfosc}_t) + \cost(\cC^\bk_t)$.  Since
phase 2(c) ensures $\ell_t \geq \crpt - 1$, an element
contributes to the backup cost only if it is in $\unc^{\rosc}$ and is also missed
by the rounding step.  Observe that
by \Cref{lem:sc-rounding}, each element in $\unc^{\rosc}$ is not covered by $\overline{\fosc}$ with probability at most $\nicefrac{p}{\widehat k} \leq \nicefrac{p}{k}$ and
by \Cref{lem:scwc_unc1}, $\expect*{\sum_{x \in \unc^{\rosc}} \cost(\backup(x))} \leq O(\nicefrac{k\Delta}{p}) \cdot \OPT$.  Combining these, the expected backup cost for
$t > t^*$ is also $O(\Delta) \cdot \OPT$.  
\end{itemize}
Thus, $\expect*{\cost(\cC^{\bk})} \leq O(\Delta) \cdot \OPT$.
\end{proof}
As in \Cref{sec:psamplesc}, by choosing $\rosc$ to be LearnOrCover \cite{KL21} and $\fosc$ to be online primal-dual \cite{buchbinder2009online}, we get the $O(\log \nicefrac{k}{p} \cdot \log m + \log n)$-competitive ratio of \Cref{thm:robust-cov} as a corollary.

\subsubsection{Near-Uniformity Lemma}
\label{sec:near-unif}

We end the section by proving~\Cref{lem:uniformitybound}.  Let $(a)_b = (a)(a-1)\cdots(a-b+1)$ denote the falling factorial.

\begin{proof}
Conditioned on $\crpt > i$ (the first corrupted sample arriving after index $i$), the first $i-1$ elements of $\csamples$ are a uniformly random $(i-1)$-subset of the $np-k$ uncorrupted elements $\Clean := \samples \cap \csamples$. Thus, the probability that the prefix equals a specific set $P$ of size $i-1$ is:
\[
    \prob*{\csamples_{<i} = P \mid \crpt > i}
    = \frac{\prob*{P\subseteq \Clean}}{\binom{np-k}{i-1}}
    = \frac{\prob*{P \subseteq \Clean \mid P \subseteq \samples} \cdot \prob{P \subseteq \samples}}{\binom{np-k}{i-1}}.
\]
Since $\samples$ is a uniformly random $np$-subset of $X$, we have $\prob*{P \subseteq \samples} = \binom{n-i+1}{np-i+1} / \binom{n}{np} = (np)_{i-1} / (n)_{i-1}$. Defining $h_P := \prob*{P \subseteq \Clean \mid P \subseteq \samples}$, we can rewrite the expectation we wish to bound as
\begin{align} 
   \expectover{\csamples_{<i}}*{\abs{\skewed(\csamples_{<i})} \mid \crpt > i}
    &= \sum_{P: \abs{P} = i-1} \frac{(np)_{i-1}/(n)_{i-1}}{\binom{np-k}{i-1}} \cdot h_P \cdot \abs{\skewed(P)}  \nonumber \\
    &\le \frac{(np)_{i-1}}{(np-k)_{i-1}}  \cdot \left(\max_{P: \abs{P} =i-1} h_P \cdot \abs{\skewed(P)}\right) \nonumber \\
    &\leq \prod_{j=0}^{i-2} \left(1 + \frac{k}{np-k-j}\right) \left(\max_{P: \abs{P} =i-1} h_P \  \abs{\skewed(P)}\right) \nonumber \\
    &= \exp\left(O\!\left(k(i-1) / (np)\right)\right) \left(\max_{P: \abs{P} =i-1} h_P \cdot \abs{\skewed(P)}\right) \nonumber \\
    & = O(1) \cdot \max_{P: \abs{P} =i-1} h_P \cdot \abs{\skewed(P)}, \nonumber
\end{align}
where we used $np-k-j \geq np/2$ for $j \in [0, i-2]$ and $1+x \leq e^{x}$ for $x \geq 0$. 

If $\max_P \abs{\skewed(P)} \leq \frac{16 k}{3p}$, we are done since $h_P \leq 1$. In the remainder of the proof, we show for any prefix $P$ of length $i-1$ such that $\abs{\skewed(P)} > \frac{16 k}{3p}$ we get an even sharper bound of $h_P \cdot \abs{\skewed(P)} \leq O(1/p)$.

Let $K(P) = \abs{\skewed(P)\cap \samples}$ be the number of skewed elements that are sampled. Since, conditioned on $P \subseteq \samples$, the set $\samples \setminus P$ is a uniformly random $(np-i+1)$-sized subset of $X \setminus P$, the random variable $K(P)$ follows a hypergeometric distribution, which has the properties
\begin{align*}
&\mu(P) := \expect*{K(P) \mid P\subseteq \samples} = \abs{\skewed(P)}\frac{np-i+1}{n-i+1}.\\
&\Var(K(P) \mid P\subseteq \samples) \leq \mu(P).
\end{align*}
Furthermore, since $\abs{\skewed(P)} > \frac{16k}{3p}$ and $i-1 \leq \frac{np}{4}$, we know that $\mu(P) >  4k$.

By definition, $\prob*{y\in \Clean \mid P \subseteq \Clean} < \frac{np-k-i+1}{2(n-i+1)} \le \frac{np-i+1}{2(n-i+1)}$, so by the linearity of expectation, the conditional expectation of the number of skewed and uncorrupted elements that get sampled is
\begin{align}
\expect*{\abs{\skewed(P)\cap \Clean} \mid P\subseteq \Clean} \le \frac{\abs{\skewed(P)}(np-i+1)}{2(n-i+1)} = \frac{\mu(P)}{2}. \label{line:kpclean_vs_mu}
\end{align}
Every element counted by $K(P)=\abs{\skewed(P)\cap \samples}$ either lies in $\skewed(P)\cap \Clean$, or is modified by the adversary. Since the adversary modifies exactly $k$ elements, we must have 
\begin{align*}
    \expect*{K(P) \mid P \subseteq \Clean} \le \expect*{ \abs{\skewed(P)\cap \Clean} \mid P \subseteq \Clean}+k \leq \frac{3}{4} \mu(P).
\end{align*}

With this in hand,
\begin{align*}
\mu(P) &\geq \Var(K(P) \mid P \subseteq \samples)  \\
&= h_P \cdot \expect*{(\mu(P) - K(P))^2 \mid P\subseteq \Clean} + (1-h_P) \cdot \expect*{(K(P)-\mu(P))^2 \mid P\not \subseteq \Clean} \\
&\geq h_P \cdot \expect*{(\mu(P) - K(P))^2 \mid P\subseteq \Clean} \\
&\geq h_P \cdot \expect*{\mu(P) - K(P) \mid P\subseteq \Clean}^2 \tag{Jensen's inequality} \\
&\geq h_P \cdot \left(\frac{\mu(P)}{4}\right)^2 \tag{Using \eqref{line:kpclean_vs_mu}}
\end{align*}
Finally, using $\mu(P) = \abs{\skewed(P)}\frac{np-i+1}{n-i+1}$, we conclude $h_P \abs{\skewed(P)} \leq \frac{16(n-i+1)}{np-i+1} \leq \frac{16n}{np - np/4} = O(1/p)$.
\end{proof}

\section{Online Metric Facility Location} \label{sec:metricpsample} Let $(V,d)$ be a metric space, and let $\mathcal{F} \subseteq V$ be a set of facilities. Each facility $f \in \mathcal{F}$ has a non-negative facility opening cost $c_f \in [1, 2^q]$. Let $X \subseteq V$ be a set of $n$ clients that arrive online. At any time $t$, the online algorithm maintains a set of open facilities $F_t$. When a client $v^t \in X$ arrives at time $t$, the algorithm has to connect it to an open facility, and can open new facilities in the process. The goal is to minimize the total cost given by the sum of opening costs of the open facilities and the distance of the clients to the open facilities they are connected to, i.e., the objective
\[
\obj(F, X) := \sum_{f \in F}c_f+ \sum_{v \in X}d(v,F).
\]
In the \emph{$p$-sample model} for metric facility location, the algorithm samples a set $\samples$ of clients uniformly at random from all subsets of clients size $s = pn$, before the online sequence begins. Then the online input consists of the remaining clients 
$X \setminus \samples$ in an arbitrary (and possibly adversarial) order. 

In the \emph{$(p,k)$-adversarially robust model}, the adversary is allowed to modify any $k$ clients in $\samples$ after they are sampled, but before they are considered by the algorithm; we use $\csamples$ to denote the corrupted samples. 

\subsection{$p$-Sample Model}

Meyerson \cite{meyerson2001online} showed an elegant online algorithm for online metric facility location that we will need to use as a white-box, so we begin the section with a brief description. For convenience, define the type of a facility $f \in \mathcal{F}$ to be $\floor{\log_2(c_f)}$, and let $q$ be the maximum facility type. Let $F^{(\leq {\ell})} \subseteq \mathcal{F}$ denote the set of facilities of less than or equal to $\ell$. 

\paragraph{Meyerson's Algorithm} When $v \in X$ arrives online, for each $0 \leq \ell \leq q$, open the closest facility of type $\ell$ to $v$ with probability 
\[
\min\left(\frac{\delta_{\ell-1}(v) - \delta_{\ell}(v) }{2^{\ell}} , 1 \right),
\]
where $\delta_{\ell}(v) := d(F^{(\leq {\ell})} \cup F_{< v} , v)$ and $F_{< v}  \subseteq \mathcal{F}$ is the set of facilities open just before $v \in X$ arrives.  We then connect $v$ to the closest available facility.
Let $F_{v}$ be the set of facilities opened by the algorithm on the arrival of $v$, and let $F_{\leq v} = F_{<v} \cup F_v$.

Meyerson showed the following guarantee.\footnote{As observed by \cite{fotakis08}, in Meyerson's original paper the guarantee is written as a simple logarithm, but a simple modification of the same argument gives the bound below. For completeness, we give the proof in \Cref{sec:metricfacility}.} 
\begin{restatable}[Meyerson's Guarantee \cite{meyerson2001online}]{lemma}{meyersonguarantee}\label{app-meyerson_mfl}
The expected cost incurred by Meyerson's algorithm on any subset $X' \subseteq X$ of clients is at most $\OPT(X') \cdot O\left(\frac{\log |X'|}{\log \log |X'|}\right).$
\end{restatable}

\subsubsection{The Algorithm}
The algorithm is the same as that of \cite{ArgueFGS22}, but the latter has an analysis issue that we fix and tighten in what follows. The algorithm requires as input an offline $\alpha$-approximation to Metric Facility Location, and proceeds in two phases.
\begin{enumerate}[nosep]
\item In the offline phase, scale the connection cost of each client in the sample by $\frac{1}{p}$ and run any constant factor approximation to the metric facility location objective
\begin{align}
\objhat(F, S) := \sum_{f \in F}c_f+ \sum_{v \in S}d(v,F) \cdot \frac{1}{p}. \label{line:phase1-obj}
\end{align}
to obtain facility set $\widehat F$.
\item In the online phase, run Meyerson's algorithm \cite{meyerson2001online} (denoted $\omfl$) on the online input $X$ to obtain the facility set $F$ (and connect clients appropriately). 
\end{enumerate}

\subsection{Analysis}
We prove the following theorem.
\begin{restatable}{thm}{mflfixthm}\label{thm:mflfixtotalcost}

The expected cost incurred by the algorithm is $O\left(\frac{\alpha \log \nicefrac{1}{p}}{\log \log \nicefrac{1}{p}}\right) \cdot \OPT(X)$.
\end{restatable}
\begin{proof}
First, we show that the expected cost incurred in the first stage is $O(\OPT(X))$. Let $F^*$ be the optimal solution for the full online input. Then
\begin{align}\expect*{\obj(\widehat F, S)} \leq \expect*{\objhat(\widehat F, S)} \leq \alpha \cdot \expect*{\objhat(F^*, S)} = \alpha \left[\sum_{f^* \in F^*} c_f + \sum_{v \in X} d(v,F^*) \frac{\Pr[v \in S] }{p}  \right] =\alpha \cdot \OPT(X) \label{line:stage1}\end{align}
It remains to bound the cost of the online phase, and for every $f^* \in F^*$ serving client set $C$, we will bound the online algorithm's cost to serve $C$ separately. We write this as $\cost(C) := \sum_{v \in C} F_v + d(F_{\leq v}, v)$.

Define the bad event $\cB_C$ to be $\{|S \cap C| < p|C| / 2\}$, and note that by a Chernoff bound, $\prob{\cB_C} \leq e^{- p|C| /8}$. In the good case that $\lnot \cB_C$ holds,  the expected cost (over the randomness of the online algorithm only) to connect $C$ is
\begin{align}
&\expectover{F}*{\cost(C)} = \sum_{v \in C} \expect*{d(F_{\leq v}, v)} + \expect*{\sum_{f \in F_v} c_f} \nonumber \\& \leq \sum_{v \in C} \expect*{d(F_{\leq v}, v)} + \sum_{\Delta_1, \ldots, \Delta_q} \expect*{\sum_{f \in F_v} c_f | \delta_1(v) = \Delta_1, \ \ldots, \ \delta_q(v) = \Delta_q} \prob{\delta_1(v) = \Delta_1, \ \ldots, \ \delta_q(v) = \Delta_q} \nonumber \\
&\leq \sum_{v \in C} \expect*{d(F_{\leq v}, v)} + \sum_{\Delta_1, \ldots, \Delta_q} \sum_{\ell=0}^q \expect*{2^{\ell+1} \min \left(\frac{\Delta_{{\ell}-1} - \Delta_\ell}{2^{\ell}}, 1 \right) \prob{\delta_1(v) = \Delta_1, \ \ldots, \ \delta_q(v) = \Delta_q}} \nonumber\\
&\leq \sum_{v \in C} \expect*{d(F_{\leq v}, v)} + \sum_{\Delta_1, \ldots, \Delta_q}\sum_{{\ell}=0}^{q} \expect*{2(\Delta_{{\ell}-1} - \Delta_{\ell}) \prob{\delta_1(v) = \Delta_1, \ \ldots, \ \delta_q(v) = \Delta_q}} \nonumber \\
&\leq \sum_{v \in C} \expect*{d(F_{\leq v}, v) + 2d(F_{< v}, v)} \leq 3 \sum_{v\in C} d(\widehat F, v). \label{line:telescope}
\intertext{Above, \eqref{line:telescope} is from evaluating the telescoping series and observing that $\delta_{-1} = d(F_{<v}, v)$ which does not depend on the realizations $\Delta_1, \ldots, \Delta_q$. Using the fact that $\widehat F \subseteq F_{< v} \subseteq F_{\leq v}$ and applying the triangle inequality multiple times, we continue to bound this by}
&\leq \sum_{v \in C}  3 (d(f^*, v) + d(\widehat F, f^*)) = 3\sum_{v \in C} d(f^*, v) + 3  \frac{|C|}{|S \cap C|}  \sum_{v \in C \cap S}d(\widehat F, f^*) \nonumber\\
&\leq  3\sum_{v \in C} d(f^*, v) + 3 \frac{|C|}{|S \cap C|}  \left[\sum_{v \in C \cap S} d(f^*, v) + d(v, \widehat F)\right] \leq 2\sum_{v \in C} d(f^*, v) + \frac{6}{p}  \left[\sum_{v \in C \cap S} d(f^*, v) + d(v, \widehat F)\right] \label{line:using_notb},
\end{align}
where \eqref{line:using_notb} used the condition $\lnot \cB_C$. Taking the expectation this time over the realization of the samples $S$ (conditioned on $\lnot \cB_C$):
\begin{align}
\expectover{F, S}*{\cost(C) | \lnot \cB_C} &\leq 3\sum_{v \in C} d(f^*, v) + 6 \expectover{S}*{\sum_{v \in C\cap S} \frac{1}{p} d(f^*, v)| \lnot \cB_C} + 6\expectover{S}*{\sum_{v \in C \cap S}\frac{1}{p} d(v, \widehat F)| \lnot \cB_C}
\label{line:uncondition}
\end{align}
Moreover, we also have
\begin{equation}\label{totalprob_1}
\prob{\lnot \cB_C}\cdot
\expectover{S}*{\sum_{v \in C\cap S}\frac{1}{p}d(v,\widehat F)\mid \lnot \cB_C}
=
\expectover{S}*{\ind\{\lnot \cB_C\}\sum_{v \in C\cap S}\frac{1}{p}d(v,\widehat F)}
\leq
\expectover{S}*{\sum_{v \in C\cap S}\frac{1}{p}d(v,\widehat F)}
\end{equation}
\begin{equation}\label{totalprob_2}
\prob{\lnot \cB_C}\cdot
\expectover{S}*{\sum_{v \in C\cap S}\frac{1}{p}d(v,F^*)\mid \lnot \cB_C}
=
\expectover{S}*{\ind\{\lnot \cB_C\}\sum_{v \in C\cap S}\frac{1}{p}d(v,F^*)}
\leq
\expectover{S}*{\sum_{v \in C\cap S}\frac{1}{p}d(v, F^*)}
\end{equation}

To finish the proof, we use the law of total probability to write the total expected cost, and substitute in \eqref{line:uncondition}, \eqref{totalprob_1}, \eqref{totalprob_2}, the probability bound for $\cB_C$, and Meyerson's gaurantee from \Cref{app-meyerson_mfl}:
\begin{align}
    \expectover{F, S}*{\cost(F)} &\leq  \left(\sum_{C}\expectover{F, S}*{\cost(C) | \lnot \cB_C} \cdot \prob{\lnot \cB_C}\right) +  \left(\sum_{C}\expectover{F, S}*{\cost(C) | \cB_C} \cdot \prob{\cB_C} \right) \nonumber\\
    &\leq  \left(3\sum_{v \in X} d(f^*, v) + 6 \cdot \expectover{S}*{\sum_{v \in S} \frac{1}{p} d(f^*, v)} + 6 \cdot\expectover{S}*{\sum_{v \in S}\frac{1}{p} d(v, \widehat F)} \right) \nonumber \\
    &\quad \quad \quad  + \left(\sum_C\left(e^{-p|C|/8} \cdot O\left(\frac{\log |C|}{\log \log |C|}\right) \right) \cdot  \OPT(C) \right) \nonumber \\
    & \leq \left(3 \OPT(X) + 6 \expectover{S}*{\objhat(F^*, S)} + 6 \expectover{S}*{\objhat(\widehat F, S)}\right)  +  O\left(\frac{\log |\nicefrac1p|}{\log \log |\nicefrac1p|}\right) \cdot  \sum_C \OPT(C) \label{line:reusing_objhat}\\
    & = O\left(\frac{\alpha \log |\nicefrac1p|}{\log \log |\nicefrac1p|}\right) \cdot  \OPT(X). \label{line:bounding_objhat}
\end{align}
    To see the bound on the first group of terms in step \eqref{line:reusing_objhat}, note that these are the connection costs of $\OPT(X)$, $\objhat(F^*, S)$, and $\objhat(\widehat F, S)$ respectively. The bound on the second group in \eqref{line:reusing_objhat} holds because when $|C| \geq 10/p$, the exponential term dominates, and when $|C| < 10/p$ then $\log |C| / \log \log |C| = O(\log \nicefrac 1p / \log \log \nicefrac1p)$ also. Finally, \eqref{line:bounding_objhat} follows from the bound on the first phase objective's costs \eqref{line:stage1}.
\end{proof}

Due to space considerations, we relegate the generalization to the model with corruptions to \Cref{sec:metricfacility}.

\section{Closing Remarks}\label{sec:closing}In this paper, we gave new sample-augmented online algorithms for fundamental problems such as set cover, facility location, and Steiner tree. Our contributions are two-fold. First, we obtained tight bounds for both set cover and metric facility location, thereby answering open questions in \cite{GuptaKL24} and \cite{ArgueFGS22}, respectively. Second, we identified a key weakness of prior work in this setting, namely that the algorithms are not robust to errors or corruptions in the samples. To remedy this, we gave the first sample-augmented online algorithms that are robust to adversarial corruption, where an adversary may arbitrarily modify a limited number of samples. For set cover, facility location, and Steiner tree, our algorithms achieve tight bounds and strictly generalize the previous guarantees obtained in the absence of sample corruption.

We hope that our work will stimulate further research on robustness in algorithms that leverage samples or learned side information. This includes not only the sample-augmented online model considered in this paper, but also other settings such as stochastic algorithms that exploit samples from the input distribution to circumvent worst-case lower bounds. It would also be interesting to study richer models of imperfect side information, including stochastic noise, distribution shift, and adaptive or strategic corruptions. More broadly, we believe that incorporating robustness into the design of sample-augmented algorithms can help bridge the gap between classical worst-case algorithms and machine-learned heuristics, leading to algorithms that are both theoretically well-founded and practically reliable.

\paragraph{AI Disclosure.} 

We used ChatGPT $5.5$ Pro for minor copy editing. The authors are responsible for all content of the paper.

\bibliographystyle{alpha}
\bibliography{main}

\appendix

\section{Online Metric Facility Location in the $(p, k)$-Adversarially Robust Model}\label{sec:metricfacility}We turn to the model with corruptions. As in the $p$-sample model, the online part of the algorithm is Meyerson's algorithm~\cite{meyerson2001online}, denoted $\omfl$. The new ingredient is an offline randomized algorithm $\rmfl$, which is run on prefixes of the corrupted sample.

For an online time $t$, let $\cost(\omfl_{\leq t})$ denote the cost incurred by $\omfl$ up to and including time $t$, not including facilities opened by $\rmfl$.

\subsection{The Algorithm}

Let $f(x) := \frac{\ln \nicefrac{x}{p}}{\ln \ln \nicefrac{x}{p}}$. Define a sequence of integers $K_1,K_2,\ldots,K_w$ as follows: $K_1=3$, and for $i\geq 2$, $K_i$ is the smallest integer at most $np$ such that $f(K_i)\geq 2f(K_{i-1})$. Let $\sigma(\csamples)$ be a uniformly random permutation of $\csamples$, fixed throughout the algorithm. For any $1\leq r\leq np$, let $\csamples_{\leq r}$ denote the first $r$ clients in this order. The algorithm uses an offline $\alpha$-approximation to metric facility location, where $\alpha=O(1)$, as a subroutine. The subroutine $\rmfl$ applies this approximation to scaled prefixes of the corrupted sample.

\paragraph{Algorithm $\rmfl$:} This algorithm takes as input the corrupted sample $\csamples$ and a budget parameter $B$. For any positive prefix length $r$, define the scaled prefix objective
\[
\widehat{\obj_r}(F,\csamples_{\leq r})
:=
\sum_{f\in F}c_f+\frac{n}{r}\sum_{v\in \csamples_{\leq r}}d(v,F).
\]
For each $i\in[w]$, the algorithm sets $L_i$ to be the greatest positive integer $r\leq \frac{np}{16K_i}$ such that the offline $\alpha$-approximation, run on the scaled prefix instance $\csamples_{\leq r}$, returns a solution of cost at most $\frac{B}{f(K_i)}$.
The algorithm then opens the corresponding facilities $\widehat{F}_i$. If no such $r$ exists, set $L_i=0$ and $\widehat{F}_i  = \emptyset$.
The online algorithm proceeds as follows.
\begin{enumerate}[nosep]
\item On every online arrival $v^t$, run $\omfl$ on $v^t$, with all previously opened $\rmfl$ facilities also available for connection. The facilities opened by $\omfl$ are denoted $F$.
\item At the end of time step $1$, run $\rmfl$ on $\csamples$ with budget parameter $\cost(\omfl_{\leq 1})$.
\item For each later time step $t$, let $t'$ be the last time at which $\rmfl$ was executed. If $\cost(\omfl_{\leq t})>2\cost(\omfl_{\leq t'})$, then run $\rmfl$ on $\csamples$ with budget parameter $\cost(\omfl_{\leq t})$. Let $t_{\mathrm{fin}}$ be the last time at which $\rmfl$ is run.
\end{enumerate}

\subsection{Analysis}\label{sec:mfl-corruptions-analysis}

We first bound the cost incurred by $\rmfl$ against that of $\omfl$. Since $f(K_i)$ increases geometrically, the total cost incurred by $\rmfl$ is at most $\sum_{i=1}^w \frac{B}{f(K_i)} \leq \frac{2B}{f(K_1)} = \frac{2B}{f(3)} = O(B)$. Since $\offst$ is called with parameter $\cost(\onst_{\leq t})$ exactly once every time $t$ the cost of $\onst$ doubles, the offline cost is at most a constant times the online cost. We summarize in the following observation.
\begin{observation}\label{ONLINEMFLcost}
$\cost(\rmfl) = O(\cost(\omfl))$.
\end{observation}

As a consequence of the above, we may assume $k\leq \frac{np}{2}$ henceforth. This is because when $k>\frac{np}{2}$, the desired ratio is $O\left(\frac{\log n}{\log\log n}\right)$, which is achieved by $\omfl$ due to \Cref{app-meyerson_mfl}.

Let $K:=K_g=\min\{K_i\mid K_i\geq k\}$. Since $K_g$ is the first value in the sequence that is at least $k$, and since $f(K_i)$ grows geometrically, we have
$f(K)=O(f(k)).$
Let $\crpt$ be the index of the first corrupted sample in the random order $\sigma(\csamples)$, i.e.\ the first element of $\csamples\setminus\samples$. Define $\Clean := \samples \cap \csamples$. Define
$\tau := \min\left\{\crpt-1,\frac{np}{16K}\right\}.$
Let $\beta_K(\csamples)$ be the cost of the $\alpha$-approximate metric facility location solution for $\csamples_{\leq \tau}$, where each client has demand $\frac{n}{\tau}$; if $\tau=0$, set $\beta_K(\csamples)=0$.

\begin{lemma}\label{mflbeta}
We have $\expect*{\beta_K(\csamples)}\leq 2 \alpha \cdot \OPT(X)$.
\end{lemma}

\begin{proof}
Let $F^*$ be an optimal metric facility location solution for $X$. Condition on a fixed realization of $\samples,\csamples$ and on $\tau=r>0$. Then $\csamples_{\leq r}$ is a uniformly random $r$-subset of $\Clean$. Since $\abs{\Clean}=np-k\geq \frac{np}{2}$, we have
\[
\expect*{\frac{n}{\tau}\sum_{v\in \csamples_{\leq \tau}} d(v,F^*) \mid \tau=r} = \frac{n}{\abs{\Clean}}\sum_{v\in \Clean}d(v,F^*) \leq \frac{2}{p}\sum_{v\in \Clean}d(v,F^*).
\]
Since this holds for every $r$ and every realization of $\samples,\csamples$, taking expectations gives
$\expect*{\frac{n}{\tau}\sum_{v\in \csamples_{\leq \tau}} d(v,F^*)} \leq \frac{2}{p}\expect*{\sum_{v\in \Clean}d(v,F^*)}.$
Consequently, the expected cost incurred by using $F^*$ to serve $\csamples_{\leq \tau}$, where each client has demand $\frac{n}{\tau}$, is at most
$\sum_{f\in F^*}c_f+\frac{2}{p}\expect*{\sum_{v\in \Clean}d(v,F^*)} \leq \sum_{f\in F^*}c_f+2\sum_{v\in X}d(v,F^*) \leq 2\OPT(X).$
Since $\beta_K(\csamples)$ is the cost of an $\alpha$-approximate solution for the scaled instance $\csamples_{\leq \tau}$, the conclusion follows.
\end{proof}

We will use the following lemma to compare the clean part of a longer prefix $\csamples_{\leq L}$ to the prefix $\csamples_{\leq \tau}$ before the first corrupted sample. Note that any sum involving $\csamples_{\leq 0}$ is interpreted to equal $0$.

\begin{lemma}\label{cl:mflcharging}
Let $L$ be any random prefix length satisfying $\tau\leq L\leq \frac{np}{16K}$. Let $F^*$ be an optimal metric facility location solution for $X$. Then,
$\expect*{\frac{n}{L}\sum_{v\in \csamples_{\leq L}\cap\samples}d(v,F^*)} \leq 7\expect*{\frac{n}{\tau}\sum_{v\in \csamples_{\leq \tau}}d(v,F^*)} \leq 14\OPT(X).$
\end{lemma}

\begin{proof}
If $L=0$, then the first inequality is immediate, so assume $L>0$. We prove the first inequality conditioned on any fixed realization of $\samples,\csamples$, and then take expectations over $\samples,\csamples$. Throughout this proof, omit this conditioning. Consider a client $v\in \Clean$ and a position $t\leq \frac{np}{16K}$. Let $\sigma_v$ denote the position of $v$ in $\sigma(\csamples)$, and define
$p_t=\prob*{\crpt<t\mid \sigma_v=t}$ as well as $q_t=\prob*{t<\crpt\leq 2t\mid \sigma_v=t}.$
We will show $p_t \leq 3q_t$. By a union bound over the $k$ corrupted samples, $p_t\leq \frac{kt}{np}\leq \frac{1}{16}$. On the other hand,
\begin{align*}
q_t &= \prob*{t<\crpt\mid \sigma_v=t}\cdot \prob*{\crpt\leq 2t\mid t<\crpt,\sigma_v=t}\\
&\geq \frac{15}{16}\prob*{\crpt\leq 2t\mid t<\crpt}\\
&= \frac{15}{16}\left(1-\frac{\binom{np-2t}{k}}{\binom{np-t}{k}}\right).
\end{align*}
Since $kt\leq Kt\leq \frac{np}{16}$, we have
\begin{align*}
\frac{\binom{np-2t}{k}}{\binom{np-t}{k}} = \prod_{j=0}^{k-1}\left(1-\frac{t}{np-t-j}\right)
\leq \left(1-\frac{t}{np-t}\right)^k
\leq e^{-\frac{kt}{np-t}}
\leq 1-\frac{kt}{2(np-t)}
\leq 1-\frac{kt}{2np}.
\end{align*}
Thus $q_t\geq \frac{15kt}{32np}\geq \frac{15}{32}p_t$, and hence $p_t \leq 3q_t$.

We now bound the scaled connection cost of the clean clients in $\csamples_{\leq L}$. Since $\tau\leq L$, we split
\begin{align}
\expect*{\frac{n}{L}\sum_{v\in \csamples_{\leq L}\cap\samples}d(v,F^*)} = \expect*{\frac{n}{L}\sum_{v\in \csamples_{\leq \tau}}d(v,F^*)}+\expect*{\frac{n}{L}\sum_{v\in (\csamples_{\leq L}\setminus \csamples_{\leq \tau})\cap\samples}d(v,F^*)}. \label{eq:mfl_terms_to_bound}
\end{align}
The first term is at most $\expect*{\frac{n}{\tau}\sum_{v\in \csamples_{\leq \tau}}d(v,F^*)}$. For the second term, if $v\in(\csamples_{\leq L}\setminus \csamples_{\leq \tau})\cap\samples$ is at position $t$, then $\tau<t$, and hence $\crpt<t$. Therefore,
\begin{align}
\expect*{\frac{n}{L}\sum_{v\in(\csamples_{\leq L}\setminus \csamples_{\leq \tau})\cap\samples}d(v,F^*)} &\leq \sum_{v\in\Clean}\sum_{t=1}^{np/(16K)}\prob*{\sigma_v=t,\crpt<t}\cdot \frac{n}{t}d(v,F^*) \nonumber\\
&\leq 3\sum_{v\in\Clean}\sum_{t=1}^{np/(16K)}\prob*{\sigma_v=t}\prob*{t<\crpt\leq 2t\mid \sigma_v=t}\cdot \frac{n}{t}d(v,F^*) \label{eq:mfl-second-charge-start1}\\
&\leq 6\expect*{\sum_{v\in\Clean}\sum_{t=1}^{np/(16K)}\ind\{\sigma_v=t,\ t<\crpt\leq 2t\}\cdot \frac{n}{\tau}d(v,F^*)} \label{eq:mfl-second-charge-start2} \\
&\leq 6\expect*{\frac{n}{\tau}\sum_{v\in\csamples_{\leq \tau}}d(v,F^*)}. \label{eq:mfl-second-charge-start3}
\end{align}
Inequality \eqref{eq:mfl-second-charge-start1} uses $p_t\leq 3q_t$, and \eqref{eq:mfl-second-charge-start2} uses that fact that if $\sigma_v=t$ and $t<\crpt\leq 2t$, then $t\leq \tau$ and $\tau\leq 2t$, so $\frac{n}{t}\leq \frac{2n}{\tau}$. Step \eqref{eq:mfl-second-charge-start1} holds because for each $v\in\Clean$, there is exactly one value of $t$ with $\sigma_v=t$, and Moreover, if $\sigma_v=t$ and $t<\crpt\leq 2t$, then $t\leq\tau$, so $v\in\csamples_{\leq \tau}$. 

Substituting the last bound \eqref{eq:mfl-second-charge-start3} into \eqref{eq:mfl_terms_to_bound} gives the first inequality of the lemma. The second inequality follows from the proof of \Cref{mflbeta}, which showed $\expect*{\frac{n}{\tau}\sum_{v\in \csamples_{\leq \tau}}d(v,F^*)}\leq 2\OPT(X)$.
\end{proof}

We next show that once sufficiently good facilities from the corrupted sample have been opened, the remaining cost of $\omfl$ is small.

\begin{lemma}\label{mflonlinepost}
Let $L$ be any random prefix length satisfying $\tau\leq L\leq \frac{np}{16K}$, and let $\widehat F$ be any random set of facilities. Suppose $\expect*{\frac{n}{L}\sum_{v\in \csamples_{\leq L}}d(v,\widehat F)}\leq O(\alpha \cdot \OPT(X))$, Then the expected cost incurred by $\omfl$ on any suffix of the online sequence $X$, if all facilities in $\widehat F$ are already open, is at most $\OPT(X)\cdot O\left(\frac{\alpha \log \nicefrac{k}{p}}{\log \log \nicefrac{k}{p}}\right)$.
\end{lemma}

\begin{proof}
Let $F^*$ be an optimal facility location solution for $X$. For every $f^*\in F^*$, let $C$ be the cluster of clients served by $f^*$ in the optimal solution, and write $\cost(C)$ for the cost incurred by $\omfl$ on the clients of $C$ in the suffix. In the following steps, $\expectation_{F}$ denotes expectation only over the internal randomness of $\omfl$, after fixing $\Gamma = (\samples,\csamples,\sigma,L$, $\widehat F)$. 

Fix such a cluster $C$. For any client $v\in C$, let $F_{<v}$ be the set of facilities open just before $v$ arrives in the suffix, and let $\delta_\ell(v)=d(F^{(\leq \ell)}\cup F_{<v},v)$. As in the proof of \Cref{thm:mflfixtotalcost}, the expected facility-opening cost incurred due to the arrival of $v$ is at most $2\delta_{-1}(v)$. Since the connection cost of $v$ is at most $\delta_{-1}(v)$, and since the facilities in $\widehat F$ are open throughout the suffix, we have
$\expectover{F}*{\cost(C)}\leq 3\sum_{v\in C}d(v,\widehat F).$
Using the triangle inequality,
\[
\sum_{v\in C}d(v,\widehat F)\leq \sum_{v\in C}(d(v,f^*)+d(f^*,\widehat F))=\sum_{v\in C}d(v,f^*)+\abs{C}\cdot d(f^*,\widehat F).
\]
To bound $d(f^*,\widehat F)$, observe that
\[
\abs{\csamples_{\leq L}\cap\samples\cap C}\cdot d(f^*,\widehat F)\leq \sum_{v\in \csamples_{\leq L}\cap\samples\cap C}(d(f^*,v)+d(v,\widehat F)).
\]
Combining the previous two inequalities, whenever $\csamples_{\leq L}\cap\samples\cap C$ is nonempty,
\begin{align}
\label{line:mfl-cluster-bound}
\expectover{F}*{\cost(C)} \leq O\left(\sum_{v\in C}d(v,f^*)+\frac{\abs{C}}{\abs{\csamples_{\leq L}\cap\samples\cap C}}\left(\sum_{v\in \csamples_{\leq L}\cap\samples\cap C}d(f^*,v)+\sum_{v\in \csamples_{\leq L}\cap\samples\cap C}d(v,\widehat F)\right)\right).
\end{align}

Call a cluster $C$ good if three conditions hold: (a) $\abs{C}\geq \frac{10k}{p}$, (b) $\abs{\samples\cap C}\geq \frac{\abs{C} p}{2}$, and (c) $\frac{L\abs{C}}{n}\geq 24\ln\abs{C}$. Call a cluster bad otherwise. We first bound the cost incurred on good clusters. Since $k\leq \frac{\abs{C} p}{10}$ for a good cluster, we have that $\abs{\Clean\cap C}\geq \abs{\samples\cap C}-k\geq \frac{\abs{C} p}{2}-\frac{\abs{C} p}{10}=\frac{2\abs{C} p}{5}$ and consequently, $\frac{L\abs{\Clean\cap C}}{np}\geq \frac{2L\abs{C}}{5n}\geq 8\ln\abs{C}.$

Applying \Cref{app-mfluniformity} with $U=\Clean\cap C$, we get that, over the randomness of the order $\sigma(\csamples)$,
\begin{align*}
\Pr\left[
    \abs{\csamples_{\leq L}\cap\samples\cap C}
    \geq \frac{1}{8}\cdot \frac{L\abs{C}}{np}
\right]&\geq 1-\frac{1}{\abs{C}}.
\intertext{Since $C$ is good, this implies}
\Pr\left[
    \abs{\csamples_{\leq L}\cap\samples\cap C}
    \geq \frac{L\abs{C}}{20n}
\right]&\geq 1-\frac{1}{\abs{C}}.    
\end{align*}
Let $\mathcal{E}_C$ denote the event
$\mathcal{E}_C :=
\left\{
\abs{\csamples_{\leq L}\cap\samples\cap C}\geq \frac{L\abs{C}}{20n}
\right\}.$

We now split according to whether $\mathcal{E}_C$ holds. Note that the event $\mathcal{E}_C$ is fully determined after fixing $\Gamma$. If $\mathcal{E}_C$ fails, then by \Cref{app-meyerson_mfl}, the expected cost incurred by $\omfl$ on $C$ is at most $\left(\sum_{v\in C}d(v,f^*)+c_{f^*}\right)\cdot O\left(\frac{\log\abs{C}}{\log\log\abs{C}}\right)$.
Since $\prob*{\neg \mathcal{E}_C}\leq \frac{1}{\abs{C}}$, the contribution of this case to the overall expected cluster cost (over the randomness of $F, \Gamma$) is at most $O\left(\sum_{v\in C}d(v,f^*)+c_{f^*}\right)$.
On the other hand, if $\mathcal{E}_C$ holds, then
$\frac{\abs{C}}{\abs{\csamples_{\leq L}\cap\samples\cap C}}\leq \frac{20n}{L}.$
Therefore, the contribution to \eqref{line:mfl-cluster-bound} in this case is at most
\[
 O\left(\sum_{v\in C}d(v,f^*)+\frac{n}{L}\left(\sum_{v\in \csamples_{\leq L}\cap\samples\cap C}d(f^*,v)+\sum_{v\in \csamples_{\leq L}\cap\samples\cap C}d(v,\widehat F)\right)\right).
\]
Summing over all good clusters and now taking expectation over $\Gamma$, we obtain that the cost incurred by $\omfl$ on good clusters is at most
\[
O\left(\sum_C\sum_{v\in C}d(v,f^*)+\expect*{\frac{n}{L}\sum_{v\in \csamples_{\leq L}\cap\samples}d(v,F^*)}+\expect*{\frac{n}{L}\sum_{v\in \csamples_{\leq L}\cap\samples}d(v,\widehat F)}+\OPT(X)\right).
\]

The first term is at most $\OPT(X)$, the second term is $O(\OPT(X))$ by \Cref{cl:mflcharging}, and the third term is $O(\alpha \cdot \OPT(X))$ by assumption. Thus, the cost incurred by $\omfl$ on good clusters is $O(\alpha \cdot OPT(X))$ in expectation.

It remains to bound the cost incurred on bad clusters which violate at least one of conditions (a), (b), or (c). If $\abs{C}<\frac{10k}{p}$, then by \Cref{app-meyerson_mfl},
$\expect*{\cost(C)}\leq \left(\sum_{v\in C}d(v,f^*)+c_{f^*}\right)\cdot O\left(\frac{\log \nicefrac{k}{p}}{\log\log \nicefrac{k}{p}}\right).$

If instead $\abs{C}\geq \frac{10k}{p}$, then since $\samples$ is a uniformly random $pn$-subset of $X$, a Chernoff bound for sampling without replacement gives
$\prob*{\abs{\samples\cap C}<\frac{\abs{C} p}{2}}\leq e^{-\abs{C} p/8}.$
By the law of total probability and \Cref{app-meyerson_mfl}, the contribution of this bad sampling event to the expected cluster cost is at most
\[e^{-\abs{C} p/8}\cdot \left(\sum_{v\in C}d(v,f^*)+c_{f^*}\right)\cdot O\left(\frac{\log\abs{C}}{\log\log\abs{C}}\right) = \left(\sum_{v\in C}d(v,f^*)+c_{f^*}\right)\cdot O\left(\frac{\log \nicefrac{k}{p}}{\log\log \nicefrac{k}{p}}\right).\] (As in the analogous calculation of \Cref{sec:metricpsample}, when $\abs{C}\geq \frac{10k}{p}$, the exponential term dominates.)

Finally, consider the case $\frac{L\abs{C}}{n}<24\ln\abs{C}$, which implies that $\abs{C} \leq O(n/L) \ln \abs{C} $. Let $\phi(x) = \nicefrac{\ln x}{\ln \ln x}$. Since $L\geq \tau$, we have
$\phi(\abs{C}) = O\left(\phi\left(\frac{n}{\max\{1,\tau\}}\right)\right).$
Therefore, by \Cref{app-meyerson_mfl} and \Cref{lem:mfl-tail}, the expected cost incurred on $C$ in this case is at most
\[
\left(\sum_{v\in C}d(v,f^*)+c_{f^*}\right)\cdot O\left(\expect*{\phi\left(\frac{n}{\max\{1,\tau\}}\right)}\right)\leq \left(\sum_{v\in C}d(v,f^*)+c_{f^*}\right)\cdot O\left(\frac{\log \nicefrac{k}{p}}{\log\log \nicefrac{k}{p}}\right).
\]

Combining the bounds for good and bad clusters, and summing over all optimal clusters, gives
\[
\sum_C\expect*{\cost(C)} = \OPT(X)\cdot O\left(\frac{\alpha \log \nicefrac{k}{p}}{\log\log \nicefrac{k}{p}}\right). \qedhere
\]
\end{proof}
We are now ready to prove our final theorem.

\begin{restatable}{thm}{mflthm}\label{thm:mfltotalcost}
The expected cost incurred by our algorithm for online metric facility location in the $(p,k)$-adversarially robust model is
$\OPT(X)\cdot O\left(\frac{\alpha \log \nicefrac{k}{p}}{\log \log \nicefrac{k}{p}}\right).$
\end{restatable}

\begin{proof}
Recall that $f(x) := \frac{\ln \nicefrac{x}{p}}{\ln \ln \nicefrac{x}{p}}$. 
Define the \emph{critical} time step
\[
t^*:=\min\{t_h\mid \cost(\omfl_{\leq t_h})\geq f(K)\cdot \beta_K(\csamples)\}.
\]
If the set is vacuous, set $t^*=\infty$. If $t^*=\infty$, then $\cost(\omfl_{\leq t_{\mathrm{fin}}})<f(K)\cdot \beta_K(\csamples)$. Since $t_{\mathrm{fin}}$ is the final doubling time, $\cost(\omfl_{\leq n})\leq 2\cost(\omfl_{\leq t_{\mathrm{fin}}})$. Therefore,
\[
\expect*{\cost(\omfl_{\leq n})}\leq 2f(K)\cdot \expect*{\beta_K(\csamples)}\leq \alpha \cdot \OPT(X)\cdot O(f(k)),
\]
where the last inequality follows from \Cref{mflbeta} and $f(K)=O(f(k))$.

Otherwise, $t^*=t_h$ for some $h$. By definition of $t^*$, $\cost(\omfl_{\leq t_{h-1}})<f(K)\cdot \beta_K(\csamples)$. Hence,
$\cost(\omfl_{\leq t^*-1})\leq 2\cost(\omfl_{\leq t_{h-1}})<2f(K)\cdot \beta_K(\csamples).$
By \Cref{mflsingleround}, the expected cost incurred by $\omfl$ at time step $t^*$ is at most $O(\OPT(X))$. By \Cref{mflbeta}, we conclude $\E[\cost(\omfl_{\leq t^*})] \leq \alpha \cdot \OPT(X)\cdot O(f(k))$.

Let $L:=L_g$ be the prefix length chosen by $\rmfl$ for $K=K_g$ at time $t^*$, and let $\widehat F$ denote the corresponding facilities opened by $\rmfl$. Since the scaled clean prefix $\csamples_{\leq \tau}$ can be served at cost $\beta_K(\csamples)\leq \frac{\cost(\omfl_{\leq t^*})}{f(K)}$, the maximality of $L$ implies $L\geq \tau$. Also, the definition of $L$ gives
$\frac{n}{L}\sum_{v\in \csamples_{\leq L}}d(v,\widehat F)\leq \frac{\cost(\omfl_{\leq t^*})}{f(K)}$ from which it follows that
$\expect*{\frac{n}{L}\sum_{v\in \csamples_{\leq L}}d(v,\widehat F)}\leq O(\alpha \cdot \OPT(X)).$
Therefore, by \Cref{mflonlinepost}, the expected cost incurred by $\omfl$ after time $t^*$ is at most $\OPT(X)\cdot O(\alpha \cdot f(k))$. Combining the cost bounds for $\omfl$ before and after $t^*$,
\[
\expect*{\cost(\omfl_{\leq n})}\leq \alpha \cdot \OPT(X)\cdot O(f(k))=\OPT(X)\cdot O\left(\frac{\alpha \log \nicefrac{k}{p}}{\log\log \nicefrac{k}{p}}\right).
\]
Finally, by \Cref{ONLINEMFLcost}, the expected cost incurred by $\rmfl$ is within a constant factor of the expected cost incurred by $\omfl$, proving the desired theorem.
\end{proof}

We end with proofs of the remaining lemmas.

\subsection{Omitted Proofs from \Cref{sec:metricpsample} and \Cref{sec:mfl-corruptions-analysis}}

\begin{lemma}\label{lem:mfl-tail}
Let $\phi(x):=\frac{\ln x}{\ln\ln x}$. Then
$\expect*{\phi\left(\frac{n}{\max\{1,\tau\}}\right)}=O(f(k)).$
\end{lemma}

\begin{proof}
Let $\rho=\frac{16K}{p}$. Recall that $\tau=\min\{\crpt-1,\frac{np}{16K}\}$. If $\crpt>\frac{np}{16K}$, then $\tau=\frac{np}{16K}$, and the contribution is at most $\phi(\rho)$. If $\crpt=1$, then $\tau=0$, and the contribution is $\phi(n)$. If $\crpt=r+1$ for some $1\leq r<\frac{np}{16K}$, then $\tau=r$, and the contribution is $\phi(n/r)$. Since the probability that any fixed position of $\sigma(\csamples)$ is corrupted is at most $\frac{K}{np}$, we have
\[
\expect*{\phi\left(\frac{n}{\max\{1,\tau\}}\right)}
\leq \phi(\rho)+\frac{K}{np}\phi(n)+\frac{K}{np}\sum_{r=1}^{\frac{np}{16K}-1}\phi\left(\frac{n}{r}\right) \leq \phi(\rho)+\frac{2K}{np}\sum_{r=1}^{\frac{np}{16K}}\phi\left(\frac{n}{r}\right).
\]

Partition the indices $r=1,\ldots,\frac{np}{16K}$ into intervals
$I_j=\left(\frac{np}{16K\cdot 2^{j+1}},\frac{np}{16K\cdot 2^j}\right].$
For every $r\in I_j$, we have $\frac{n}{r}\leq 2^{j+1}\rho$. Therefore, the contribution of $I_j$ to the summation above is at most
$\frac{2K}{np}\cdot \abs{I_j}\cdot \phi(2^{j+1}\rho)\leq 2^{1-j}\phi(2^{j+1}\rho).$
Also, $\phi(2^{j+1}\rho)\leq (j+2)\phi(\rho)$. Summing over $j\geq 0$, we obtain
\[
\sum_{j\geq 0}2^{1-j}\phi(2^{j+1}\rho)\leq \phi(\rho)\sum_{j\geq 0}2^{1-j}(j+2)=O(\phi(\rho)).
\]
Thus,
\[
\expect*{\phi\left(\frac{n}{\max\{1,\tau\}}\right)}\leq O(\phi(\rho))=O(f(K))=O(f(k)). \qedhere
\]
\end{proof}

We have the following one-step bound for Meyerson's algorithm.
\begin{claim}\label{mflsingleround}
The expected cost incurred by $\omfl$ in any single time step is at most $O(\OPT(X))$.
\end{claim}

\begin{proof}
Consider the arrival of a client $v$, and let $F_{<v}$ be the set of facilities open just before $v$ arrives. Define
\[
\kappa_v(F_{<v}) := \min_{f\in\mathcal{F}}\left(\ind\{f\notin F_{<v}\}\cdot c_f+d(f,v)\right).
\]
We show that the expected cost incurred by $\omfl$ to serve $v$, conditioned on $F_{<v}$, is at most $O(\kappa_v(F_{<v}))$. Let
\[
f'=\argmin_{f\in\mathcal{F}}\left(\ind\{f\notin F_{<v}\}\cdot c_f+d(f,v)\right).
\]
For $-1\leq j\leq q$, let $\delta_j=d(F^{(\leq j)}\cup F_{<v},v)$. For $0\leq j\leq q$, let $p_j=\min\left(\frac{\delta_{j-1}-\delta_j}{2^j},1\right)$ be the probability that $\omfl$ opens a facility of type $j$ when $v$ arrives.

First suppose $f'\in F_{<v}$. Then $\kappa_v(F_{<v})=d(f',v)$ and $\delta_{-1}=d(F_{<v},v)\leq d(f',v)=\kappa_v(F_{<v})$. The connection cost is at most $\delta_{-1}$. The expected facility-opening cost is at most
\[
\sum_{j=0}^q 2^{j+1}p_j \leq \sum_{j=0}^q 2^{j+1}\cdot \frac{\delta_{j-1}-\delta_j}{2^j}=2\sum_{j=0}^q(\delta_{j-1}-\delta_j)\leq 2\delta_{-1}\leq 2\kappa_v(F_{<v}).
\]
Thus the total expected cost is at most $3\kappa_v(F_{<v})$ in this case.

Now suppose $f'\notin F_{<v}$. Let $i$ be the type of $f'$, so $2^i\leq c_{f'}<2^{i+1}$. Then $\kappa_v(F_{<v})=c_{f'}+d(f',v)\geq 2^i+d(f',v)$, and since $f'\in F^{(\leq i)}$, we have $\delta_i\leq d(f',v)$. The expected facility-opening cost is at most
\[
\sum_{j=0}^i 2^{j+1}p_j+\sum_{j=i+1}^q2^{j+1}p_j\leq \sum_{j=0}^i2^{j+1}+\sum_{j=i+1}^q2(\delta_{j-1}-\delta_j)\leq 2^{i+2}+2\delta_i=O(\kappa_v(F_{<v})).
\]
For the connection cost, let $r=\max(\{-1\}\cup\{j\leq i\mid p_j=1\})$. If $r\geq 0$, then after processing $v$ there is an open facility at distance at most $\delta_r$ from $v$; if $r=-1$, the connection cost is at most $\delta_{-1}$. In either case, because $p_j<1$ for every $r<j\leq i$, we have $\delta_{j-1}-\delta_j<2^j$ for every $r<j\leq i$, and therefore
\[
\delta_r\leq \delta_i+\sum_{j=r+1}^i(\delta_{j-1}-\delta_j)\leq \delta_i+\sum_{j=r+1}^i2^j\leq \delta_i+2^{i+1}=O(\kappa_v(F_{<v})).
\]
Thus the total expected cost is $O(\kappa_v(F_{<v}))$ in the second case as well.

Finally, let $f^*$ be the facility serving $v$ in an optimal solution for $X$. The algorithm can always open $f^*$, if it is not already open, and connect $v$ to it. Hence $\kappa_v(F_{<v})\leq c_{f^*}+d(f^*,v)\leq \OPT(X)$, giving the desired result.
\end{proof}

\begin{lemma}\label{app-mfluniformity}
Let $Z$ be a set of size $np$, and order $Z$ in uniform random order. Fix any subset $U \subseteq Z$. For each $0 \leq a \leq np$, let
$W(a) := \abs{Z_{\leq a}\cap U}$
and
$\mu(a) :=\E[W(a)] = \frac{a\abs{U}}{np}.$
Then for any  $\abs{C} \geq 2$, with probability at least $1-1/\abs{C}$, the following statement holds simultaneously for every $a \leq np$:
\[
\mu(a)\geq 8\ln \abs{C} \implies
W(a)\geq \frac{1}{8} \mu(a).
\]
\end{lemma}

\begin{proof}
If there is no $a\leq np $ such that $\mu(a)\geq 8\ln \abs{C}$, then the statement trivially holds. Otherwise, let $a_0$ be the smallest prefix length satisfying
$\mu(a_0)\geq 8\ln \abs{C}.$
For integers $0 \leq j  \leq \floor{\log_2(np/a_0)}$, define  $a_j := 2^j a_0$,  Note that $Z_{\leq a_j}$ is a uniformly random subset of $Z$ of size $a_j$. Using Chernoff bounds for sampling without replacement,
$\Pr\left[W(a_j)<\frac{1}{4}\mu(a_j)\right]
\leq
\exp\left(-\frac{9}{32}\mu(a_j)\right).$
Since $\mu(a_j)=2^j\mu(a_0)\geq 2^j\cdot 8\ln \abs{C}$, we get
\[
\Pr\left[W(a_j)<\frac14\mu(a_j)\right]
\leq
\exp\left(-\frac{9}{32}\cdot 2^j\cdot 8\ln \abs{C}\right)
=
\abs{C}^{-(9/4)2^j}.
\]
Therefore, taking a union bound over all $0 \leq j$,
\[
\Pr\left[\exists j:\ W(a_j)<\frac{1}{4}\mu(a_j)\right]
\leq
\sum_{j\geq 0}\abs{C}^{-(9/4)2^j}
\leq
\sum_{j\geq 0}\abs{C}^{-(9/4)(j+1)}
=
\frac{\abs{C}^{-9/4}}{1-\abs{C}^{-9/4}}
\leq
\frac1{\abs{C}}.
\]

Thus, with probability at least $1-1/\abs{C}$, every $0 \leq j$ satisfies $W(a_j)\geq \frac{1}{4}\mu(a_j)$. Now fix any $a\leq np$ such that $\mu(a) \geq 8 \ln \abs{C}$, and suppose the former event occurs. By definition of $a_0$, we have $a_0 \leq a$.   Choose $j\geq 0$ such that  $a_j \leq a < 2a_j$. Note that $\mu(a) < 2\mu(a_j)$. Since $W(a)$ is monotonically increasing in $a$, we have
$W(a) \geq W(a_j) \geq \frac{1}{4} \mu(a_j) \geq \frac{1}{8} \mu(a) .$
Thus, we have that with probability at least $1-1/\abs{C}$, every $a$ such that $a\leq np$ and $\mu(a) \geq 8 \ln \abs{C}$ satisfies $W(a) \geq \frac{1}{8} \mu(a)$.
\end{proof}

\begin{lemma} \label{waiting_time_lemma}
Let $E_1, E_2, \ldots$ be a sequence of independent events, where $p_i =   \pr(E_i)$.
Let $\tau = \inf \{ i \geq 1 \mid E_i \text{ occurs} \}$ be the index of the first successful event \textup{(}setting $\tau = \infty$ if no event succeeds\textup{)}. Then, if $\alpha > 0$ is a constant, the expected sum of $\alpha  \cdot p_i$ prior to success is 
$\E\left[\sum_{i=1}^{\tau} \alpha  \cdot p_i \right] \leq \alpha.$
\end{lemma}

\begin{proof}
We rewrite the expected sum as shown below. Note that if $\tau \geq i$ and $E_i$ occurs, then $\tau = i$.
\[
\E\left[ \sum_{i=1}^{\tau} \alpha \cdot  p_i \right] = \E\left[ \sum_{i=1}^{\infty} \alpha \cdot  p_i \ind\{\tau \geq i \}  \right]  = \sum_{i=1}^{\infty} \alpha \cdot p_i \cdot  \pr[\tau \geq i ]  =  \sum_{i=1}^\infty \alpha \cdot \pr[\tau = i].
\]

Since the events $\{ \tau = i\}$ are mutually exclusive, we obtain the desired result.
\end{proof}

\meyersonguarantee*

\begin{proof}[Proof of \Cref{app-meyerson_mfl}]
Consider an optimal integral facility location solution for $X'$. Fix a facility $f^*$ of type $\ell^*$ that is opened in the optimal solution, and let $C^*$ be the set of clients in $X'$ assigned to $f^*$ in the optimal solution. Let $d^*_v = d(v,f^*)$ denote the distance from $v \in C^*$ to the center $f^*$, and let $\delta^*$ be the average optimal assignment cost of these clients, i.e., $\delta^* = \frac{\sum_{v \in C^*} d^*_v}{\abs{C^*}}$. 
We also pick naturals $a,b = \Theta\left(\frac{\log n'}{\log \log n'} \right)$  such that $a^b > n'.$

Recall that $F^{(\leq \ell)}$ is the set of facilities of type at most $\ell$, and let $F_v$ be the set of facilities opened just before $v \in C^*$ arrives. For every $\ell \in \{-1, 0, \ldots, q \}$, let $\delta_\ell(v)= d(F^{(\leq \ell)} \cup F_{v} , v)$ denote the distance from $v$ to the closest facility that is of type at most $\ell$ or is already open. When $v \in C^*$ arrives, $\ONLINEMFL$ independently opens a facility at the location of type $\ell$ (for all types $ 0 \leq \ell \leq q$, where $q$ is the maximum type) closest to $v$ with probability $\min\left(\frac{\delta_{\ell-1}(v) - \delta_\ell(v)}{2^\ell}, 1\right)$. Consequently, the expected facility cost incurred due to the arrival of $v$ is at most
\begin{equation}\label{faccosts}
F(v) = \sum_{\ell=0}^q 2^{\ell+1} \cdot \min\left(\frac{\delta_{\ell-1}(v) - \delta_\ell(v)}{2^\ell}, 1\right) \leq \sum_{\ell=0}^q 2(\delta_{\ell-1}(v) - \delta_\ell(v)) \leq 2 \delta_{-1}(v).
\end{equation}
Moreover, the expected assignment cost $A(v)$ due to $v$ is at most the sum of the expected facility cost due to $v$ and $d^*_{v}$, i.e.,
$A(v) \leq F(v) + d^*_{v}$. Indeed, let $r$ be the highest index such that $\delta_{r-1}(v) - \delta_r(v) > 2^r$, setting $r = -1$ if no such index exists. Then, we can always assign $v$ to a facility at distance at most $\delta_r(v)$. The expected facility cost incurred by $v$ due to types $\geq r +1$ is at least
\[
\sum_{\ell = r+1}^{q} 2^\ell \cdot \frac{\delta_{\ell-1}(v) - \delta_{\ell}(v)}{2^\ell} = \delta_{r}(v) - \delta_q(v) \geq \delta_r(v) - d^*_{v}.
\]
Consequently, we have that the expected cost incurred by $\ONLINEMFL$ due to $C^*$ is at most
\begin{equation} \label{asg_ag_fac}
\sum_{v\in C^*} (F(v) + A(v)) \leq \sum_{v \in C^*} \left(F(v) + d^*_v + F(v)\right) = 2 \sum_{v \in C^*}F(v) + \sum_{v \in C^*} d^*_v.
\end{equation}
We will now bound the expected facility cost incurred due to clients in $C^*$. For every $j \in [b]$, define $S_{j} = \{v \in C^*: a^{j-1} \delta^* <d^*_v \leq a^j \delta^* \}$ to be the set of clients in $C^*$ which are between $a^{j-1} \delta^*$ and $a^j \delta^*$ away from $f^*$. Also define $q_0 = \{v \in C^*: d^*_v \leq \delta^* \}$ to be the set of clients in $C^*$ that are at most $\delta^*$ away from $f^*$. Note that $\cup_{j=0}^{b} S_j = C^*$ as $    a^b \delta^* > n' \delta^* \geq d^*_v$ for any $v \in C^*$. 

Now, consider any $S_j$. If we label the clients in $S_j$ sequentially according to the order of arrival, say $v_1, v_2, \ldots , v_{\alpha}$, then the expected facility cost incurred by clients in $S_j$ due to types $\geq \ell^* -1$ is

. The expected facility cost incurred by $v_i$ due to types $\geq \ell^*+1$ is at most
\begin{equation} \label{greaterthank*}
\sum_{v_i \in S_j} \sum_{\ell = \ell^*+ 1}^q 2^{\ell+1} \cdot \frac{\delta_{\ell-1}(v_i) - \delta_\ell(v_i) }{2^\ell} \leq \sum_{v_i \in S_j} 2 \delta_{\ell^*}(v_i) \leq \sum_{v_i \in S_j} 2d^*_{v_i} \leq \sum_{v_i \in C} 2d^*_{v_i}.
\end{equation}

Let $\rho_\ell = d(F^{(\leq \ell)}, f^*)$ denote the distance from $f^*$ to the closest facility $\tilde{f}_\ell$ of type at most $\ell$. Define $B_\ell$ ($0 \leq \ell \leq \ell^*$) to be the event that $\ONLINEMFL$ opens a facility $f$ such that $d(f,f^*) \leq \rho_\ell + 2a^j \delta^*$. Also define $E$ to be the event that $\ONLINEMFL$ opens a facility $f$ such that $d(f,f^*) \leq 8 a^j \delta^*$.

Note that if any $v_i \in S_j$ opens a facility $f$ of type $\geq \ell$, then $B_\ell$ occurs. Indeed, using the triangle inequality
\[
\delta_\ell(v_i) \leq d(v_i,\tilde{f}_\ell) \leq d(v_i,f^*) + d(f^*, \tilde{f}_\ell) \leq a^j \delta^* + \rho_\ell.
\]
\[
d(f,f^*) \leq d(f,v_i) + d(v_i,f^*) \leq \delta_\ell(v_i) + a^j \delta^* \leq \rho_\ell + 2 a^j \delta^*.
\]

Note that the expected facility cost $\tilde{F}(\ell)$ incurred by clients in $S_j$ due to facilities of type $\ell$ until $B_\ell$ occurs is at most $2^{\ell+1}$.   This is because of standard waiting time arguments (see \Cref{waiting_time_lemma}).
Consequently, 
\begin{equation} \label{beforeB}
\sum_{\ell=0}^{\ell^*} \tilde{F}(\ell) \leq \sum_{\ell=0}^{\ell^*} 2^{\ell+1} \leq 2^{\ell^* +2 }.
\end{equation}

Note that once $B_{\ell^*}$ occurs, the event $E$ must also have occurred, since we opened a facility $f$ such that
\[
d(f,f^*) \leq \rho_{\ell^*} + 2a^j \delta^* = d(F^{(\ell^*)}, f^*) + 2a^j \delta^* \leq  0 + 2a^j \delta^* \leq 8a^j \delta^*.
\]
For any client $v_i \in S_j$ arriving after $B_\ell$ has occurred but before $E$ has occurred, we have
\[
\delta_{-1}(v_i) = d(F_{v_i}, v_i) \leq d(F_{v_i}, f^*) + d(f^*,v_i) \leq (\rho_\ell + 2 a^j \delta^* )+ a^j \delta^* = \rho_\ell + 3 a^j \delta^*.
\]
Moreover, there must also be an open facility of type at most $\ell$ that is at most $\delta_{\ell}(v_i) + a^j \delta^*$ away from $f^*$. Consequently, $\delta_\ell(v_i) + a^j \delta^* \geq \rho_\ell$. Since $B_\ell$ has occurred but $E$ has not, we have that $\rho_\ell + 2a^j \delta^* \geq 8 a^j \delta^*$ which implies that $\rho_\ell \geq 6a^j \delta^*$. Since the function $f(x) = \frac{x- C}{x+ 3C}$ is monotonic in $[0,\infty)$ for any $C >0$, we obtain 
\[
\frac{\delta_\ell(v_i)}{\delta_{-1} (v_i)} \geq \frac{\rho_\ell - a^j \delta^*}{\rho_\ell + 3a^j \delta^*} \geq \frac{5 a^j \delta^*}{9 a^j \delta^*} \geq \frac{1}{2}.
\] 
Thus, at least half of the expected facility cost incurred by $v_i \in S_j$ goes to types $\geq \ell+1$ after $B_\ell$ occurs. We combine this observation with \eqref{greaterthank*} and \eqref{beforeB} to conclude that the expected facility cost incurred by clients in $S_j$ before $E$ occurs is at most
\begin{equation} \label{beforeE}
\frac{1}{\frac{1}{2}}\cdot \left(2^{\ell^* + 2}+ \sum_{v_i \in S_j} 2d^*_{v_i} \right)\leq  2^{\ell^* + 3} + \sum_{v_i \in S_j} 4d^*_{v_i} \leq 8 c_{f^*} + \sum_{v_i \in S_j} 4d^*_{v_i} .
\end{equation}

Using \eqref{faccosts} and the definition of the event $E$, the expected facility cost due to $v_i \in S_j$ after $E$ has occurred is at most
\[
2\delta_{-1}(v_i) = 2d(F_{v_i} , v_i) \leq 2d(F_{v_i}, f^*) + 2d(v_i,f^*) \leq 16 a^j \delta^* + 2d^*_{v_i} \leq (16a + 2)d^*_{v_i} + 16 \delta^*.
\]
The last inequality is because $d^*_{v_i} > a^{j-1} \delta^*$ for $j \geq 1$. Consequently, the expected facility cost due to $S_j$ after $E$ has occurred is at most $\sum_{v_i \in S_j} \left((16a +  2)d^*_{v_i} + 16 \delta^*\right)$. Combining with \eqref{beforeE}, the total expected facility cost due to clients in $S_j$ is at most
\[
8c_{f^*}   + \sum_{v_i \in S_j}4 d^*_{v_i}  + \sum_{v_i \in S_j} (16 \delta^* +(16a + 2)d^*_{v_i})  = 8 c_{f^*} +  \sum_{v_i \in S_j}( 16 \delta^* +(16a + 6)d^*_{v_i})  .
\]

Adding this up for each $0 \leq j \leq b$, we conclude that the expected facility cost due to clients in $C^*$ is at most
\[
8 (b+1) c_{f^*} +  \sum_{v \in C^*} \left(16 \delta^* + (16 a + 6)d^*_v \right) \leq O(b) c_{f^*} + O(a) \sum_{v \in C^*} d^*_v.
\]
Combining this with \eqref{asg_ag_fac}, the expected cost incurred by $\ONLINEMFL$ due to clients in $X'$ assigned to $f^*$ is at most
\begin{equation} \label{clustercost}
2 \cdot \left(O(b) c_{f^*} + O(a) \sum_{v \in C^*} d^*_v\right)  + \sum_{v \in C^*} d^*_v = O(b) c_{f^*} + O(a) \sum_{v \in C^*} d^*_v.
\end{equation}

Since $a = b = \Theta\left(\frac{\log n'}{\log \log n'}\right)$, we sum \eqref{clustercost} over every facility $f^*$ that is open in the optimal solution for $X'$ to conclude that the expected cost incurred by $\ONLINEMFL$ (for a fixed input $X$) for any subset of clients $X' \subseteq X$ such that $\abs{X'} \leq n'$ is at most $O\left(\frac{\log n'}{\log \log n'}\right) \cdot \OPT(X') $, as desired.
\end{proof}

\section{Online Steiner Tree in the $(p,k)$-Adversarially Robust Model}\label{sec:steinertree}Let $(V,d)$ be a metric space with a fixed root node $v^0 := r$. Let $X = \{v^1, v^2, \ldots, v^n\} \subseteq V~$ be a set of $n$ vertices that arrive online. In the $(p,k)$-adversarially robust model, let $\samples \subseteq X$ be a uniformly random subset of size $pn$, and let $\csamples$ be the corrupted sample obtained after the adversary modifies $k$ vertices of $\samples$. The root node $r$, the metric space $(V,d)$, and $\csamples$ are given as input to the algorithm.

When a vertex $v^t \in X$ arrives at time $1 \leq t \leq n$, we must add edges so that $v^t$ is connected to the root $r$. The goal is to maintain connectivity of $X \cup \{r\}$ at minimum cost.

\subsection{The Algorithm}

The algorithm has two components, an offline randomized algorithm $\offst$ and an online algorithm $\onst$. Denote by $\SOLST$ the output of the algorithm. For an online time $t$, let $\cost(\onst_{\leq t})$ denote the cost incurred by $\onst$ up to and including time $t$, not including edges added by $\offst$.

Let $\sigma(\csamples)$ be a uniformly random permutation of $\csamples$, fixed throughout the algorithm. For any $1\leq q\leq pn$, let $\csamples_{\leq q}$ denote the first $q$ vertices in this order. The algorithm uses an offline $\alpha$-approximation to Steiner tree, where $\alpha=O(1)$, as a subroutine. The subroutine $\offst$ applies this approximation to prefixes of the corrupted sample.

\paragraph{Algorithm $\offst$:}
This algorithm takes as input the corrupted sample $\csamples$ and a budget parameter $B$. It uses the offline $\alpha$-approximation for Steiner tree to find the greatest positive integer $q$ such that the prefix $\csamples_{\leq q}$ can be connected to the root $r$ by the $\alpha$-approximation with cost at most $B$. It then adds the corresponding edges to $\SOLST$ and marks the vertices in the prefix. If no such positive integer $q$ exists, it marks no vertices and adds no edges.

\paragraph{Algorithm $\onst$:}
When $v^t \in X$ arrives, $\onst$ connects $v^t$ to the closest point in $\widehat Z \cup \{v^0, v^1, \ldots, v^{t-1}\}$, where $\widehat Z \subseteq \csamples$ is the set of vertices marked by $\offst$ so far. It then adds the corresponding edge to $\SOLST$.

The online algorithm proceeds as follows.
\begin{enumerate}[nosep]
\item On every online arrival $v^t$, run $\onst$ on $v^t$.
\item At the end of time step $1$, run $\offst$ on $\csamples$ with budget parameter $\cost(\onst_{\leq 1})$.
\item For each later time step $t$, let $t'$ be the last time at which $\offst$ was executed. If $\cost(\onst_{\leq t})>2\cost(\onst_{\leq t'})$, then run $\offst$ on $\csamples$ with budget parameter $\cost(\onst_{\leq t})$.
\end{enumerate}

\subsection{Analysis}

Algorithm $\offst$ pays at most the budget parameter $B$ it is given as input, and $\offst$ is called with parameter $\cost(\onst_{\leq t})$ exactly once every time $t$ the cost of $\onst$ doubles. Hence the offline cost is at most a constant times the online cost, which we summarize in the following observation.

\begin{observation}\label{offstcost}
$\cost(\offst) = O(\cost(\onst))$.
\end{observation}

Let $\crpt$ be the index of the first corrupted sample in the random order $\sigma(\csamples)$; equivalently, $\csamples_{\crpt}$ is the first element of $\csamples\setminus\samples$. Define $\theta:=\crpt-1$, and let $\widetilde Z:=\csamples_{\leq \theta}$. Thus, $\widetilde Z$ is the prefix of vertices in $\csamples$ appearing before the first corrupted vertex in the order $\sigma(\csamples)$.

We next show that once the clean prefix before the first corrupted sample has been marked, the remaining cost of $\onst$ is small.

\begin{lemma}\label{stonlinepostcritical}
Suppose all vertices in $\widetilde Z$ have been marked. Then the expected cost incurred by $\onst$ on any suffix of the online sequence $X$ is $\OPT(X)\cdot O\left(\log \frac{k}{p}\right)$.
\end{lemma}

\begin{proof}
Let $T^*$ be a minimum-cost tree connecting $\{r\} \cup X$, with cost $\OPT(X)$. By taking an Eulerian tour of $T^*$ and shortcutting repeated vertices, we obtain a cycle $C^*$ that visits all vertices in $\{r\} \cup X$ and has total length at most $2\OPT(X)$.

The vertices of $\widetilde Z \cup \{r\}$ partition $C^*$ into maximal paths with endpoints in $\widetilde Z \cup \{r\}$ and internal vertices in $X \setminus (\widetilde Z \cup \{r\})$. For an edge $e\in C^*$, let $P_Z(e)$ be the path in this partition containing $e$, let $L_e$ be the number of vertices on this path, and let $d_e$ be the length of $e$.

Consider one such path $P=\langle s=u_0,u_1,\ldots,u_r=s'\rangle$,
where $s,s'\in \widetilde Z\cup\{r\}$. By taking alternating edges along $P$, we partition it into two matchings, each of total length at most $\sum_{e\in P} d_e$. Between the two matchings, choose the one that covers at least half the remaining vertices. For every matched edge $(u_i,u_j)$, the later-arriving endpoint can connect greedily to the earlier-arriving endpoint, or to an even closer available point, at cost at most $d(u_i,u_j)$. This connects at least half the remaining vertices of $P$ at cost $\sum_{e\in P} d_e$. Repeating $O(\log |P|)$ times connects all vertices on the path $P$ at cost $O(\log |P|) \cdot \sum_{e\in P} d_e$.
Summing over all paths in the partition of $C^*$ shows that the expected cost incurred by $\onst$ in the suffix is at most
\begin{equation}\label{eq:interstbound}
O\left(\expect*{\sum_{e\in C^*} d_e \log L_e}\right)
=
O\left(\sum_{e\in C^*} d_e \expect*{\log L_e}\right).
\end{equation}

We now show that for every fixed edge $e\in C^*$, we can bound $\expect*{\log L_e}=O\left(\log \frac{k}{p}\right)$. Fix an edge $e=(x,y)\in C^*$ and orient the cycle. Let $D_e^+$ and $D_e^-$ be the number of vertices encountered when walking clockwise and counterclockwise, respectively, from $e$ until hitting the first vertex of $\widetilde Z\cup\{r\}$. Note that if $x,y \in \widetilde Z \cup \{r\}$, then $D_e^+=D_e^-=1$. We have $L_e\leq D_e^+ + D_e^-$. We will bound $D_e^+$; the argument for $D_e^-$ is identical.

Let $F:=\csamples\setminus\samples$ be the set of corrupted vertices, so $|F|=k$, and let $G:=\samples\cap\csamples$ denote the uncorrupted vertices in $\csamples$. Randomly permuting the vertices in $\csamples$ is equivalent to the following random process. Assign each vertex $z\in \csamples$ an independent random priority $\rho_z\sim \mathrm{Unif}[0,1]$ and order the vertices by increasing priority. Define
$\theta:=\min_{v\in F}\rho_v$ to be the priority of the first corrupted vertex, such that $\widetilde Z=\{v\in G\mid \rho_v<\theta\}$.

Now, condition on $\samples$, the adversary's choice of $F$, and the value of $\theta$. Suppose the uncorrupted sampled vertices $G$ in clockwise order from $e$ are $g_1,g_2,\ldots,g_{pn-k}$. Each $g_j$ belongs to $\widetilde Z$ independently with probability $\theta$. Let $H$ be a random variable denoting the index of the first vertex of $\widetilde Z$ in this list, setting $H=pn-k+1$ if $\widetilde Z$ is empty. Conditioned on $\theta$, $H$ is stochastically dominated by a geometric random variable with parameter $\theta$. Consequently, $\expect*{H\mid \theta}\leq \frac{1}{\theta}$.

Now, for each element of the corrupted sample $\csamples$, consider its pre-corruption counterpart in $\samples$, and color this original sampled vertex red on the cycle $C^*$. Since the adversary corrupts only $k$ vertices, the $H$-th uncorrupted sampled vertex appears no later than the $(H+k)$-th red vertex of $\samples$ in the same cyclic direction. For $q\geq 1$, let $A_q^+(e)$ denote the number of vertices encountered when walking clockwise from $e$ until seeing the $q$-th red vertex of $\samples$, setting $A_q^+(e)=n+1$ if $q>|\samples|$. Then, $D_e^+\leq A_{H+k}^+(e)$.

Since $\samples$ is a uniformly random subset of $X$ of size $pn$, the expected number of vertices encountered before the $q$-th red vertex of $\samples$ in any fixed cyclic direction is $\expectover{\samples}*{A_q^+(e)}=O\left(\nicefrac{q}{p}\right)$. Indeed, consider the cyclic order of the $n$ non-root vertices induced by $C^*$. Since $\samples$ is a uniformly random subset of these $n$ vertices of size $pn$, the expected gap between consecutive red vertices is $O\left(\frac{n}{pn}\right)=O\left(\frac{1}{p}\right)$ by symmetry. The number of vertices encountered before the $q$-th red vertex is bounded by the sum of $q$ consecutive gaps, and therefore has expectation at most $O(q/p)$. Returning to the full cycle $C^*$, the root contributes at most one additional vertex to the bound.

We now average over the randomness of $\samples$, and condition only on $\theta$ and $H$. Since $\ln x$ is concave, $\expect*{\ln W}\leq \ln \expect*{W}$ for any random variable $W$ taking values in $[1,\infty)$. Using Jensen's inequality,
\[
\begin{aligned}
\expectover{\samples}*{\ln D_e^+\mid \theta,H}
\leq
\expectover{\samples}*{\ln A_{H+k}^+(e)\mid \theta,H} 
\leq
\ln \expectover{\samples}*{A_{H+k}^+(e)\mid \theta,H}
\leq
O\left(\log \frac{k+H}{p}\right).
\end{aligned}
\]
Averaging over $H$ and applying Jensen's inequality again gives
\[
\begin{aligned}
\expectover{H}*{\ln(k+H)\mid \theta}
\leq
\ln(k+\expect*{H\mid \theta})
\leq
\ln\left(k+\frac{1}{\theta}\right)
\leq
\ln(k+1)+\ln(1/\theta).
\end{aligned}
\]
Thus, $\expect*{\ln D_e^+\mid \theta} \leq O\left(\log\nicefrac{k+1}{p}+\log\nicefrac{1}{\theta}\right)$.

Since $\E[\ln \frac{1}{\theta} ] = O(\log k)$ (it is known that the minimum of $k$ i.i.d.\ $U[0,1]$ variables follows a $\beta(1,k)$ distribution. If $\theta \sim \beta(1,k)$, then $\expect{\log \nicefrac{1}{\theta}} = \psi(k+1) - \psi(1) = H_k$, where $\psi$ is the digamma function, and $H_k$ is the $k$th harmonic number), we conclude $\expect*{\log D_e^+} = O\left(\log \nicefrac{k+1}{p}\right)$. The same argument applies for $D_e^-$, and hence $\expect*{\log L_e} = O\left(\log \nicefrac{k+1}{p}\right)$.

Substituting this bound into \eqref{eq:interstbound}, we obtain that the expected cost incurred by $\onst$ in the suffix is at most
\[
O\left(\sum_{e\in C^*} d_e \expect*{\log L_e}\right)
\leq
O\left(\sum_{e\in C^*} d_e\right)\cdot O\left(\log \frac{k+1}{p}\right)
\leq
\OPT(X)\cdot O\left(\log \frac{k}{p}\right). \qedhere
\]
\end{proof}

We finish with the proof of the main theorem for this section.

\begin{restatable}{thm}{steinerthm}\label{thm:steinercorruptiontotalcost}
The expected cost incurred by the algorithm for online Steiner tree in the $(p,k)$-adversarially robust model is $\OPT(X)\cdot O\left(\alpha\log \frac{k}{p}\right)$.
\end{restatable}

\begin{proof}
Define the \emph{critical} time step
\[
t^*:=\min\{t_h\mid \cost(\onst_{\leq t_h})\geq \alpha\OPT(X)\}.
\]
If the set is vacuous, set $t^*=\infty$. If $t^*=\infty$, then $\cost(\onst_{\leq t_{\mathrm{fin}}})<\alpha\OPT(X)$. Since $t_{\mathrm{fin}}$ is the final doubling time, $\cost(\onst_{\leq n})\leq 2\cost(\onst_{\leq t_{\mathrm{fin}}})$. Therefore, $\expect*{\cost(\onst_{\leq n})}\leq 2\alpha\OPT(X)$.
Otherwise, $t^*=t_h$ for some $h$. By definition of $t^*$, $\cost(\onst_{\leq t_{h-1}})<\alpha\OPT(X)$.
Hence, $\cost(\onst_{\leq t^*-1})\leq 2\cost(\onst_{\leq t_{h-1}})<2\alpha\OPT(X)$.

Since the cost incurred by $\onst$ in a single round is at most $\OPT(X)$, we have $\expect*{\cost(\onst_{\leq t^*})} = \expect*{\cost(\onst_{< t^*})} + O(\OPT(X)) = O(\alpha\OPT(X))$.

At time $t^*$, the algorithm runs $\offst$ with budget parameter $\cost(\onst_{\leq t^*})\geq \alpha\OPT(X)$. Since $\widetilde Z\subseteq \samples$, we have $\OPT(\widetilde Z)\leq \OPT(X)$. As $\offst$ uses an offline $\alpha$-approximation for Steiner tree, the prefix $\widetilde Z$ is feasible for budget parameter $\cost(\onst_{\leq t^*})$. Therefore, by maximality of the prefix marked by $\offst$, all vertices in $\widetilde Z$ are marked at time $t^*$.

Thus, by \Cref{stonlinepostcritical}, the expected cost incurred by $\onst$ after time $t^*$ is at most $\OPT(X)\cdot O\left(\log \nicefrac{k}{p}\right)$.
Combining the cost bounds for $\onst$ before and after $t^*$ gives $\expect*{\cost(\onst_{\leq n})}
\leq \OPT(X)\cdot O\left(\alpha\log \nicefrac{k}{p}\right)$.
Finally, by \Cref{offstcost}, the expected cost incurred by $\offst$ is within a constant factor of the expected cost incurred by $\onst$, proving the desired theorem.
\end{proof}
\section{Removing the Assumption of Known Sample Noise for Set Cover}\label{sec:app-guess-double} 
In the main body, we assumed we have an estimate $\widehat{k}$ such that $\log(k/p)\leq \log(\widehat{k}/p)\leq 2\log(k/p)$. We now explain how to remove this assumption using a guess-and-double framework. The algorithm is the same as before, except that the rounded online solution $\cC^{\rfosc}$ is now generated in phases indexed by guesses for $\OPT_{\text{LP}}(X)$ (this is the fractional set cover optimum) and $k$.

As in the main body, our algorithm uses the random-order set cover algorithm $\rosc$, the fractional online set cover algorithm $\fosc$, and the backup algorithm. The algorithm $\rosc$ has competitive ratio $\Delta$, the fractional algorithm $\fosc$ has competitive ratio $\alpha$, and the backup algorithm covers an incoming element $x$ by buying $\backup(x)$.

The online phase proceeds in sub-phases indexed by pairs $(g,\widehat{k})$, where $g$ estimates $\OPTLP(X)$ and $\widehat{k}$ estimates $k$. The sub-phases form two nested loops: $g$ remains fixed in the outer loop, and $\widehat{k}$ increases monotonically within the inner loop. We define an \emph{epoch} to be the union of all sub-phases associated with a fixed value of $g$.

We initialize $\widehat{k}=3$, and set $g$ to be the fractional optimum for the elements in $X$ seen so far. The guess $g$ is recomputed every time the fractional optimum at least doubles. When this happens, we start a new sub-phase $(g,3)$ with the new value of $g$. Within a sub-phase, we run $\fosc$ and round its fractional solution online using rounding parameter $\ln(\widehat{k}/p)$, as in \Cref{lem:sc-rounding}. Let $\cC^{\rfosc}_t$ denote the rounded online solution maintained up to time $t$, including all sets bought by the rounded online algorithm in all sub-phases so far.

We end the current sub-phase $(g,\widehat{k})$ and replace $\widehat{k}$ by $\widehat{k}^2/p$ if the cost incurred by the backup algorithm during this sub-phase exceeds $g\cdot \alpha\ln(\widehat{k}/p)$. Note that this doubles the rounding parameter, since $\ln((\widehat{k}^2/p)/p)=2\ln(\widehat{k}/p)$.

The overall algorithm proceeds as follows:
\begin{enumerate}[nosep]
    \item In the offline phase, run the $\rosc$ algorithm on the randomly permuted corrupted sample $\csamples$, without adding any sets to the solution. Let $\cC^\rosc_i$ be the solution obtained after processing the first $i$ elements of $\csamples$. Let $\ell_0=0$ and $\cC^\bk_0=\emptyset$.

    \item In the online phase, let $x$ be the element arriving at step $t$. Suppose we are in sub-phase $(g,\widehat{k})$.
    \begin{enumerate}[nosep]
        \item Apply the $\fosc$ algorithm and round the fractional solution online with the boosting parameter $\ln(\widehat{k}/p)$ to update $\cC^{\rfosc}_t$.

        \item If $x\notin \cover(\cC^{\rfosc}_t\cup \cC^\bk_{t-1}\cup \cC^\rosc_{\ell_{t-1}})$, then $\cC^\bk_t=\cC^\bk_{t-1}\cup\{\backup(x)\}$, else $\cC^\bk_t=\cC^\bk_{t-1}$.

        \item Let
        \begin{equation}\label{eq:ell-guess-double}
        \ell_t
        =
        \max\left\{
        i \mid i\geq \ell_{t-1}
        \text{ and }
        \cost(\cC^\rosc_i)\leq
        \cost(\cC^{\rfosc}_t)+\cost(\cC^\bk_t)
        \right\}.
        \end{equation}

        \item If the cost incurred by $\BACKUP$ in the current sub-phase $(g,\widehat{k})$ exceeds $g\cdot \alpha\ln(\widehat{k}/p)$, then end the current sub-phase. Replace $\widehat{k}$ by $\widehat{k}^2/p$ and start the next sub-phase with the same value of $g$. If the fractional optimum has doubled, start a new epoch with the new value of $g$ and with $\widehat{k}=3$.
    \end{enumerate}
\end{enumerate}
The solution at time $t$ is $\cC^{\rfosc}_t\cup \cC^\bk_t\cup \cC^\rosc_{\ell_t}$.

Let $\cC^{\rfosc}$, $\cC^\bk$, and $\cC^\rosc$ be the final values of $\cC^{\rfosc}_t$, $\cC^\bk_t$, and $\cC^\rosc_{\ell_t}$, respectively. The final solution is $\cC^{\rfosc}\cup \cC^\bk\cup \cC^\rosc$. By the invariant in \eqref{eq:ell-guess-double}, we always have $\cost(\cC^\rosc)\leq \cost(\cC^{\rfosc})+\cost(\cC^\bk)$.

\subsection{Analysis}

Let $\crpt$ be the index of the first corrupted sample in $\csamples$, ordered by the random permutation used by $\rosc$. The algorithm does not know $\crpt$. Let $\cC^\rosc_{\clean}:=\cC^\rosc_{\crpt-1}$ be the solution of $\rosc$ on the clean prefix, and let $\unc^\rosc=\{x\in \univ\setminus\samples\mid x\notin \cover(\cC^\rosc_{\clean})\}$ be the elements not covered by $\rosc$ on the clean prefix.

We first bound the cost incurred by the rounded online solution in terms of the cost incurred by the backup algorithm.

\begin{lemma}\label{onlinecost-guess-double}
Suppose $D$ is the expected cost incurred by $\BACKUP$ over the whole algorithm. Then, the expected cost of $\cC^{\rfosc}$ is at most $\OPT(X)\cdot O(\alpha\log(k/p))+O(D)$.
\end{lemma}

\begin{proof}
Fix a guess $g$ for $\OPTLP(X)$, and consider the epoch corresponding to $g$. The cost incurred by the fractional solution produced by $\fosc$ in any sub-phase $(g,\widehat{k})$ within this epoch is at most $g\cdot O(\alpha)$. Indeed, the fractional optimum is at most $2g$ throughout this sub-phase, and the fractional algorithm is $\alpha$-competitive. Thus, by \Cref{lem:sc-rounding}, the expected cost incurred by the rounded online solution in sub-phase $(g,\widehat{k})$ is $g\cdot O(\alpha\log(\widehat{k}/p))$. We consider two cases, based on the number of sub-phases in the epoch.

\begin{itemize}[nosep]
    \item Suppose the epoch corresponding to the guess $g$ has a single sub-phase, with $\widehat{k}=3$. The expected cost incurred by the rounded online solution within this sub-phase is $g\cdot O(\alpha\log(3/p))=g\cdot O(\alpha\log(k/p))$. Since $g$ increases geometrically across epochs, the expected cost incurred by the rounded online solution across all single-sub-phase epochs is at most $\OPT(X)\cdot O(\alpha\log(k/p))$.

    \item Suppose the epoch corresponding to the guess $g$ has at least two sub-phases. Let $\widehat{k}_{\max}$ be the maximum value of $\widehat{k}$ in this epoch. If $\widehat{k}\neq \widehat{k}_{\max}$, then the cost incurred by $\BACKUP$ in sub-phase $(g,\widehat{k})$ must exceed $g\cdot \alpha\ln(\widehat{k}/p)$, as this sub-phase is not the final sub-phase in the epoch. If $\widehat{k}=\widehat{k}_{\max}$, then the penultimate sub-phase has parameter $\widehat{k}'$ satisfying $\widehat{k}_{\max}=(\widehat{k}')^2/p$. Therefore, $\ln(\widehat{k}'/p)=\frac{1}{2}\ln(\widehat{k}_{\max}/p)$. Since the penultimate sub-phase ended, the cost incurred by $\BACKUP$ in that sub-phase must exceed $g\cdot \alpha\ln(\widehat{k}_{\max}/p)/2$. Thus, the expected cost incurred by the rounded online solution across all epochs with at least two sub-phases is at most a constant times the expected cost incurred by $\BACKUP$ over the whole algorithm.
\end{itemize}
Combining the two cases proves the lemma.
\end{proof}

We classify a sub-phase $(g,\widehat{k})$ as \emph{early} if $\widehat{k}<k$, and \emph{late} otherwise. Note that the value $\ln(\widehat{k}/p)$ doubles when we start a new sub-phase inside an epoch, since we replace $\widehat{k}$ by $\widehat{k}^2/p$.

\begin{lemma}\label{earlybackup-guess-double}
The expected cost incurred by $\BACKUP$ during early sub-phases $(\widehat{k}<k)$ is at most $\OPT(X)\cdot O(\alpha\log(k/p))$.
\end{lemma}

\begin{proof}
Consider any early sub-phase $(g,\widehat{k})$, where $\widehat{k}<k$. The cost incurred by $\BACKUP$ in this sub-phase, excluding the last time step in the sub-phase, is at most $g\cdot \alpha\log(\widehat{k}/p)$ by the sub-phase-ending rule. The cost incurred by $\BACKUP$ in a single time step is at most $2g$ within sub-phase $(g,\widehat{k})$. This is because the fractional optimum within the sub-phase does not exceed $2g$, and $\BACKUP$ buys the cheapest set to cover any element passed to it. Thus, the cost incurred by $\BACKUP$ in an early sub-phase $(g,\widehat{k})$ is at most $g\cdot O(\alpha\log(\widehat{k}/p))$. Since $g$ grows geometrically across epochs and $\ln(\widehat{k}/p)$ grows geometrically across early sub-phases within an epoch, the expected cost incurred by $\BACKUP$ across early sub-phases is at most $\OPT(X)\cdot O(\alpha\log(k/p))$.
\end{proof}

We have the following bound on the expected cost of $\BACKUP$ in late sub-phases, if the sets bought by $\rosc$ on the clean prefix of $\csamples$ have been added to the maintained solution.

\begin{lemma}\label{latebackup-guess-double}
Suppose all sets in $\cC^\rosc_{\clean}$ are already contained in the maintained solution. Then, the expected cost incurred by $\BACKUP$ during late sub-phases $(\widehat{k}\geq k)$ of any suffix of the online sequence $X$ is at most $\OPT(X) \cdot O(\Delta)$.
\end{lemma}

\begin{proof}
By assumption, all sets in $\cC^\rosc_{\clean}$ are already contained in the maintained solution throughout the suffix. Therefore, in this suffix, an element is passed to $\BACKUP$ in a late sub-phase only if it is in $\unc^\rosc$ and is also missed by the rounded online solution.

In any late sub-phase, the rounding parameter is $\ln(\widehat{k}/p)\geq \ln(k/p)$. Hence, by \Cref{lem:sc-rounding}, the probability that an incoming element is uncovered by the rounded online solution is at most $p/\widehat{k}\leq p/k$. By \Cref{lem:scwc_unc1}, if every incoming element left uncovered by the clean prefix of $\rosc$ were passed to $\BACKUP$, then the expected cost incurred by $\BACKUP$ would be $\OPT(X)\cdot O((k/p)\Delta)$. Multiplying this bound by the additional probability $p/k$ that the rounded online solution also misses the element, the expected cost incurred by $\BACKUP$ in late sub-phases during the suffix is at most $\frac{p}{k}\cdot \OPT(X)\cdot O((k/p)\Delta)=\OPT(X)\cdot O(\Delta)$.
\end{proof}

We are now ready to prove the theorem.

\begin{theorem}\label{thm:totalcost-guess-double}
The expected cost incurred by the algorithm is $\OPT(X)\cdot O(\Delta+\alpha\log(k/p))$.
\end{theorem}

\begin{proof}
We first bound the expected cost incurred by $\BACKUP$. Define the \emph{critical} time step
\begin{equation}\label{eq:setcovert*-guess-double}
t^*
:=
\min\left\{
t
\mid
\cost(\cC^{\rfosc}_t)+\cost(\cC^\bk_t)
\geq
\cost(\cC^\rosc_{\clean})
\right\}.
\end{equation}
If the set is vacuous, set $t^*=\infty$.

If $t^*=\infty$, then $\cost(\cC^\bk)\leq \cost(\cC^{\rfosc})+\cost(\cC^\bk)<\cost(\cC^\rosc_{\clean})$. Using the guarantee of $\rosc$ on the clean prefix, we get $\expect*{\cost(\cC^\bk)}\leq O(\Delta)\cdot \OPT(X)$.

Otherwise, $t^*$ is finite. By definition of $t^*$, $\cost(\cC^{\rfosc}_{t^*-1})+\cost(\cC^\bk_{t^*-1})<\cost(\cC^\rosc_{\clean})$. In particular, $\cost(\cC^\bk_{t^*-1})<\cost(\cC^\rosc_{\clean})$. If a backup set is added at time $t^*$, its cost is at most $\OPT(X)$. Therefore, the expected backup cost up to and including time $t^*$ is at most $O(\Delta)\cdot \OPT(X)$.

After time $t^*$, the invariant defining $\ell_t$ ensures that $\ell_t\geq \crpt-1$, and hence all sets in $\cC^\rosc_{\clean}$ are contained in the maintained solution. The expected backup cost during early sub-phases is at most $\OPT(X)\cdot O(\alpha\log(k/p))$ by \Cref{earlybackup-guess-double}. The expected backup cost during late sub-phases after time $t^*$ is at most $\OPT(X)\cdot O(\Delta)$ by \Cref{latebackup-guess-double}. Thus, $\expect*{\cost(\cC^\bk)}\leq \OPT(X)\cdot O(\Delta+\alpha\log(k/p))$.

Next, by \Cref{onlinecost-guess-double}, the expected cost incurred by the rounded online solution is at most $\OPT(X)\cdot O(\alpha\log(k/p))+O(\expect*{\cost(\cC^\bk)})\leq \OPT(X)\cdot O(\Delta+\alpha\log(k/p))$.

Finally, by the invariant in the definition of $\ell_t$, the cost of the sets added from the simulated $\rosc$ run is at most $\cost(\cC^{\rfosc})+\cost(\cC^\bk)$. Therefore, the total expected cost incurred by the algorithm is $\OPT(X)\cdot O(\Delta+\alpha\log(k/p))$.
\end{proof}

\section{Online Covering Integer Programs}\label{sec:cip}Let $X=\{a^1,a^2,\ldots,a^n\}$ be a collection of $n$ covering constraints over $m$ columns. Each column $j\in[m]$ is equipped with a nonnegative cost $c_j$. Each constraint $a\in X$ is a vector in $[0,1]^m$. For an integral solution $x\in\mathbb{Z}_{\geq 0}^m$, we say that $x$ satisfies $a$ if $\inp{a,x}\geq 1$, and define $\cost(x)=\sum_{j=1}^m c_jx_j$. The goal in covering integer programming is to find an integral solution $x\in\mathbb{Z}_{\geq 0}^m$ satisfying every constraint in $X$ while minimizing $\cost(x)$.

For a set of constraints $Y\subseteq X$, let $\OPT(Y)$ denote the minimum cost integral solution satisfying all constraints in $Y$, and let $\OPTLP(Y)$ denote the corresponding fractional optimum.

In the \emph{online} version, $X$ is initially unknown and its constraints are revealed one at a time. When a constraint $a$ arrives, the algorithm must monotonically increase its current integral solution so that $a$ is satisfied. The goal is to minimize the total cost of the final integral solution.

In the \emph{$p$-sample model}, the algorithm samples a set $\samples$ of constraints uniformly at random from all subsets of size $s=pn$, before the online sequence begins. Then the online input $X$ arrives in arbitrary order.

In the \emph{$(p,k)$-adversarially robust model}, the adversary is allowed to modify any $k$ constraints in $\samples$ after they are sampled, but before they are considered by the algorithm; we use $\csamples$ to denote the corrupted sample.

For the algorithms below, we use the following backup notation. For a current integral solution $x$ and a constraint $a$ not satisfied by $x$, let $\backargs{a}{x}$ denote a minimum-cost nonnegative integral vector $z$ such that $\inp{a,x+z}\geq 1$. If $a$ is already satisfied by $x$, we define $\backargs{a}{x}=0$. We use the monotonicity property that if $x'\geq x$ coordinate-wise, then $\cost(\backargs{a}{x'})\leq \cost(\backargs{a}{x})$.

We will also use that for every constraint $a$ and every fractional or integral feasible solution $x^*$ satisfying $a$, we have $\cost(\backargs{a}{0})\leq O(\cost(x^*))$. Indeed, since $\inp{a,x^*}\geq 1$, some coordinate $j$ with $a_j>0$ has $c_j/a_j\leq \cost(x^*)$. Buying $\ceil{1/a_j}$ copies of column $j$ gives cost at most $O(\cost(x^*))$.

We use the following rounding theorem for covering integer programs, which appears in \cite{gupta14}.

\begin{lemma}[CIP rounding {\cite[Section 4.2]{gupta14}}]
\label{lem:cip-rounding}
Let $\cA$ be a fractional online covering integer programming algorithm that takes a constraint sequence $a^1,a^2,\ldots$ as input and produces coordinate-wise nondecreasing fractional solutions $x^1\leq x^2\leq\ldots$, with competitive ratio $\alpha$. There is an online rounding algorithm $\mathsf{Round}$ that takes $x^1\leq x^2\leq\ldots$ as input and a parameter $\beta\in(0,1]$, and produces coordinate-wise nondecreasing integral solutions $z^1\leq z^2\leq\ldots$ with the following properties: for every step $t$, $\expect*{\cost(z^t)}\leq O(\alpha\log\nicefrac1\beta)\cdot \OPTLP(\{a^1,\ldots,a^t\})$, and for every $t'\leq t$, $\Pr[\inp{a^{t'},z^t}<1]\leq \beta$.
\end{lemma}

We refer to the not necessarily feasible integral online covering integer program algorithm $\overline{\cA}=\mathsf{Round}(\cA,\beta)$ as \emph{rounding} $\cA$ with \emph{boosting parameter} $\log\nicefrac1\beta$.

\subsection{$p$-Sample Model}
\label{sec:cipsample}

As a warm up, we begin with the no-corruptions case. This algorithm assumes access to a random-order online covering integer programming algorithm $\rocip$, with competitive ratio $\Delta$, and a fractional online covering integer programming algorithm $\fcip$, with competitive ratio $\alpha$.

The algorithm has the following phases.
\begin{enumerate}[nosep]
    \item In the offline phase, run $\rocip$ on the randomly permuted sample $\samples$. Let $x^\rocip_\samples$ be the integral solution obtained. Let $x^{\rfcip}_0=x^\bk_0=0$.

    \item In the online phase, let $a$ be the constraint arriving at step $t$.
    \begin{enumerate}[nosep]
        \item If $\inp{a,x^\rocip_\samples}\geq 1$, then set $x^{\rfcip}_t=x^{\rfcip}_{t-1}$ and $x^\bk_t=x^\bk_{t-1}$, and skip.

        \item Apply $\fcip$ and round the fractional solution online with boosting parameter $\log\nicefrac1p$, using \Cref{lem:cip-rounding}, to obtain an integral solution $x^{\rfcip}_t$.

        \item If $\inp{a,x^\rocip_\samples+x^{\rfcip}_t+x^\bk_{t-1}}<1$, then set $x^\bk_t=x^\bk_{t-1}+\backargs{a}{x^\rocip_\samples+x^{\rfcip}_t+x^\bk_{t-1}}$. Otherwise, set $x^\bk_t=x^\bk_{t-1}$.
    \end{enumerate}
\end{enumerate}
The solution at time $t$ is $x^\rocip_\samples+x^{\rfcip}_t+x^\bk_t$. Let $x^{\rfcip}$ and $x^\bk$ denote the final values of $x^{\rfcip}_t$ and $x^\bk_t$, respectively. The final solution is $x^\rocip_\samples+x^{\rfcip}+x^\bk$.

\subsubsection{Analysis}

To analyze the algorithm, we first prove the analog of \Cref{lem:backup_costs}. Let $\unc^\rocip=\{a\in X\setminus\samples\mid \inp{a,x^\rocip_\samples}<1\}$ be the constraints not satisfied by the random-order solution on the sample.

\begin{lemma}
\label{lem:cip-backup-costs}
It holds that
$\expect*{\sum_{a\in \unc^\rocip}\cost(\backargs{a}{x^\rocip_\samples})}
\leq
O\left(\frac{\Delta}{p}\right)\cdot \OPT(X)$.
\end{lemma}

\begin{proof}
Let $\pi=(\pi_1,\ldots,\pi_n)$ be a uniformly chosen random permutation of $X$, and let $\samples$ denote the first $s=pn$ constraints according to $\pi$. Since the constraints in $\samples$ are the input to $\rocip$, let $x^\rocip_i$ be its output on the prefix $(\pi_1,\ldots,\pi_i)$. Thus, $x^\rocip_\samples=x^\rocip_s$.

For simplicity, let $b_x(a)=\cost(\backargs{a}{x})$. We need to upper bound $\expect*{\sum_{a\in\unc^\rocip}\cost(\backargs{a}{x^\rocip_\samples})}=\expect*{\sum_{i=s+1}^n b_{x^\rocip_s}(\pi_i)}$. By monotonicity, if $j\leq s$, then $b_{x^\rocip_s}(\pi_i)\leq b_{x^\rocip_{j-1}}(\pi_i)$. Averaging this inequality over $j=1,\ldots,s$, taking expectations, and summing over $i=s+1,\ldots,n$, we obtain
\begin{equation}
\label{eq:cip-psample-average}
\expect*{\sum_{i=s+1}^n b_{x^\rocip_s}(\pi_i)}
\leq
\sum_{i=s+1}^n \frac{1}{s}\sum_{j=1}^s \expect*{b_{x^\rocip_{j-1}}(\pi_i)}.
\end{equation}
We claim that $\expect*{b_{x^\rocip_{j-1}}(\pi_i)}=\expect*{b_{x^\rocip_{j-1}}(\pi_j)}$. Indeed, conditioned on $(\pi_1,\ldots,\pi_{j-1})$ and the randomness of $\rocip$, the solution $x^\rocip_{j-1}$ is fixed, and $\pi_i$ and $\pi_j$ are both uniformly distributed over $X\setminus\{\pi_1,\ldots,\pi_{j-1}\}$. Also, the cost incurred by $\rocip$ in round $j$ is at least $b_{x^\rocip_{j-1}}(\pi_j)$, since $\rocip$ must increase its solution enough to satisfy $\pi_j$. Applying these observations to \eqref{eq:cip-psample-average}, we obtain
\begin{align*}
\expect*{\sum_{i=s+1}^n b_{x^\rocip_s}(\pi_i)}
&\leq
\frac{n-s}{s}\cdot \sum_{j=1}^s \expect*{b_{x^\rocip_{j-1}}(\pi_j)} \\
&\leq
\frac{1}{p}\cdot \expect*{\cost(x^\rocip_\samples)}
\leq
\frac{\Delta}{p}\cdot \expect*{\OPT(\samples)}
\leq
O\left(\frac{\Delta}{p}\right)\cdot \OPT(X).
\qedhere
\end{align*}
\end{proof}

\begin{theorem}
\label{thm:cip-psample}
If $\rocip$ is a random-order online covering integer programming algorithm with competitive ratio $\Delta$ and $\fcip$ is a fractional online covering integer programming algorithm with competitive ratio $\alpha$, then the algorithm above has expected competitive ratio $O(\Delta+\alpha\log\nicefrac1p)$.
\end{theorem}

\begin{proof}
We bound the expected cost of $x^\rocip_\samples+x^{\rfcip}+x^\bk$. The expected cost of $x^\rocip_\samples$ is at most $\Delta\cdot \expect*{\OPT(\samples)}\leq \Delta\cdot \OPT(X)$. By the guarantee of $\fcip$ and \Cref{lem:cip-rounding}, the expected cost of $x^{\rfcip}$ is at most $O(\alpha\log\nicefrac1p)\cdot \OPT(X)$.

It remains to bound the backup cost. A constraint can contribute to $x^\bk$ only if it lies in $\unc^\rocip$ and is missed by the rounded online solution. By \Cref{lem:cip-rounding}, each such constraint is missed by the rounded online solution with probability at most $p$. By monotonicity, the backup cost paid on such a constraint is at most $\cost(\backargs{a}{x^\rocip_\samples})$. Therefore, by \Cref{lem:cip-backup-costs}, $\expect*{\cost(x^\bk)}\leq p\cdot \expect*{\sum_{a\in\unc^\rocip}\cost(\backargs{a}{x^\rocip_\samples})}\leq O(\Delta)\cdot \OPT(X)$. Summing the three expected costs completes the proof.
\end{proof}
We obtain an expected $O(\log \nicefrac{1}{p} \cdot \log m  + \log n)$-competitive ratio by instantiating $\rocip$ to be the $O(\log (mn))$-competitive LearnOrCoverCIP algorithm~\cite{KL21} and $\fcip$ to be the $O(\log m)$-competitive online fractional primal-dual algorithm~\cite{buchbinder2009online}. 

\subsection{$(p,k)$-Adversarially Robust Model}
\label{sec:cip-corruptions}

We now turn to the model in which the adversary is allowed to corrupt $k$ of the $s=pn$ samples; let $\csamples$ be the resulting set. As in the set cover case, the main idea is to use only a carefully chosen prefix of the corrupted sample. 

The algorithm uses $\rocip$, $\fcip$, and the backup algorithm defined previously. The rounded online part proceeds in sub-phases indexed by pairs $(g,\widehat{k})$, where $g$ estimates $\OPTLP(X)$ and $\widehat{k}$ estimates $k$. The sub-phases form two nested loops: $g$ remains fixed in the outer loop, and $\widehat{k}$ increases monotonically within the inner loop. We define an \emph{epoch} to be the union of all sub-phases associated with a fixed value of $g$.

We initialize $\widehat{k}=3$, and set $g$ to be the fractional optimum for the constraints in $X$ seen so far. The guess $g$ is recomputed every time the fractional optimum at least doubles. When this happens, we start a new sub-phase $(g,3)$ with the new value of $g$. Within a sub-phase, we run $\fcip$ and round its fractional solution online using rounding parameter $\log(\widehat{k}/p)$, as in \Cref{lem:cip-rounding}. Let $x^{\rfcip}_t$ denote the rounded online solution maintained up to time $t$, including all variables bought by the rounded online algorithm in all sub-phases so far.

We end the current sub-phase $(g,\widehat{k})$ if $\widehat{k}<s$ and the cost incurred by the backup algorithm during this sub-phase exceeds $g\cdot \alpha\ln(\widehat{k}/p)$. When this happens, we replace $\widehat{k}$ by $\min\{s,\widehat{k}^2/p\}$ and start the next sub-phase with the same value of $g$. If the fractional optimum has doubled, we start a new epoch with the new value of $g$ and with $\widehat{k}=3$.

The overall algorithm proceeds as follows.
\begin{enumerate}[nosep]
    \item In the offline phase, run $\rocip$ on the randomly permuted corrupted sample $\csamples$, without initially adding its solution to the maintained solution. Let $x^\rocip_i$ be the solution obtained after processing the first $i$ constraints of $\csamples$. Let $\ell_0=0$ and $x^\bk_0=0$.

    \item In the online phase, let $a$ be the constraint arriving at step $t$. Suppose we are in sub-phase $(g,\widehat{k})$.
    \begin{enumerate}[nosep]
        \item Apply $\fcip$ and round the fractional solution online with boosting parameter $\ln(\widehat{k}/p)$ to update $x^{\rfcip}_t$.

        \item If $\inp{a,x^{\rfcip}_t+x^\bk_{t-1}+x^\rocip_{\ell_{t-1}}}<1$, then set $x^\bk_t=x^\bk_{t-1}+\backargs{a}{x^{\rfcip}_t+x^\bk_{t-1}+x^\rocip_{\ell_{t-1}}}$; otherwise, set $x^\bk_t=x^\bk_{t-1}$.

        \item Let $\ell_t=\max\{i\mid i\geq \ell_{t-1}\text{ and }\cost(x^\rocip_i)\leq \cost(x^{\rfcip}_t)+\cost(x^\bk_t)\}$.
    \end{enumerate}
\end{enumerate}
The solution at time $t$ is $x^{\rfcip}_t+x^\bk_t+x^\rocip_{\ell_t}$. Let $x^{\rfcip}$, $x^\bk$, and $x^\rocip$ be the final values of $x^{\rfcip}_t$, $x^\bk_t$, and $x^\rocip_{\ell_t}$, respectively. The final solution is $x^{\rfcip}+x^\bk+x^\rocip$. By the invariant in the definition of $\ell_t$, we always have $\cost(x^\rocip)\leq \cost(x^{\rfcip})+\cost(x^\bk)$.

\subsubsection{Analysis}

Let $\crpt$ be the index of the first corrupted sample in $\csamples$, ordered by the random permutation used by $\rocip$. The algorithm does not know $\crpt$. Let $x^\rocip_{\clean}:=x^\rocip_{\crpt-1}$ be the solution of $\rocip$ on the clean prefix, and let $\unc^\rocip=\{a\in X\setminus\samples\mid \inp{a,x^\rocip_{\clean}}<1\}$ be the constraints not satisfied by $\rocip$ on the clean prefix.

For a prefix $P$ of $\csamples$, let $\event(P)$ denote the event ``$\csamples_{<|P|+1}=P$ and $\crpt>|P|+1$,'' i.e., the first $|P|$ positions are $P$ and the next position is uncorrupted. For a subset $P\subseteq\csamples$, define
\[
\skewed(P) := \left\{a \in X \setminus P 
\ \ \middle| \ \ \prob*{\csamples_{\abs{P}+1} = a  | \event(P)} 
< \frac{1}{2(n-\abs{P})} \right\},
\]
to be the set of constraints whose chance of appearing next is too small compared with the uniform benchmark, given that the prefix $P$ seen so far is uncorrupted.  Without corruptions, each of the remaining constraints in $\univ \setminus P$\, would be equally likely to appear next in random order. 

Recall the following near-uniformity lemma, proven in \Cref{sec:setcover}.
\uniformity*
We first prove the robust analog of \Cref{lem:cip-backup-costs}.

\begin{lemma}\label{lem:cipwc_unc}
It holds that
$\expect*{\sum_{a\in \unc^\rocip}\cost(\backargs{a}{x^\rocip_{\clean}})}
\leq
O\left(\frac{k\Delta}{p}\right)\cdot \OPT(X)$.
\end{lemma}

\begin{proof}
Let $T=s/(4k)$, and assume $T$ is a nonzero integer for simplicity. For a solution $x$, let $b_x(a)=\cost(\backargs{a}{x})$. Let $B_i=\sum_{a\in X} b_{x^\rocip_i}(a)$. By monotonicity, $B_i\leq B_{i-1}$. Since every constraint in $\unc^\rocip$ is unsatisfied by $x^\rocip_{\clean}=x^\rocip_{\crpt-1}$, we have $\sum_{a\in\unc^\rocip}\cost(\backargs{a}{x^\rocip_{\clean}})\leq B_{\crpt-1}$.

We claim that
\begin{eqnarray}
\label{eqn:cip-bkupcost}
\expect*{B_{\crpt-1}}
&\leq&
O\left(\frac{k}{p}\right)\cdot \OPT(X)
+
O\left(\frac{k}{s}\right)
\sum_{i=1}^T
\prob*{\crpt>i}\cdot
\expect*{B_{i-1}\mid \crpt>i}.
\end{eqnarray}
We omit this proof as it is identical to the proof of \eqref{eqn:bkupcost}, and follows by a case-analysis on where the first corruption appears.

We now analyze $B_{i-1}$ for $i\leq T$. Condition on $\event(P)$, where $|P|=i-1$. If a constraint is not in $\skewed(P)$, then by definition its conditional probability of appearing at position $i$ in $\csamples$ is at least $1/(2(n-|P|))\geq 1/(2n)$. Also, $b_{x^\rocip_{i-1}}(a)=0$ for every $a\in P$, since $\rocip$ has already satisfied every constraint in $P$. Therefore,
\begin{equation}
\label{eq:cip-B-prefix}
B_{i-1} = \sum_{a\in X} b_{x^\rocip_{i-1}}(a)
\leq
2n\cdot \expect*{b_{x^\rocip_{i-1}}(\csamples_i)\mid \event(P)}
+
\abs{\skewed(P)}\cdot O(\OPT(X)).
\end{equation}
Taking expectations, summing over $i=1,\ldots,T$, and using \Cref{lem:uniformitybound}, we obtain
\begin{align*}
&\sum_{i=1}^T \prob*{\crpt>i}\cdot \expect*{B_{i-1}\mid \crpt>i} \\
&\leq
2n\expect*{\sum_{i=1}^T \ind\{\crpt>i\}\cdot b_{x^\rocip_{i-1}}(\csamples_i)}
+
O(\OPT(X))\sum_{i=1}^T \prob*{\crpt>i}\cdot
\expect*{\abs{\skewed(\csamples_{<i})}\mid \crpt>i} \\
&\leq
2n\expect*{\cost(x^\rocip_{\clean})}
+
T\cdot O\left(\frac{k}{p}\right)\cdot \OPT(X) \\
&\leq
O(n\Delta)\cdot \OPT(X).
\end{align*}
In the last inequality, we used the random-order guarantee of $\rocip$ on the clean prefix. Substituting this bound into \eqref{eqn:cip-bkupcost} completes the proof.
\end{proof}

We now bound the cost incurred by the rounded online solution in terms of the cost incurred by the backup algorithm.

\begin{lemma}\label{cip-onlinecost-guess-double}
Suppose $D$ is the expected cost incurred by the backup algorithm over the whole algorithm. Then, the expected cost of $x^{\rfcip}$ is at most $\OPT(X)\cdot O(\alpha\log(k/p))+O(D)$.
\end{lemma}

\begin{proof}
Fix a guess $g$ for $\OPTLP(X)$, and consider the epoch corresponding to $g$. The cost incurred by the fractional solution produced by $\fcip$ in any sub-phase $(g,\widehat{k})$ within this epoch is at most $g\cdot O(\alpha)$. Indeed, the fractional optimum is at most $2g$ throughout this sub-phase, and $\fcip$ is $\alpha$-competitive. Thus, by \Cref{lem:cip-rounding}, the expected cost incurred by the rounded online solution in sub-phase $(g,\widehat{k})$ is $g\cdot O(\alpha\log(\widehat{k}/p))$.

If the epoch corresponding to $g$ has a single sub-phase, then $\widehat{k}=3$, and the expected cost incurred by the rounded online solution in this sub-phase is $g\cdot O(\alpha\log(k/p))$. Since $g$ grows geometrically across epochs and $\OPTLP(X)\leq \OPT(X)$, the expected cost incurred by the rounded online solution across all such epochs is at most $\OPT(X)\cdot O(\alpha\log(k/p))$.

Now suppose the epoch corresponding to $g$ has at least two sub-phases. Let $\widehat{k}_{\max}$ be the maximum value of $\widehat{k}$ in this epoch. For every non-final sub-phase $(g,\widehat{k})$, the cost incurred by the backup algorithm in that sub-phase exceeds $g\cdot \alpha\ln(\widehat{k}/p)$. For the final sub-phase, let $\widehat{k}'$ be the value of $\widehat{k}$ in the penultimate sub-phase. Since the penultimate sub-phase ended and $\widehat{k}_{\max}\leq (\widehat{k}')^2/p$, we have $\ln(\widehat{k}_{\max}/p)\leq 2\ln(\widehat{k}'/p)$. Hence, the cost of the final sub-phase is also bounded by a constant times the backup cost incurred in the penultimate sub-phase. Therefore, the expected cost incurred by the rounded online solution across all epochs with at least two sub-phases is $O(D)$.
\end{proof}

We classify a sub-phase $(g,\widehat{k})$ as early if $\widehat{k}<k$, and late otherwise.

\begin{lemma}\label{cip-earlybackup-guess-double}
The expected cost incurred by the backup algorithm during early sub-phases $(\widehat{k}<k)$ is at most $\OPT(X)\cdot O(\alpha\log(k/p))$.
\end{lemma}

\begin{proof}
Consider any early sub-phase $(g,\widehat{k})$. The backup cost in this sub-phase, excluding the last time step, is at most $g\cdot \alpha\ln(\widehat{k}/p)$ by the sub-phase-ending rule. The backup cost in a single time step is at most $O(g)$ within this sub-phase. Indeed, the fractional optimum on the constraints seen so far within the sub-phase is at most $2g$, and the cost of satisfying a single constraint by the backup algorithm is at most a constant times this fractional optimum. Therefore, the backup cost in an early sub-phase $(g,\widehat{k})$ is at most $g\cdot O(\alpha\log(\widehat{k}/p))$. Since $g$ grows geometrically across epochs and $\ln(\widehat{k}/p)$ grows geometrically across early sub-phases within an epoch, the total expected cost incurred by the backup algorithm during early sub-phases is at most $\OPT(X)\cdot O(\alpha\log(k/p))$.
\end{proof}

\begin{lemma}\label{cip-latebackup-guess-double}
Suppose the initial solution is at least $x^\rocip_{\clean}$ coordinate-wise. Then, the expected cost incurred by the backup algorithm during late sub-phases $(\widehat{k}\geq k)$ of any suffix of the online sequence $X$ is at most $\OPT(X)\cdot O(\Delta)$.
\end{lemma}

\begin{proof}
By assumption, the initial solution is at least $x^\rocip_{\clean}$ coordinate-wise. Therefore, in this suffix, a constraint is passed to the backup algorithm in a late sub-phase only if it is in $\unc^\rocip$ and is also missed by the rounded online solution.

In any late sub-phase, the rounding parameter is at least $\ln(k/p)$. Hence, by \Cref{lem:cip-rounding}, the probability that an incoming constraint is not satisfied by the rounded online solution is at most $p/k$. By \Cref{lem:cipwc_unc}, if every incoming constraint left unsatisfied by $x^\rocip_{\clean}$ were passed to the backup algorithm, then the expected cost incurred would be $\OPT(X)\cdot O((k/p)\Delta)$. Multiplying by the additional probability $p/k$ that the rounded online solution also misses the constraint, the expected cost incurred by the backup algorithm in late sub-phases during the suffix is at most $\OPT(X)\cdot O(\Delta)$.
\end{proof}

\begin{theorem}\label{thm:cip-corruptions}
If $\rocip$ is a random-order online covering integer programming algorithm with competitive ratio $\Delta$ and $\fcip$ is a fractional online covering integer programming algorithm with competitive ratio $\alpha$, then the algorithm above has expected competitive ratio $O(\Delta+\alpha\log(k/p))$.
\end{theorem}

\begin{proof}
We first bound the expected cost incurred by the backup algorithm. Define the critical time step $t^*:=\min\{t\mid \cost(x^{\rfcip}_t)+\cost(x^\bk_t)\geq \cost(x^\rocip_{\clean})\}$. If the set is vacuous, set $t^*=\infty$.

If $t^*=\infty$, then $\cost(x^\bk)\leq \cost(x^{\rfcip})+\cost(x^\bk)<\cost(x^\rocip_{\clean})$. Using the guarantee of $\rocip$ on the clean prefix, we get $\expect*{\cost(x^\bk)}\leq O(\Delta)\cdot \OPT(X)$.

Otherwise, $t^*$ is finite. By definition of $t^*$, $\cost(x^{\rfcip}_{t^*-1})+\cost(x^\bk_{t^*-1})<\cost(x^\rocip_{\clean})$. In particular, $\cost(x^\bk_{t^*-1})<\cost(x^\rocip_{\clean})$. The expected cost incurred by $\BACKUP$ in time step $t^*$ is at most $O(\OPT(X))$. Therefore, the expected backup cost up to and including time $t^*$ is at most $O(\Delta)\cdot \OPT(X)$.

After time $t^*$, the invariant defining $\ell_t$ ensures that $\ell_t\geq \crpt-1$, and hence the current solution is at least $x^\rocip_{\clean}$ coordinate-wise. The expected backup cost during early sub-phases is at most $\OPT(X)\cdot O(\alpha\log(k/p))$ by \Cref{cip-earlybackup-guess-double}. The expected backup cost during late sub-phases after time $t^*$ is at most $\OPT(X)\cdot O(\Delta)$ by \Cref{cip-latebackup-guess-double}. Thus, $\expect*{\cost(x^\bk)}\leq \OPT(X)\cdot O(\Delta+\alpha\log(k/p))$.

Next, by \Cref{cip-onlinecost-guess-double}, the expected cost incurred by the rounded online solution is at most $\OPT(X)\cdot O(\alpha\log(k/p))+O(\expect*{\cost(x^\bk)})\leq \OPT(X)\cdot O(\Delta+\alpha\log(k/p))$.

Finally, by the invariant in the definition of $\ell_t$, the cost of the columns added from the simulated $\rocip$ run is at most $\cost(x^{\rfcip})+\cost(x^\bk)$. Therefore, the total expected cost incurred by the algorithm is $\OPT(X)\cdot O(\Delta+\alpha\log(k/p))$.
\end{proof}
We obtain an expected $O(\log \nicefrac{k}{p} \cdot \log m  + \log n)$-competitive ratio by instantiating $\rocip$ to be the $O(\log (mn))$-competitive LearnOrCoverCIP algorithm~\cite{KL21} and $\fcip$ to be the $O(\log m)$-competitive online fractional primal-dual algorithm~\cite{buchbinder2009online}. 
\section{Online Non-Metric Facility Location}\label{sec:nmfl}Let $\mathcal{F}$ be a collection of $m$ facilities and let $X=\{v^1,v^2,\ldots,v^n\}$ be a collection of $n$ clients. Each facility $f\in\mathcal{F}$ is equipped with a nonnegative opening cost $c_f$, and each facility-client pair $(f,v)$ has a nonnegative connection cost $d(f,v)$. A solution consists of a set of open facilities and an assignment of each  client to an open facility. Its cost is the sum of the opening costs of the open facilities and the connection costs of the assigned clients. For a set of clients $Y\subseteq X$, let $\OPT(Y)$ denote the minimum cost integral solution serving all clients in $Y$, and let $\OPTLP(Y)$ denote the corresponding fractional optimum.

In the \emph{online} version, the set $X$ is initially unknown and its clients are revealed one at a time. When a client $v$ arrives, the algorithm must connect $v$ to an open facility, possibly after opening additional facilities. The goal is to minimize the total cost of the final solution.

In the \emph{$p$-sample model}, the algorithm samples a set $\samples$ of clients uniformly at random from all subsets of size $s=pn$, before the online sequence begins. Then the online input consists of $X$ in arbitrary order.

In the \emph{$(p,k)$-adversarially robust model}, the adversary is allowed to modify any $k$ clients in $\samples$ after they are sampled, but before they are considered by the algorithm; we use $\csamples$ to denote the corrupted sample.

For the algorithms below, we use the following notation for solution states. A solution state $\SOL$ consists of a set $F(\SOL)\subseteq \mathcal{F}$ of open facilities and a set $C(\SOL)$ of clients already connected by $\SOL$. We write $\cost(\SOL)$ for the opening cost of $F(\SOL)$ plus the connection cost of the clients in $C(\SOL)$. For two solution states $\SOL$ and $\SOL'$, we write $\SOL\cup \SOL'$ for the state obtained by taking the union of the opened facilities and the union of the client connections, picking an arbitrary connection for a client if it appears in both states. In particular, $\cost(\SOL\cup\SOL')\leq \cost(\SOL)+\cost(\SOL')$.

For a current solution state $\SOL$ and a client $v$, let $\backargs{v}{\SOL}$ denote the cheapest way to serve $v$ starting from $\SOL$. If $v\in C(\SOL)$, then $\backargs{v}{\SOL}=0$. Otherwise, $\backargs{v}{\SOL}$ either connects $v$ to an already open facility in $F(\SOL)$, or opens a new facility and connects $v$ to it, whichever is cheaper. Thus,
\[
\cost(\backargs{v}{\SOL})
=
\begin{cases}
0, & v\in C(\SOL),\\
\min\left\{
\min_{f\in F(\SOL)} d(f,v),
\min_{f\in\mathcal{F}}(c_f+d(f,v))
\right\}, & v\notin C(\SOL).
\end{cases}
\]
We use the monotonicity property that if $\SOL'\supseteq \SOL$, then $\cost(\backargs{v}{\SOL'})\leq \cost(\backargs{v}{\SOL})$.

We will also use that for every client $v$ and every fractional or integral feasible solution $\SOL^*$ serving $v$, we have $\cost(\backargs{v}{\emptyset})\leq O(\cost(\SOL^*))$.

We use the following rounding theorem for online non-metric facility location.

\begin{lemma}[NMFL rounding~{\cite[Section 4.1]{nonmetricfacility}}]
\label{lem:nmfl-rounding}
Let $\cA$ be a fractional online non-metric facility location algorithm that takes a client sequence $v^1,v^2,\ldots$ as input and produces a coordinate-wise nondecreasing fractional solution, with competitive ratio $\alpha$. There is an online rounding algorithm $\mathsf{Round}$ that takes this fractional solution as input and a parameter $\beta\in(0,1]$, and produces monotone integral solution states $\SOL_1\subseteq \SOL_2\subseteq\cdots$ with the following properties: for every step $t$, $\expect*{\cost(\SOL_t)}\leq O(\alpha\log\nicefrac1\beta)\cdot \OPTLP(\{v^1,\ldots,v^t\})$, and for every $t'\leq t$, $\Pr[v^{t'}\notin C(\SOL_t)]\leq \beta$.
\end{lemma}

We refer to the not necessarily feasible integral online non-metric facility location algorithm $\overline{\cA}=\mathsf{Round}(\cA,\beta)$ as \emph{rounding} $\cA$ with \emph{boosting parameter} $\ln\nicefrac1\beta$.

\subsection{$p$-Sample Model}
\label{sec:nmflsample}

As a warm up, we begin with the no-corruptions case. This algorithm assumes access to a random-order online non-metric facility location algorithm $\ronmfl$, with competitive ratio $\Delta$, and a fractional online non-metric facility location algorithm $\fnmfl$, with competitive ratio $\alpha$.

The algorithm has the following phases.
\begin{enumerate}[nosep]
    \item In the offline phase, run $\ronmfl$ on the randomly permuted sample $\samples$. Let $\SOL^\ronmfl_\samples$ be the solution state obtained. Let $\SOL^{\rfnmfl}_0=\SOL^\bk_0=\emptyset$.

    \item In the online phase, let $v$ be the client arriving at step $t$.
    \begin{enumerate}[nosep]
        \item If $v\in C(\SOL^\ronmfl_\samples)$, then set $\SOL^{\rfnmfl}_t=\SOL^{\rfnmfl}_{t-1}$ and $\SOL^\bk_t=\SOL^\bk_{t-1}$, and skip.

        \item Apply $\fnmfl$ and round the fractional solution online with boosting parameter $\ln\nicefrac1p$, using \Cref{lem:nmfl-rounding}, to obtain a solution state $\SOL^{\rfnmfl}_t$.

        \item If $v\notin C(\SOL^\ronmfl_\samples\cup \SOL^{\rfnmfl}_t\cup \SOL^\bk_{t-1})$, then set $\SOL^\bk_t=\SOL^\bk_{t-1}\cup \backargs{v}{\SOL^\ronmfl_\samples\cup \SOL^{\rfnmfl}_t\cup \SOL^\bk_{t-1}}$. Otherwise, set $\SOL^\bk_t=\SOL^\bk_{t-1}$.
    \end{enumerate}
\end{enumerate}
The solution at time $t$ is $\SOL^\ronmfl_\samples\cup \SOL^{\rfnmfl}_t\cup \SOL^\bk_t$. Let $\SOL^{\rfnmfl}$ and $\SOL^\bk$ denote the final values of $\SOL^{\rfnmfl}_t$ and $\SOL^\bk_t$, respectively. The final solution is $\SOL^\ronmfl_\samples\cup \SOL^{\rfnmfl}\cup \SOL^\bk$.

\subsubsection{Analysis}

To analyze the algorithm, we first prove the analog of \Cref{lem:backup_costs}. Let $\unc^\ronmfl=\{v\in X\setminus\samples\mid v\notin C(\SOL^\ronmfl_\samples)\}$ be the clients not connected by the random-order solution on the sample.

\begin{lemma}
\label{lem:nmfl-backup-costs}
It holds that
$\expect*{\sum_{v\in \unc^\ronmfl}\cost(\backargs{v}{\SOL^\ronmfl_\samples})} \leq O\left(\frac{\Delta}{p}\right)\cdot \OPT(X)$.
\end{lemma}

\begin{proof}
Let $\pi=(\pi_1,\ldots,\pi_n)$ be a uniformly chosen random permutation of $X$, and let $\samples$ denote the first $s=pn$ clients according to $\pi$. Since the clients in $\samples$ are the input to $\ronmfl$, let $\SOL^\ronmfl_i$ be its output on the prefix $(\pi_1,\ldots,\pi_i)$. Thus, $\SOL^\ronmfl_\samples=\SOL^\ronmfl_s$.

For simplicity, let $b_{\SOL}(v)=\cost(\backargs{v}{\SOL})$. We need to upper bound 
\[\expect*{\sum_{v\in\unc^\ronmfl}\cost(\backargs{v}{\SOL^\ronmfl_\samples})}=\expect*{\sum_{i=s+1}^n b_{\SOL^\ronmfl_s}(\pi_i)}.
\]By monotonicity, if $j\leq s$, then $b_{\SOL^\ronmfl_s}(\pi_i)\leq b_{\SOL^\ronmfl_{j-1}}(\pi_i)$. Averaging this inequality over $j=1,\ldots,s$, taking expectations, and summing over $i=s+1,\ldots,n$, we obtain
\begin{equation}
\label{eq:nmfl-psample-average}
\expect*{\sum_{i=s+1}^n b_{\SOL^\ronmfl_s}(\pi_i)}
\leq
\sum_{i=s+1}^n \frac{1}{s}\sum_{j=1}^s \expect*{b_{\SOL^\ronmfl_{j-1}}(\pi_i)}.
\end{equation}
We claim that $\expect*{b_{\SOL^\ronmfl_{j-1}}(\pi_i)}=\expect*{b_{\SOL^\ronmfl_{j-1}}(\pi_j)}$. Indeed, conditioned on $(\pi_1,\ldots,\pi_{j-1})$ and the randomness of $\ronmfl$, the solution state $\SOL^\ronmfl_{j-1}$ is fixed, and $\pi_i$ and $\pi_j$ are both uniformly distributed over $X\setminus\{\pi_1,\ldots,\pi_{j-1}\}$. Also, the cost incurred by $\ronmfl$ in round $j$ is at least $b_{\SOL^\ronmfl_{j-1}}(\pi_j)$, since $\ronmfl$ must serve $\pi_j$. Applying these observations to \eqref{eq:nmfl-psample-average}, we obtain
\begin{align*}
\expect*{\sum_{i=s+1}^n b_{\SOL^\ronmfl_s}(\pi_i)}
&\leq
\frac{n-s}{s} \sum_{j=1}^s \expect*{b_{\SOL^\ronmfl_{j-1}}(\pi_j)} \leq
\frac{1}{p} \expect*{\cost(\SOL^\ronmfl_\samples)}
\leq
O\left(\frac{\Delta}{p}\right) \OPT(X).
\qedhere
\end{align*}
\end{proof}

\begin{theorem}
\label{thm:nmfl-psample}
If $\ronmfl$ is a random-order online non-metric facility location algorithm with competitive ratio $\Delta$ and $\fnmfl$ is a fractional online non-metric facility location algorithm with competitive ratio $\alpha$, then the algorithm above has expected competitive ratio $O(\Delta+\alpha\log\nicefrac1p)$.
\end{theorem}

\begin{proof}
We bound the expected cost of $\SOL^\ronmfl_\samples\cup \SOL^{\rfnmfl}\cup \SOL^\bk$. The expected cost of $\SOL^\ronmfl_\samples$ is at most $\Delta\cdot \expect*{\OPT(\samples)}\leq \Delta\cdot \OPT(X)$. By the guarantee of $\fnmfl$ and \Cref{lem:nmfl-rounding}, the expected cost of $\SOL^{\rfnmfl}$ is at most $O(\alpha\log\nicefrac1p)\cdot \OPT(X)$.

It remains to bound the backup cost. A client can contribute to $\SOL^\bk$ only if it lies in $\unc^\ronmfl$ and is missed by the rounded online solution. By \Cref{lem:nmfl-rounding}, each such client is missed by the rounded online solution with probability at most $p$. By monotonicity, the backup cost paid on such a client is at most $\cost(\backargs{v}{\SOL^\ronmfl_\samples})$. Therefore, by \Cref{lem:nmfl-backup-costs}, $\expect*{\cost(\SOL^\bk)}\leq p\cdot \expect*{\sum_{v\in\unc^\ronmfl}\cost(\backargs{v}{\SOL^\ronmfl_\samples})}\leq O(\Delta)\cdot \OPT(X)$. Summing the three expected costs completes the proof.
\end{proof}

We obtain an expected $O(\log\nicefrac1p\cdot \log m + \log n)$-competitive ratio by instantiating $\ronmfl$ to be the $O(\log mn)$-competitive LearnOrCoverNMFL algorithm \cite{GuptaKL24} and $\fnmfl$ to be the $O(\log m)$-competitive fractional online algorithm \cite{nonmetricfacility}.

\subsection{$(p,k)$-Adversarially Robust Model}
\label{sec:nmfl-corruptions}

We now turn to the model in which the adversary is allowed to corrupt $k$ of the $s=pn$ samples; let $\csamples$ be the resulting set. As in the set cover and CIP cases, the main idea is to use only a carefully chosen prefix of the corrupted sample.

The algorithm uses $\ronmfl$, $\fnmfl$, and the backup algorithm defined previously. The online phase proceeds in sub-phases indexed by pairs $(g,\widehat{k})$, where $g$ estimates $\OPTLP(X)$ and $\widehat{k}$ estimates $k$. The sub-phases form two nested loops: $g$ remains fixed in the outer loop, and $\widehat{k}$ increases monotonically within the inner loop. We define an \emph{epoch} to be the union of all sub-phases associated with a fixed value of $g$.

We initialize $\widehat{k}=3$, and set $g$ to be the fractional optimum for the clients in $X$ seen so far. The guess $g$ is recomputed every time the fractional optimum at least doubles. When this happens, we start a new sub-phase $(g,3)$ with the new value of $g$. Within a sub-phase, we run $\fnmfl$ and round its fractional solution online using rounding parameter $\ln(\widehat{k}/p)$, as in \Cref{lem:nmfl-rounding}. Let $\SOL^{\rfnmfl}_t$ denote the rounded online solution maintained up to time $t$, including all facilities opened and client connections made by the rounded online algorithm in all sub-phases so far.

We end the current sub-phase $(g,\widehat{k})$ if the cost incurred by the backup algorithm during this sub-phase exceeds $g\cdot \alpha\ln(\widehat{k}/p)$. When this happens, we replace $\widehat{k}$ by $\widehat{k}^2/p$ and start the next sub-phase with the same value of $g$. If the fractional optimum has doubled, we start a new epoch with the new value of $g$ and $\widehat{k}=3$.

The overall algorithm proceeds as follows.
\begin{enumerate}[nosep]
    \item In the offline phase, run $\ronmfl$ on the randomly permuted corrupted sample $\csamples$, without initially adding its solution to the maintained solution. Let $\SOL^\ronmfl_i$ be the solution state obtained after processing the first $i$ clients of $\csamples$. Let $\ell_0=0$ and $\SOL^\bk_0=\emptyset$.

    \item In the online phase, let $v$ be the client arriving at step $t$. Suppose we are in sub-phase $(g,\widehat{k})$.
    \begin{enumerate}[nosep]
        \item Apply $\fnmfl$ and round the fractional solution online with boosting parameter $\ln(\widehat{k}/p)$ to update $\SOL^{\rfnmfl}_t$.

        \item If $v\notin C(\SOL^{\rfnmfl}_t\cup \SOL^\bk_{t-1}\cup \SOL^\ronmfl_{\ell_{t-1}})$, then set $\SOL^\bk_t=\SOL^\bk_{t-1}\cup \backargs{v}{\SOL^{\rfnmfl}_t\cup \SOL^\bk_{t-1}\cup \SOL^\ronmfl_{\ell_{t-1}}}$. Otherwise, set $\SOL^\bk_t=\SOL^\bk_{t-1}$.

        \item Let $\ell_t=\max\{i\mid i\geq \ell_{t-1}\text{ and }\cost(\SOL^\ronmfl_i)\leq \cost(\SOL^{\rfnmfl}_t)+\cost(\SOL^\bk_t)\}$.
    \end{enumerate}
\end{enumerate}
The solution at time $t$ is $\SOL^{\rfnmfl}_t\cup \SOL^\bk_t\cup \SOL^\ronmfl_{\ell_t}$. Let $\SOL^{\rfnmfl}$, $\SOL^\bk$, and $\SOL^\ronmfl$ be the final values of $\SOL^{\rfnmfl}_t$, $\SOL^\bk_t$, and $\SOL^\ronmfl_{\ell_t}$, respectively. The final solution is $\SOL^{\rfnmfl}\cup \SOL^\bk\cup \SOL^\ronmfl$. By the invariant in the definition of $\ell_t$, we always have $\cost(\SOL^\ronmfl)\leq \cost(\SOL^{\rfnmfl})+\cost(\SOL^\bk)$.

\subsubsection{Analysis}

Let $\crpt$ be the index of the first corrupted sample in $\csamples$, ordered by the random permutation used by $\ronmfl$. The algorithm does not know $\crpt$. Let $\SOL^\ronmfl_{\clean}:=\SOL^\ronmfl_{\crpt-1}$ be the solution of $\ronmfl$ on the clean prefix, and let $\unc^\ronmfl=\{v\in X\setminus\samples\mid v\notin C(\SOL^\ronmfl_{\clean})\}$ be the clients not connected by $\ronmfl$ on the clean prefix.

For a prefix $P$ of $\csamples$, let $\event(P)$ denote the event ``$\csamples_{<|P|+1}=P$ and $\crpt>|P|+1$,'' i.e., the first $|P|$ positions are $P$ and the next position is uncorrupted. For a subset $P\subseteq\csamples$, define
\[
\skewed(P) := \left\{v \in X \setminus P 
\ \ \middle| \ \ \prob*{\csamples_{\abs{P}+1} = v  | \event(P)} 
< \frac{1}{2(n-\abs{P})} \right\},
\]
to be the set of clients whose chance of appearing next is too small compared with the uniform benchmark, given that the prefix $P$ seen so far is uncorrupted. Without corruptions, each of the remaining clients in $X\setminus P$ would be equally likely to appear next in random order.

Recall the following near-uniformity lemma, proven in \Cref{sec:setcover}.
\uniformity*

We first prove the robust analog of \Cref{lem:nmfl-backup-costs}.

\begin{lemma}\label{lem:nmflwc_unc}
It holds that
$\expect*{\sum_{v\in \unc^\ronmfl}\cost(\backargs{v}{\SOL^\ronmfl_{\clean}})}
\leq
O\left(\frac{k\Delta}{p}\right)\cdot \OPT(X)$.
\end{lemma}

\begin{proof}
Let $T=s/(4k)$, and assume $T$ is a nonzero integer for simplicity. For a solution state $\SOL$, let $b_{\SOL}(v)=\cost(\backargs{v}{\SOL})$. Let $B_i=\sum_{v\in X} b_{\SOL^\ronmfl_i}(v)$. By monotonicity, $B_i\leq B_{i-1}$. Since every client in $\unc^\ronmfl$ is not connected by $\SOL^\ronmfl_{\clean}=\SOL^\ronmfl_{\crpt-1}$, we have $\sum_{v\in\unc^\ronmfl}\cost(\backargs{v}{\SOL^\ronmfl_{\clean}})\leq B_{\crpt-1}$.

We claim that
\begin{eqnarray}
\label{eqn:nmfl-bkupcost}
\expect*{B_{\crpt-1}}
&\leq&
O\left(\frac{k}{p}\right)\cdot \OPT(X)
+
O\left(\frac{k}{s}\right)
\sum_{i=1}^T
\prob*{\crpt>i}\cdot
\expect*{B_{i-1}\mid \crpt>i}.
\end{eqnarray}
We omit this proof as it is identical to the proof of \eqref{eqn:bkupcost}, and follows by a case-analysis on where the first corruption appears.

We now analyze $B_{i-1}$ for $i\leq T$. Condition on $\event(P)$, where $|P|=i-1$. If a client is not in $\skewed(P)$, then by definition its conditional probability of appearing at position $i$ in $\csamples$ is at least $1/(2(n-|P|))\geq 1/(2n)$. Also, $b_{\SOL^\ronmfl_{i-1}}(v)=0$ for every $v\in P$, since $\ronmfl$ has already connected every client in $P$. Therefore,
\begin{equation}
\label{eq:nmfl-B-prefix}
B_{i-1} = \sum_{v\in X} b_{\SOL^\ronmfl_{i-1}}(v)
\leq
2n\cdot \expect*{b_{\SOL^\ronmfl_{i-1}}(\csamples_i)\mid \event(P)}
+
\abs{\skewed(P)}\cdot O(\OPT(X)).
\end{equation}
Taking expectations, summing over $i=1,\ldots,T$, and using \Cref{lem:uniformitybound}, we obtain
\begin{align*}
&\sum_{i=1}^T \prob*{\crpt>i}\cdot \expect*{B_{i-1}\mid \crpt>i} \\
&\leq
2n\expect*{\sum_{i=1}^T \ind\{\crpt>i\}\cdot b_{\SOL^\ronmfl_{i-1}}(\csamples_i)}
+
O(\OPT(X))\sum_{i=1}^T \prob*{\crpt>i}\cdot
\expect*{\abs{\skewed(\csamples_{<i})}\mid \crpt>i} \\
&\leq
2n\expect*{\cost(\SOL^\ronmfl_{\clean})}
+
T\cdot O\left(\frac{k}{p}\right)\cdot \OPT(X) \\
&\leq
O(n\Delta)\cdot \OPT(X).
\end{align*}
In the last inequality, we used the random-order guarantee of $\ronmfl$ on the clean prefix. Substituting this bound into \eqref{eqn:nmfl-bkupcost} completes the proof.
\end{proof}

We now bound the cost incurred by the rounded online solution in terms of the cost incurred by the backup algorithm.

\begin{lemma}\label{nmfl-onlinecost-guess-double}
Suppose $D$ is the expected cost incurred by the backup algorithm over the whole algorithm. Then, the expected cost of $\SOL^{\rfnmfl}$ is at most $\OPT(X)\cdot O(\alpha\log(k/p))+O(D)$.
\end{lemma}

\begin{proof}
Fix a guess $g$ for $\OPTLP(X)$, and consider the epoch corresponding to $g$. The cost incurred by the fractional solution produced by $\fnmfl$ in any sub-phase $(g,\widehat{k})$ within this epoch is at most $g\cdot O(\alpha)$. Indeed, the fractional optimum is at most $2g$ throughout this sub-phase, and $\fnmfl$ is $\alpha$-competitive. Thus, by \Cref{lem:nmfl-rounding}, the expected cost incurred by the rounded online solution in sub-phase $(g,\widehat{k})$ is $g\cdot O(\alpha\log(\widehat{k}/p))$.

If the epoch corresponding to $g$ has a single sub-phase, then $\widehat{k}=3$, and the expected cost incurred by the rounded online solution in this sub-phase is $g\cdot O(\alpha\log(k/p))$. Since $g$ grows geometrically across epochs and $\OPTLP(X)\leq \OPT(X)$, the expected cost incurred by the rounded online solution across all such epochs is at most $\OPT(X)\cdot O(\alpha\log(k/p))$.

Now suppose the epoch corresponding to $g$ has at least two sub-phases. Let $\widehat{k}_{\max}$ be the maximum value of $\widehat{k}$ in this epoch. For every non-final sub-phase $(g,\widehat{k})$, the cost incurred by the backup algorithm in that sub-phase exceeds $g\cdot \alpha\ln(\widehat{k}/p)$. For the final sub-phase, let $\widehat{k}'$ be the value of $\widehat{k}$ in the penultimate sub-phase. Since the penultimate sub-phase ended and $\widehat{k}_{\max}\leq (\widehat{k}')^2/p$, we have $\ln(\widehat{k}_{\max}/p)\leq 2\ln(\widehat{k}'/p)$. Hence, the cost of the final sub-phase is also bounded by a constant times the backup cost incurred in the penultimate sub-phase. Therefore, the expected cost incurred by the rounded online solution across all epochs with at least two sub-phases is $O(D)$.
\end{proof}

We classify a sub-phase $(g,\widehat{k})$ as early if $\widehat{k}<k$, and late otherwise.

\begin{lemma}\label{nmfl-earlybackup-guess-double}
The expected cost incurred by the backup algorithm during early sub-phases $(\widehat{k}<k)$ is at most $\OPT(X)\cdot O(\alpha\log(k/p))$.
\end{lemma}

\begin{proof}
Consider any early sub-phase $(g,\widehat{k})$. The backup cost in this sub-phase, excluding the last time step, is at most $g\cdot \alpha\ln(\widehat{k}/p)$ by the sub-phase-ending rule. The backup cost in a single time step is at most $O(g)$ within this sub-phase. Indeed, the fractional optimum on the clients seen so far within the sub-phase is at most $2g$, and the cost of serving a single client by the backup algorithm is at most a constant times this fractional optimum. Therefore, the backup cost in an early sub-phase $(g,\widehat{k})$ is at most $g\cdot O(\alpha\log(\widehat{k}/p))$. Since $g$ grows geometrically across epochs and $\ln(\widehat{k}/p)$ grows geometrically across early sub-phases within an epoch, the total expected cost incurred by the backup algorithm during early sub-phases is at most $\OPT(X)\cdot O(\alpha\log(k/p))$.
\end{proof}

\begin{lemma}\label{nmfl-latebackup-guess-double}
Suppose the initial solution state contains $\SOL^\ronmfl_{\clean}$. Then, the expected cost incurred by the backup algorithm during late sub-phases $(\widehat{k}\geq k)$ of any suffix of the online sequence $X$ is at most $\OPT(X)\cdot O(\Delta)$.
\end{lemma}

\begin{proof}
By assumption, the initial solution state contains $\SOL^\ronmfl_{\clean}$. Therefore, in this suffix, a client is passed to the backup algorithm in a late sub-phase only if it is in $\unc^\ronmfl$ and is also missed by the rounded online solution.

In any late sub-phase, the rounding parameter is at least $\ln(k/p)$. Hence, by \Cref{lem:nmfl-rounding}, the probability that an incoming client is not connected by the rounded online solution is at most $p/k$. By \Cref{lem:nmflwc_unc}, if every incoming client left unconnected by $\SOL^\ronmfl_{\clean}$ were passed to the backup algorithm, then the expected cost incurred would be $\OPT(X)\cdot O((k/p)\Delta)$. Multiplying by the additional probability $p/k$ that the rounded online solution also misses the client, the expected cost incurred by the backup algorithm in late sub-phases during the suffix is at most $\OPT(X)\cdot O(\Delta)$.
\end{proof}

\begin{theorem}\label{thm:nmfl-corruptions}
If $\ronmfl$ is a random-order online non-metric facility location algorithm with competitive ratio $\Delta$ and $\fnmfl$ is a fractional online non-metric facility location algorithm with competitive ratio $\alpha$, then the algorithm above has expected competitive ratio $O(\Delta+\alpha\log(k/p))$.
\end{theorem}

\begin{proof}
We first bound the expected cost incurred by the backup algorithm. Define the critical time step $t^*:=\min\{t\mid \cost(\SOL^{\rfnmfl}_t)+\cost(\SOL^\bk_t)\geq \cost(\SOL^\ronmfl_{\clean})\}$. If the set is vacuous, set $t^*=\infty$.

If $t^*=\infty$, then $\cost(\SOL^\bk)\leq \cost(\SOL^{\rfnmfl})+\cost(\SOL^\bk)<\cost(\SOL^\ronmfl_{\clean})$. Using the guarantee of $\ronmfl$ on the clean prefix, we get $\expect*{\cost(\SOL^\bk)}\leq O(\Delta)\cdot \OPT(X)$.

Otherwise, $t^*$ is finite. By definition of $t^*$, $\cost(\SOL^{\rfnmfl}_{t^*-1})+\cost(\SOL^\bk_{t^*-1})<\cost(\SOL^\ronmfl_{\clean})$. In particular, $\cost(\SOL^\bk_{t^*-1})<\cost(\SOL^\ronmfl_{\clean})$. The expected cost incurred by the backup algorithm in time step $t^*$ is at most $O(\OPT(X))$. Therefore, the expected backup cost up to and including time $t^*$ is at most $O(\Delta)\cdot \OPT(X)$.

After time $t^*$, the invariant defining $\ell_t$ ensures that $\ell_t\geq \crpt-1$, and hence the current solution state contains $\SOL^\ronmfl_{\clean}$. The expected backup cost during early sub-phases is at most $\OPT(X)\cdot O(\alpha\log(k/p))$ by \Cref{nmfl-earlybackup-guess-double}. The expected backup cost during late sub-phases after time $t^*$ is at most $\OPT(X)\cdot O(\Delta)$ by \Cref{nmfl-latebackup-guess-double}. Thus, $\expect*{\cost(\SOL^\bk)}\leq \OPT(X)\cdot O(\Delta+\alpha\log(k/p))$.

Next, by \Cref{nmfl-onlinecost-guess-double}, the expected cost incurred by the rounded online solution is at most $\OPT(X)\cdot O(\alpha\log(k/p))+O(\expect*{\cost(\SOL^\bk)})\leq \OPT(X)\cdot O(\Delta+\alpha\log(k/p))$.

Finally, by the invariant in the definition of $\ell_t$, the cost of the solution state added from the simulated $\ronmfl$ run is at most $\cost(\SOL^{\rfnmfl})+\cost(\SOL^\bk)$. Therefore, the total expected cost incurred by the algorithm is $\OPT(X)\cdot O(\Delta+\alpha\log(k/p))$.
\end{proof}

We obtain an expected $O(\log\nicefrac{k}{p}\cdot \log m + \log n)$-competitive ratio by instantiating $\ronmfl$ to be the $O(\log mn)$-competitive LearnOrCoverNMFL algorithm \cite{GuptaKL24} and $\fnmfl$ to be the $O(\log m)$-competitive fractional online algorithm \cite{nonmetricfacility}.

\section{Issue in analysis of \cite{ArgueFGS22} for metric facility location} \label{sec:mfl-error}In this subsection, we explain the issue in the analysis of \cite{ArgueFGS22}. Their algorithm is the same as ours (see \Cref{sec:metricpsample} for a description), but restricted to the unit cost case. The relevant portion of their analysis is as follows.
Let $F^*$ be an optimum solution, and fix an optimum facility $f^* \in F^*$. Let $C$ be the set of clients assigned to $f^*$ in the optimum solution. Let $j_1 < \ldots < j_r$ be the clients of $C$ ordered by increasing distance from $f^*$. 

For each sampled client $j_a$, let $i_a \in \widehat F$ \ be the facility serving $j_a$ in the solution for the offline phase. For an online client $j \in C$ that lies between two consecutive sampled clients $j_a$ and $j_{a+1}$, its connection in the online phase is at most $d(j,\widehat{F})$, because $\widehat{F}$ is already open. It then bounds this further as 
\begin{equation}\label{mfl-charging-scheme}
d(j,\widehat F)
\leq d(j,i_a)
\leq d(j,f^*) + d(f^*,j_a) + d(j_a,i_a)
\leq 2d(j,f^*) + d(j_a,i_a),
\end{equation}
where the final inequality uses that $j_a$ precedes $j$ in the ordering by distance from $f^*$. 

The issue is bounding the expected contribution from the second term in \eqref{mfl-charging-scheme} across all clients. The analysis in \cite{ArgueFGS22} argues that each sampled client $j_a$ is charged at most $O(1/p)$ times in expectation, as $N_a$, the number of clients in the interval $(j_a, j_{a+1})$ is $O(1/p)$ in expectation.  It then concludes that the expected contribution from the second term in \eqref{mfl-charging-scheme} across all clients in $X$ is at most the connection cost in the offline phase, because the offline phase objective scaled up client demands by $\nicefrac1p$ (recall \eqref{line:phase1-obj}).

However, the actual expected cost charged to $j_a$ is $\expect*{\sum_{a=1}^{r} N_a \cdot d(j_a,i_a)}$, and the quantities within the expectation are not independent, so this is not a priori bounded by $\sum_{a=1}^{r} \expect*{N_a} \expect*{d(j_a,i_a)}$.

\end{document}